\documentclass[conference]{IEEEtran}
\IEEEoverridecommandlockouts
\usepackage{cite}
\usepackage{amsmath,amssymb,amsfonts}
\usepackage{algorithmic}
\usepackage{graphicx}
\usepackage{textcomp}
\usepackage{xcolor}
\def\BibTeX{{\rm B\kern-.05em{\sc i\kern-.025em b}\kern-.08em
    T\kern-.1667em\lower.7ex\hbox{E}\kern-.125emX}}

\newcommand{\hide}[1]{}
\newtheorem{Axiom}{Axiom}

\usepackage[linesnumbered,ruled,vlined]{algorithm2e}
\usepackage{url}
\usepackage{verbatim}
\usepackage{graphicx}
\usepackage{multirow}
\usepackage{balance}
\usepackage{lscape}
\usepackage{float}
\usepackage{xcolor}
\usepackage{paralist}
\usepackage{booktabs}
\usepackage{xspace}
\usepackage{listings}
\usepackage{lipsum}
\usepackage{subcaption}

\usepackage{enumitem}
\usepackage[table]{xcolor} 
\usepackage[most]{tcolorbox}  

\definecolor{lightblue}{RGB}{227, 242, 253}

\usepackage{mdframed}

\newtheorem{innerexample}{Example}[section]

\newenvironment{example}
  {\begin{mdframed}[
      backgroundcolor=lightblue,
      linecolor=lightblue,
      roundcorner=5pt,
      leftmargin=0pt,
      rightmargin=0pt,
      innerleftmargin=0pt,
      innerrightmargin=0pt,
      innertopmargin=2pt,
      innerbottommargin=2pt,
      skipabove=5pt,
      skipbelow=0pt
    ]
    \begin{innerexample}\itshape}
  {\end{innerexample}\end{mdframed}}

\newtheorem{definition}{Definition}
\newcommand{\qed}{$\square$}
\newtheorem{proof}{Proof}
\newtheorem{thm}{Theorem}
\usepackage{xcolor}

\definecolor{ForestGreen}{rgb}{0.0, 0.66, 0.47}
\definecolor{RubineRed}{rgb}{1.0, 0.0, 0.31}
\usepackage{hyperref}

\DeclareMathOperator*{\argmax}{\arg\!\max}

\usepackage{dsfont}
\newcommand{\AMOU}{\Cup}

\newcommand{\sTable}{table\xspace}
\newcommand{\cTable}{Table\xspace}
\newcommand{\degreeOfDilution}{degree of dilution \xspace}
\newcommand{\dod}{$\Delta$\xspace}
\newcommand{\Qschema}{$Q(B_1, \dots, B_z)$\xspace}
\newcommand{\Tschema}{$T(A_1, \dots, A_n)$\xspace}
\newcommand{\Tprimeschema}{$T'(A'_1, \dots, A'_m)$\xspace}
\newcommand{\nullval}{\text{$\perp$}\xspace}
\newcommand{\estimatename}{\ensuremath{\mathsf{ANTs}}\xspace}
\newcommand{\DUST}{\ensuremath{\mathsf{DUST}}\xspace}
\newcommand{\Baseline}{\ensuremath{\mathsf{Baseline}}\xspace}

\newcommand{\StarmieGold}{\textcolor{black}{\ensuremath{\mathsf{StarmieGold}}}\xspace}
\newcommand{\StarmieOne}{\ensuremath{\mathsf{StarmieOne}}\xspace}
\newcommand{\blockone}{\ensuremath{\mathsf{ScaNN}}\xspace}
\newcommand{\blocktwo}{\ensuremath{\mathsf{kNN-Join}}\xspace}
\newcommand{\blockthree}{\ensuremath{\mathsf{DeepBlocker}}\xspace}
\newcommand{\ERNov}{\ensuremath{\mathsf{ER}}\xspace}
\newcommand{\GMC}{\ensuremath{\mathsf{GMC}}\xspace}
\newcommand{\Starmie}{\ensuremath{\mathsf{Starmie}}\xspace}

\newcommand{\name}{\ensuremath{\mathsf{NTS}}\xspace}
\newcommand{\semnov}{\ensuremath{\mathsf{SemNov}}\xspace}
\newcommand{\GMM}{\ensuremath{\mathsf{GMM}}\xspace}

\newcommand{\bigA}{\mathcal{A}}

\newcommand{\bigT}{\mathcal{T}}

\usepackage{xcolor}      
\usepackage[most]{tcolorbox}  
\newif\iftechreport
\newcommand{\ifnottechreport}[2]{\iftechreport #2\else #1\fi}

\begin{document}



\title{Novel Table Search [Technical Report]}

\author{\IEEEauthorblockN{1\textsuperscript{st} Besat  Kassaie}
\IEEEauthorblockA{
\textit{David R. Cheriton School of Computer Science}\\
\textit{University of Waterloo}\\
Canada\\
bkassaie@uwaterloo.ca}
\and
\IEEEauthorblockN{2\textsuperscript{nd} Ren\'ee J. Miller}
\IEEEauthorblockA{
\textit{David R. Cheriton School of Computer Science}\\
\textit{University of Waterloo}\\
Canada\\
rjmiller@uwaterloo.ca}
}
\techreporttrue
\maketitle

\begin{abstract}
Avoiding redundancy in query results has been extensively studied in relational databases and information retrieval, yet its implications for data lakes remain largely unexplored. We bridge this gap by investigating how to discover unionable tables that contribute new information for a given query table in large-scale data lakes.
We formally define Novel Table Search (\name) as the problem of finding tables that are novel with respect to a given query table and identify two desirable properties that any  scoring function for \name should satisfy.
We introduce a concrete scoring mechanism designed to maximize syntactic novelty, prove that it satisfies the proposed properties, and show that the associated optimization problem is NP-hard. To address this challenge, we develop an efficient approximation technique based on penalization, i.e., Attribute-Based Novel Table Search (\estimatename).
We propose three additional \name variants to achieve syntactic novelty and  introduce two evaluation metrics for syntactic novelty. 
Through extensive experiments, we demonstrate that \estimatename  outperforms other methods in capturing syntactic novelty across evaluation metrics and various benchmarks, while also achieving the lowest execution time. 
\end{abstract}


\section{Introduction}\label{sec:intro}
There is a well-established line of research focused on developing techniques for discovering and retrieving targeted information from vast collections of datasets, i.e., data lakes~\cite{DBLP:journals/pvldb/NargesianZPM18,santos23,       DBLP:conf/icde/SantosBMF22, DBLP:journals/pvldb/CasteloRSBCF21, DBLP:conf/icde/DongT0O21, DBLP:conf/icde/KoutrasSIPBFLBK21, DBLP:conf/sigmod/ZhangI20, DBLP:conf/www/BrickleyBN19, DBLP:conf/sigmod/ZhuDNM19, DBLP:journals/debu/MillerNZCPA18, DBLP:journals/pvldb/FanWLZM23,DBLP:journals/vldb/ChapmanSKKIKG20}. 
Some methods allow users to search for relevant datasets by expressing requirements through keywords, while others enable users to expand their query tables either vertically (by adding more attributes) or horizontally (by adding more tuples) to support downstream data analysis tasks. These methods use a variety of relevance metrics to evaluate how relevant a dataset or table is for a query and typically do not consider novelty or diversity of results. 
Diversity is typically understood as a coverage or breadth requirement while novelty means avoiding redundancy or duplication in the results~\cite{DBLP:conf/sigir/ClarkeKCVABM08}.  
This oversight may 
have negative impact on both data utilization and decision-making quality.  As an  example, consider a physician analyzing a dataset to study the side effects of a medication on a patient population. If the discovered tables from the data lake are relevant, but only contain patients with the same characteristics as the patients in a query table, the analysis may become skewed by neglecting other segments of the population.
This arises because relevance metrics are typically based on a variety of similarity metrics (including syntactic measures like set-overlap  and semantic measures like word-embedding or table-embedding similarity)~\cite{DBLP:journals/vldb/ChapmanSKKIKG20}.
In many applications, users want relevant answers that contain novel information and where different answers are diverse.  Consider a data market~\cite{DBLP:journals/isr/MehtaDJM21, DBLP:journals/corr/abs-2411-07267} where 
a data buyer wants unionable tables. Since each table has an acquisition cost, she will often prefer to spend money on tables that are unionable, but that contain the most new information and when buying multiple tables, she will want a diverse or varied set of tables, meaning they are novel with respect to the query table and each other.

Given a query table, the goal of unionable table search is to identify the most unionable tables from a data lake~\cite{DBLP:journals/pvldb/NargesianZPM18}. Two tables are considered unionable if they share  attributes with similar semantics indicating the union of the two tables is semantically meaningful~\cite{DBLP:journals/pvldb/NargesianZPM18}.   More recent work considers tables unionable if they share attributes with similar semantics and the relationship between attributes or table context is similar~\cite{santos23,DBLP:journals/pvldb/FanWLZM23}. A common approach is to compute similarity between attributes using syntactic features (e.g., observed values) and between semantic representations (e.g., embeddings)~\cite{DBLP:journals/pvldb/NargesianZPM18,DBLP:conf/icde/BogatuFP020}. Consequently, the most trivially unionable table for a query table would be the table itself. However, identifying novel unionable tables introduces a new twist: we seek tables that are sufficiently similar to support union operations, yet dissimilar enough to introduce novelty.  In this work, we formally define Novel
Table Search (\name) as the problem of finding unionable tables that are
novel with respect to a given query table.

We define \name as a reranking step (Step~2 in Figure~\ref{fig:novelTabSearch}) that operates on the output of a unionable-table search system. Given an initial set of candidate tables that are unionable with the query table, \name selects those that are most novel. We define a specific novelty measure for \name that is  based on syntactic novelty  that penalizes tables that contain duplicate (redundant) data (while still having novelty). We propose several techniques for solving \name, each employing a different strategy, and evaluate both their ability to capture syntactic novelty and their efficiency.

\estimatename is designed expressly to maximize novelty and minimize duplicates in tables; \GMC is a method originally developed for query-result diversification~\cite{DBLP:conf/icde/VieiraRBHSTT11} that we adapt to target syntactic novelty; \ERNov an entity-centric approach that maximizes entity-resolution distance with the goal of returning tables containing novel entities; and 
 \semnov an approach that maximizes table semantic distance.

\begin{figure}[ht]
    \centering
    \includegraphics[width=0.99\linewidth]{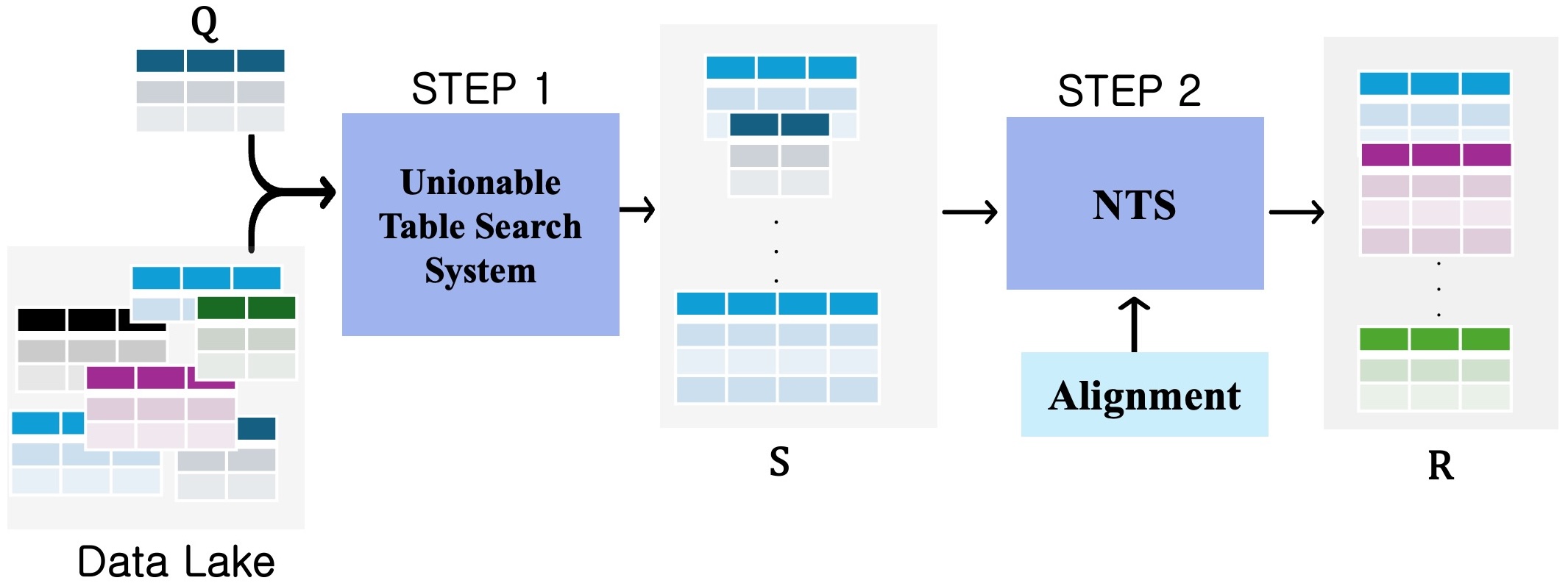}
    \vspace{-1.5 em}
    \caption{\small Given a query table $Q$, \name takes as input  $k$ unionable tables from Step~1 ($|S| = k$) and selects the $l$ most novel tables ($|R|\!=\!l$) using a novelty-aware scoring function. If the  search method lacks built-in attribute alignment, an external aligner is applied.}
    \label{fig:novelTabSearch}
\end{figure}

To illustrate, we walk through an example.

\begin{example}\label{exm:runing_exmple}
Figure~\ref{fig:mtvexample}(a-c) illustrates portions of three tables in the ugen-v2 dataset~\cite{DBLP:conf/vldb/PalKSM24}. Ugen-v2 is a recent table unionability benchmark. The user wishes to expand the query table $Q$ (Figure~\ref{fig:query1}), which contains information about artwork, by adding additional rows. She submits $Q$ to a unionable table search system and retrieves a ranked list of $k$ candidate tables. Among the benchmark candidates, the most unionable table has significant overlap with the query table. Our reranking component (\estimatename) then ranks $T_1$ (Figure~\ref{fig:u1}) as the most novel unionable table and $T_2$ (Figure~\ref{fig:u2_t2}) in seventh place.  \estimatename (Section~\ref{sec:ANTs}), identifies unionable tables by balancing high unionability with strong novelty. Notice that $T_1$ is more unionable with the query table (sharing five attributes) and novel (while the values in {\tt Medium} overlap with the query table, other values in {\tt Artwork} and {\tt Artist} are syntactically 
different.  Here $T_2$ is less unionable (only two unionable attributes) and has more {\tt Artist} overlap than $T_1$.
\end{example}

\begin{figure*}[t!]
\centering
\resizebox{.95\textwidth}{!}{%
 \begin{tabular}{c}
    \begin{subfigure}[t]{\textwidth}
      \centering
\begin{tabular}{llllll}
\textbf{ID} & \textbf{Artwork} & \textbf{Artist} & \textbf{Date Created} & \textbf{Medium} & \textbf{Style} \\
\hline
$t_1$ & The Mona Lisa & Leonardo da Vinci & 1503–1506 & Oil on poplar panel & High Renaissance \\

$t_2$ & The Hay Wain & John Constable & 1821 & Oil on canvas & Romanticism \\

$t_3$ & The Burial at Ornans & Gustave Courbet & 1849–1850 & Oil on canvas & Early Netherlandish \\
\end{tabular}
      \caption{\small Query table $Q$. Tuples are sampled from Art-History\_YZMEPGTH.csv}
      \label{fig:query1}
    \end{subfigure} \\

    \begin{subfigure}[t]{\textwidth}
      \centering
\begin{tabular}{lllllll}
\textbf{ID} & \textbf{Artwork} & \textbf{Artist} & \textbf{Date Created} & \textbf{Medium} & \textbf{Style} & \textbf{Condition} \\
\hline
$t_1$ & Water Lilies & Claude Monet & 1897–1926 & Oil on canvas & Nature & Good \\
$t_2$ & The Swing & Jean-Honoré Fragonard & 1767 & Oil on canvas & Figurative & Excellent \\
$t_3$ & The Fighting Temeraire & J.M.W. Turner & 1839 & Oil on canvas & Historical & Good \\
\end{tabular}
      \caption{\small Unionable table $T_1$. Tuples are sampled from Art-History\_UPFR2P3Y.csv}
      \label{fig:u1}
    \end{subfigure} \\ 

    \begin{subfigure}[t]{\textwidth}
      \centering
\begin{tabular}{lllll}
\textbf{ID} & \textbf{Artwork} & \textbf{Artist} & \textbf{Subject Matter} & \textbf{Cultural Context} \\
\hline
$t_1$ & Mona Lisa & Leonardo da Vinci & Portrait of Lisa Gherardini & General ... on later portraiture \\

$t_2$ & The Persistence of Memory & Salvador Dalí & Melting clocks & Permanent collection ... New York \\
$t_3$ & The Persistence of Memory & Salvador Dalí & Melting clocks & Museum of Modern Art, New York \\
\end{tabular}
      \caption{\small Unionable table $T_{2}$. Tuples are sampled from Art-History\_CW81XE6V.csv}
      \label{fig:u2_t2}
    \end{subfigure} \\ 

    \begin{subfigure}[t]{\textwidth}
      \centering
\begin{tabular}{llllll}
\textbf{ID} & \textbf{Artwork} & \textbf{Artist} & \textbf{Date Created} & \textbf{Medium} & \textbf{Style} \\
\hline
$t_1$ & The Mona Lisa & Leonardo da Vinci & 1503–1506 & Oil on poplar panel & High Renaissance \\

$t_2$ & The Hay Wain & John Constable & 1821 & Oil on canvas & Romanticism \\
$t_3$ & The Burial at Ornans & Gustave Courbet & 1849–1850 & Oil on canvas & Early Netherlandish \\
$t_4$ & Mona Lisa & Leonardo da Vinci & \nullval & \nullval & \nullval \\
$t_5$ & The Persistence of Memory & Salvador Dalí & \nullval & \nullval & \nullval \\
$t_{6}$ & The Persistence of Memory & Salvador Dalí & \nullval & \nullval & \nullval \\
\end{tabular}
      \caption{\small The result of left-outer-union between $Q$ and $T_2$, i.e.,  $T'=\AMOU(Q,T_{2})$.}
      \label{fig:amou}
    \end{subfigure} \\ 

    \begin{subfigure}[t]{\textwidth}
      \centering
\begin{tabular}{lllllll}
\textbf{ID} & \textbf{Artwork} & \textbf{Artist} & \textbf{Date Created} & \textbf{Medium} & \textbf{Style} & \textbf{Condition} \\
\hline
$t_1$ & Water Lilies & Claude Monet & 1897–1926 & Oil on canvas & Nature & Good \\
$t_2$ & The Swing & Jean-Honoré Fragonard & 1767 & Oil on canvas & Figurative & Excellent \\
$t_3$ & The Fighting Temeraire & J.M.W. Turner & 1839 & Oil on canvas & Historical & Good \\
$t_4$ & \textit{The Burial at Ornans} & \textit{Gustave Courbet} & \textit{1849–1850} & \textit{Oil on canvas} & \textit{Early Netherlandish} & \nullval \\
\end{tabular}
      \caption{\small Diluted unionable table $\tau(T_1, Q, \Delta)$ where $\Delta=0.2$. Tuple from $Q$ is \textit{italicized}.}
      \label{fig:u1Diluted}
    \end{subfigure} \\ 
\end{tabular}
       }
   \vspace{-0.6em}
\caption{\small Motivating example. Each tuple  has been augmented by  a unique identifier ID to distinguish it from others.} 
\label{fig:mtvexample}
\end{figure*}

This paper makes the following key contributions:

\noindent $\bullet$ We formally define the problem of identifying novel unionable tables, denoted \name, and specify two general properties that any novelty-aware scoring function $gscore$ should satisfy in the context of unionable table discovery.\\
\noindent $\bullet$ We propose a concrete syntactic scoring function, denoted $nscore$, that satisfies the above properties.  For this scoring function, we
prove that the associated optimization problem is NP-hard.\\
\noindent $\bullet$ We propose a new attribute-based approximation algorithm, \estimatename, to efficiently address \name.\\
\noindent $\bullet$  We develop three additional techniques.  Two methods are based on the syntactic novelty measure: \GMC, an adaptation of an existing method, and \ERNov, an entity-based technique, and \semnov, a new semantic novelty measure that we empirically evaluate to assess how well maximizing semantic novelty yields high syntactic novelty.\\
\noindent $\bullet$ 
We demonstrate that \estimatename also improves performance on a downstream machine learning task.

The remainder of the paper is organized as follows. Section~\ref{sec-preliminaires} introduces  notation.
Section~\ref{sec:overview} presents our scoring framework and outlines key properties essential for effective novelty scoring. In Section~\ref{sec:ANTs}, we describe our proposed approximation algorithm  \estimatename for efficiently identifying syntactically novel tables. Section~\ref{sec:baselines} introduces the baselines and the various methods developed for comparison. Section~\ref{sec:evaluationts} discusses the evaluation metrics developed to assess syntactic novelty and performance.  Section~\ref{sec:experiments} provides details on the experimental setup, datasets, and a comprehensive analysis of results.\footnote{Code and datasets are available at~\url{https://github.com/Besatkassaie/NTS/tree/master}}
Section~\ref{sec:ANTsMLTask_summary} presents the application of \estimatename to a downstream machine learning task.
Section~\ref{sec:related} reviews related work, and finally, Section~\ref{sec:conclusion} concludes the paper and outlines future research directions. \ifnottechreport{Detailed proofs and extended results are provided in Appendix of the technical report~\footnote{\url{https://github.com/Besatkassaie/NTS/blob/master/NTS_Technical_Report.pdf}}.}{Detailed proofs and extended results are provided in the Appendix.}

\section{Preliminaries}\label{sec-preliminaires}

We represent a data lake as a collection (set) of tables denoted as  $\bigT$.
We represent the schema of a \sTable $ T \! \in \! \bigT$  as \Tschema. With a slight abuse of notation, for a table \Tschema, we use  $T$  (without explicitly listing attributes) to refer to a \sTable instance  that conforms to the schema,  
$A_k \! \in \! T$ to refer to an attribute (column) of $T$ (equivalent to projecting the table $T$ on attribute $A_k$), and $X \!\subseteq\!$ \Tschema to refer to a subset of attributes in \Tschema. We denote the domain of attribute $A_i$ by $D_i$, that is, $T \! \subseteq\! D_1 \! \times\! \dots \!\times\! D_{n}$. We write $t \! \in \! T$ for a tuple in $T$  which is denoted as  $t= \langle v_1,\dots ,v_n \rangle$.  Throughout the paper, we assume that all data lake tables and query tables are non-empty.

We adopt a slight deviation from the traditional relational model by allowing tables to be multisets (bags) instead of strict sets since this is common in real data lakes.  Throughout this paper, unless explicitly stated otherwise, all set operations and relational operations over tables are interpreted in the context of multisets. For example, the projection operator $\pi$ is multiset projection and can therefore produce duplicate tuples.  Specifically, when referring to the union of two relations $ T$  and  $T'$, we assume it to be the multiset union (bag union), denoted as  $T \cup T'$. In a multiset, it is possible for two tuples to have identical values for all attributes. For notational convenience, we may assume that each tuple has a unique identifier to distinguish it from others.  

To perform the union on two unionable data lake tables (which may have different attribute names), we use alignments as defined by Nargesian et al.~\cite{DBLP:journals/pvldb/NargesianZPM18}.

\begin{definition}
    An  alignment $\mathcal{A}(T, T')$ is a one-to-one mapping from $X$ to $Y$ where $X\! \subseteq\!$ \Tschema and $Y \!\subseteq\!$ \Tprimeschema,  and $|X| \!= \!|Y|\! \geq \!1$.\footnote{When unambiguous, $\mathcal{A}(T, T’)$ is abbreviated as $\mathcal{A}$.}
\end{definition}

Tuples may have null (\nullval) values.  Since we assume that no additional information is available regarding the relations, such as details about their creation, schema, or possible constraints, we do not differentiate among the various interpretations of \nullval~\cite{DBLP:conf/edbt/GlavicMMPSV24}. Unlike SQL, we adopt a simple two-valued logic for comparisons with null. 
 That is, comparisons involving \nullval are treated as if  \nullval were a specific value, hence (  \nullval ==  \nullval is true and  \nullval == \textit{anything else} is false).  Furthermore, given  two tables \Tschema and \Tprimeschema and their corresponding  alignment $\mathcal{A} (T, T')$ we introduce a new ordered operator  denoted as $\AMOU(T, T', \mathcal{A})$ which is an asymmetric multiset  extension to the classic outer union introduced by Codd~\cite{DBLP:journals/tods/Codd79}.\footnote{$\mathcal{A}$ might be omitted for brevity; $T \AMOU T’$, $\AMOU(T, T’, \mathcal{A})$, and $\AMOU(T, T’)$ denote the same operation.}
This operation, which we call left-outer-union, computes the union of tuples from two relations, $T$  and  $T'$, that are only partially aligned, retaining all attributes of $T$ even if they have no aligned attribute in $T'$. Unaligned attributes from $T'$ are removed. \ifnottechreport{}{Let $X$ be the aligned set of attributes from T and $\mathcal{A}(X)$ the aligned attributes from $T'$.  Let $X'$ be the remaining (unaligned) attributes of T.  We project $T'$ on the aligned attributes $\mathcal{A}(X)$ and add to $T'$ additional attributes containing \nullval,  call this new relation $T''$.  Then $T \AMOU T'$ is the multiset union of $T$ and $T''$. Unlike SQL, even if all attributes align, we use multiset semantics and do not eliminate duplicates.}

\begin{example}
Consider the tables $Q$ and $T_2$ from Figure~\ref{fig:mtvexample}. Assuming attribute alignment is based on matching attribute names, the result of  $\AMOU(Q, T_2)$ is illustrated in Figure~\ref{fig:amou}. Notice all attributes from $Q$ are retained and notice this operator uses multiset semantics. For attributes of $Q$ that are not present in $T_2$, e.g., {\tt Style}, the corresponding entries for tuples originating from $T_2$ are padded with \nullval.  
\end{example}

\section{Novel Table Search}\label{sec:overview}
	
\subsection{\name Problem Definition}\label{subsec-problem}
Given a query \sTable  $Q$, the unionability search problem involves identifying the $k$ most unionable tables from a data lake~\cite{DBLP:journals/pvldb/NargesianZPM18,DBLP:journals/pvldb/FanWLZM23}. We adopt the definition of the unionability search problem from Fan et al.~\cite{DBLP:journals/pvldb/FanWLZM23}. That is, for two tables $T$ and  $T'$, a unionability scoring function $U(T, T')$ maps the input tables to a score in the range $[0,1]$. Note the unionability scoring function assumes an alignment has been computed between $T$ and $T'$.

\begin{definition}{\bf(\cTable Union Search)}\label{def-tus}
Given a collection of data lake tables $\bigT$, a query \sTable $Q$, and a unionability scoring function $U$ over pairs of tables,  \sTable union search determines a subset $S \subseteq \bigT$ where $|S| = k$ and $\forall T \in S$ and $T' \in \bigT-S $, we have $U(Q, T') \leq U(Q, T)$.
\end{definition}

 In contrast, Novel \cTable Search (\name)  aims to identify the  $l$  most novel and unionable tables from a data lake. In this work, we define \name as the second step (Step~2 in Figure~\ref{fig:novelTabSearch}) that operates on the output of a unionable-table search system. Specifically, we assume that, given a query \sTable  $Q$, a \sTable union search step has already identified a set  $S$  of unionable tables, where $|S| = k$. The second step then identifies an $R \subseteq S$  based on  novelty with respect to $Q$ where $|R|=l \leq k$. 
To do this, we assume a novelty scoring function that takes as input a query table $Q$ and a set of tables $R$ and outputs a score indicating the novelty of $R$ with respect to $Q$.
\begin{definition}{(\bf Novel \cTable Search)} \label{def-nts}
Given a query \sTable $Q$,   a set  of unionable tables $S$ where $|S| = k$,  and a novelty scoring function 
$\mathrm{gscore}$, Novel \cTable Search (\name) finds a subset $R \subseteq S$ where $|R| = l$ and $\forall S' \subseteq S$ where $|S'|=l$ and $R \neq S'$, we have $ \mathrm{gscore}(Q, S') \leq \mathrm{gscore}(Q, R)$.
\end{definition}

\subsection{Properties of Novelty Scoring Functions}\label{sec:generalproperties}
Given a query \sTable $Q$, we wish to  identify $l$ unionable tables (from a set $S$ of unionable tables) that exhibit maximum novelty  with $Q$. 
We  identified two  properties that a novelty scoring function should satisfy. 

\begin{Axiom} \label{axm:blDup}
(Blatant Duplicate Axiom)  
\vspace{-0.5em}
$$S'\!=\!S \cup \{Q\} \land Q \notin S \Rightarrow \mathrm{gscore}(S') \!<\! \mathrm{gscore}(S)$$
\end{Axiom}

\vspace{-0.5em}
This axiom states that if the query \sTable is included in the result set, the novelty score of the search result must be lower in comparison to a set without $Q$.

Given a query table  $Q$  and a unionable table  $T$, any tuple in  $T$  that already exists in  $Q$ should not contribute to the overall novelty of the search result, independent of the other tuples present in  $T$. 
To formally define this, we first define a \textit{diluted version} of $T$.
       
\begin{definition}
   Given a query $Q\in~$\Qschema, a unionable \sTable $T \in$~\Tschema, their corresponding alignment $\mathcal{A}$, and \degreeOfDilution $0<\! \Delta\! \leq 1$, a diluted version of $T$ with respect to $Q$ is denoted as $\tau(T,Q, \Delta)$ and contains $T$ outer-unioned with $\lceil \Delta \cdot |Q| \rceil$ tuples from $Q$.
   \begin{equation*}
       \tau(T,Q, \Delta)=\AMOU (T,Q',\mathcal{A} )
   \end{equation*}
   where   
   $Q' \subseteq Q$, $|Q'|\!=\!\lceil \Delta \cdot |Q| \rceil$.
  
\end{definition}

\begin{example}\label{exm:dilution}
Assume that Table $T_1$ in Figure~\ref{fig:u1} contains only the tuples shown, a possible dilution of $T_1$ with respect to $Q$, denoted $\tau(T_1, Q, \Delta)$ for $\Delta = 0.2$, is shown in Figure~\ref{fig:u1Diluted}.  
\end{example}

\begin{Axiom}(Dilution Axiom)
\vspace{-0.5em}
     $$S'\!\!=\!S \cup \{\tau(T_m, Q, \Delta)\} \land T_m \in S \! \Rightarrow\! \mathrm{gscore}(S')\! < \!\!\mathrm{gscore}(S) $$ 
 \end{Axiom}

Notice that both axioms are syntactic in that they state that the novelty score should penalize the inclusion of exact copies of tuples.  They do not say anything about semantically similar tuples. Next, we introduce a concrete syntactic scoring function, $nscore$, as a specific instance of $gscore$, and prove that it satisfies these axioms.

\subsection{Novelty Score}\label{subsec:NoveltyScore}
To build a novelty score over a table, we  use a tuple novelty score. Given a pair of tuples $t=\langle v_1, \cdots v_n\rangle$ and $t'=\langle v'_1, \cdots, v'_n\rangle$  we expect the tuple novelty score to satisfy the following conditions: \textbf{C1)} The tuple novelty score is  minimum  when two tuples are identical; \textbf{C2)} The tuple novelty score is maximum when both tuples have all non-null values and no corresponding attribute values are the same.  With this intuition, we define a tuple novelty score.


\begin{definition} \label{def:tup_pair_novelty} {\bf(Tuple Pair Novelty Score)}
Given two tuples with the same arity $n$, 
$t=\langle v_1, \cdots v_n\rangle$,  and $t'=\langle v'_1, \cdots, v'_n\rangle$, their tuple pair novelty score is defined as:
\vspace{-0.5em}
\begin{equation} \label{eq:pair_nscore}
\mathrm{tuple\_nscore}(t,t')=\frac{1}{n} \sum^{n}_{i=1} \textit{novelty}(v_i, v'_i) \end{equation} 
\vspace{-0.5em}
s.t. 
\vspace{-1.5em}
\begin{equation*}
 \textit{novelty}(v_i, v'_i)=\begin{cases}
1 &  v_i \neq v'_i \land  v_i, v'_i \neq \nullval, \\
\beta_{i} &   v_i \neq v'_i  \land (v_i = \nullval \lor v'_i = \nullval), \\
0& \textit{else.} \\
\end{cases}
\end{equation*}
where  $0 \!\leq \!\beta_{i} \! \leq  \!1$. Note that $0 \! \leq  \!\mathrm{tuple\_nscore}(t,t') \leq 1$. 
\end{definition} 

We assign a novelty  of $1$ to every pair of distinct constants.  
When one attribute value is null, we  assign a different reward $\beta_{{i}}$ with the intuition that a null value could be equal to the given constant.  
The likelihood that this would happen depends on the number of repetitions in the domain of attribute $i$.  
Let $P_i$ be the probability that two random tuples share the same (non-null) value on attribute $i$, then we set $\beta_i\!=\!1\!-\! P_i$.  
Finally, when both values are null, there is no novelty.



\begin{definition} \label{def:novelty_support}
    {\bf({Tuple Novelty})} 
    We define the novelty of a tuple $t$ in table $T$ ({\em Tuple Novelty})  as the minimum $\mathrm{tuple\_nscore}$: 
    \vspace{-0.8em}
$$N(t)=min \{ \mathrm{tuple\_nscore}(t, t'): t' \in T \wedge t \neq t' \}$$
\end{definition}

\begin{example}
For the table  in Figure~\ref{fig:u1Diluted},  the tuple novelty of $t_3$ is the minimum of the tuple pair novelty of 
$\{(t_3,t_1), (t_3, t_2), (t_3, t_4)\}$.  
The attribute \texttt{Condition} contains one null value and two equal values, so its corresponding $\beta$ is computed from the data as
$1\!-\!\frac{1}{3}\!=\!\frac{2}{3}$.  For \texttt{Medium}, there is a single value so $P_{Medium}\!=\! 1$, hence its $\beta$ is $0$ meaning nulls in this attribute would not contributing to novelty.  All other attributes in this example have only unique values meaning $\beta\!=\!1$.  However, since the only null value is on \texttt{Condition}, we only need $\beta_{Condition}$ for this simple example.   Hence,
\vspace{-0.7em}  
$$ N(t_3)=\min\{\frac{4}{6}, \; \frac{5}{6}, \; \frac{4+ 1\times\frac{2}{3}}{6}\}=\frac{4}{6}$$

The intuition is that $t_3$ provides the least novelty when compared to $t_1$ (since they have the most overlap), the most when compared to $t_2$.  And the parameter $\beta$ models  that the novelty of the null is between that of an identical value ($t_1$) and a distinct value ($t_2$).
\end{example}

 \begin{definition}\label{def:table_nov_support}
{\bf ({\cTable Novelty Score})} Given a \sTable $T$ with $y$ tuples $|T|=y$ the novelty score of $T$  is defined as: \vspace{-0.5em}
\begin{equation} \label{eq:table_nov_support}
\mathrm{table\_nscore}(T)=\frac{1}{y} \sum^{}_{t \in T } N(t)
\end{equation}
\vspace{-0.5em}
\end{definition} 
Therefore, the novelty score of a \sTable is defined as   the sum of novelty  of all  tuples of the \sTable divided by the total number of tuples.  We carefully account for both the number of tuples  and the degree of novelty each introduces. 
For a high novelty, we expect each tuple to contribute meaningful added value in terms of novelty, i.e., $N(t)\!\!>\!\!0$. In a data-market setting where each table (and therefore tuple) has an acquisition cost, the novelty score measures novelty per unit cost, penalizing tables with many redundant tuples (those with $N(t)=0$).

To compute the novelty score of a query table $Q$ with respect to $l$  unionable tables, we take the left outer-union of the query table with the $l$ unionable tables and compute the table novelty of the result.

\begin{definition}  \label{df:nscore_resultset}
{\bf (Search Novelty Score)} Given  a query \sTable $Q$ and a set of $l$ unionable tables $R=\{T_1, \cdots T_l\}$, and a set of alignments from $Q$ to each of the $T_i$ denoted as $\mathbb{A}=\{\mathcal{A}_1, \cdots, \mathcal{A}_l\}$; 
the novelty score  of the search, $\mathrm{nscore}(Q,R, \mathcal{A})$, is defined:
\vspace{-0.5em}
$$\mathrm{nscore}(Q, R, \mathbb{A})= \mathrm{table\_nscore}(\Breve{T})$$
where $\Breve{T}$ is:
\vspace{-1em}
\begin{equation} \label{eq:resTable}
    \Breve{T}= Q \cup \bigcup \limits_{T_i \in R}^{} \AMOU(Q_{\emptyset},T_i, \mathcal{A}_i)
\end{equation}
where $Q_{\emptyset}=Q \setminus Q$ represents an empty table that retains the same schema as $Q$.
\hfill \qed
\end{definition}

By incorporating $Q$ into the computation of the search novelty score ($\mathrm{nscore}$), we reflect the novelty over the query table, in the overall novelty score of the results. Next, we show that the search novelty score function satisfies the axioms introduced earlier.

\begin{thm}\label{clm:bdupl}\label{clm:dil}
The Search Novelty Score  (Definition~\ref{df:nscore_resultset}) satisfies the blatant duplicate axiom and the dilution axiom. \hfill \qed
\end{thm}

We define the Novel Table Search (\name) problem as selecting  $l$  tables from  the $k$  unionable tables such that the search novelty score of the search result is maximized. 
\begin{definition} \label{def:objective_main}
{\bf (Novel \cTable Search)} Given a query \sTable $Q$, a  set of $k$ unionable tables $S=\{T_1, \cdots T_k\}$, and a set of alignments from $Q$ to $S$ denoted as $\mathbb{A}=\{\mathcal{A}_1, \cdots, \mathcal{A}_k\}$, \name finds a set of  $l$ tables $R$  where: 
 \vspace{-0.5em}
\begin{equation}
R=\argmax_{S' \subseteq  S, |S'|=l} \; \mathrm{nscore}(Q,S',  \mathbb{A})
\end{equation}
\end{definition}

To find the exact solution for \name, we must evaluate all subsets of size  $l$  from  $S$. For each subset, we compute its novelty score, leading to an exponential number of evaluations. As a result, finding the exact solution is  computationally infeasible, especially for large  $k$  and  $l$. 
\begin{thm}\label{clm:nphard} 
   Finding an optimal solution to the Novel Table Search problem
is NP-Hard.   \hfill \qed
\end{thm} 

To address this complexity, we propose an approximate algorithm in the next section.

\section{Attribute-Based Novel \cTable Search}\label{sec:ANTs}

To approximate the $\mathrm{nscore}$, we use an attribute-based approach that estimates the novelty of attributes rather that iterating over tuples individually.
Our solution (\estimatename) estimates the syntactic dissimilarity (novelty) scores between  aligned attributes.  
\estimatename  returns an ordered list of tables, and hence can be used to rank (or rerank) a list of tables produced by a table union search system.
In this computation, we maximize syntactic dissimilarity of attributes while still prioritizing unionable (semantically similar) attributes.
The building blocks of our approach are two functions that, for each pair of aligned attributes, compute their syntactic and semantic similarities. Two attributes are semantically similar if their domains contain the same or closely related concepts. Since we want to prioritize more unionable attributes, we follow the unionability literature and represent each attribute by an embedding in a semantic vector space and use the distance between their embeddings as a measure of semantic similarity.

\subsection{Syntactic Similarity of Attributes}\label{sec:syn}

Abstractly, we need a function $\mathrm{syn\_sim}(B_i, A_j) \in [0,1]$  that represents the {\bf syntactic similarity} of two attributes $\langle B_i, A_j \rangle$.
Two attributes are syntactically similar if their domains contain many identical
values. We characterize each attribute by its observed value domain and use the overlap or distribution of these values as our measure of syntactic similarity. Note that it is common to use set-similarity measures (such as Jaccard) to measure syntactic similarity.  This works well for large domains.  However, for small domains, simple set overlap may not be enough.  Consider, for example, two attributes containing days of the week (at most seven values).  One may represent online sales and have a distribution over all seven days.  Another, representing band-tour dates may contain all seven days but be strongly biased to include mostly weekends.  The Jaccard between these attributes is one, but the distributions are quite different.  So for small domains, we instead exploit the  distribution of  values and use discrepancies between them as a signal of syntactic novelty.

We define a hyperparameter called $s$, to distinguish whether an attribute belongs to a small or large domain. This value is determined empirically based on the characteristics of a  data lake, as it  influences the effectiveness of the ranking algorithm. Section~\ref{subsec:alignment_intialresult_impact} discusses our study of this parameter.

\noindent{\bf Syntactic Similarity for Large Domains}
Let $\mathcal{V}(B_i)$ and $\mathcal{V}(A_j)$, be the domains of attributes $B_i$ and $A_j$, respectively. If the size of the union of the domains is larger than $s$, we use the Jaccard similarity~\cite{DBLP:journals/fcsc/YuLDF16} to compute the attribute syntactic similarity.

\noindent{\bf Syntactic Similarity for Small Domain}
For small-domain attributes, we compute 
the distribution of observed (non-$\nullval$) values for each attribute and use the distance between these distributions as an indication of novelty.  We denote the  discrete probability distribution  associated with attributes $B_i$ and $A_j$  by $\mathbb{B}_i$ and $\mathbb{A}_j$, respectively. These distributions are computed over the union of the domains, and so each of them may have some missing values. We assign probability $0$ for missing values.  

We then require a measure of difference, or a divergence metric, between two probability distributions that satisfies the following properties: i) The metric has a bounded range; ii) It remains defined in the presence of missing values; iii) It is symmetric, as no reference distribution exists and both attributes ($B_i$ and $A_j$) are treated equally.
The \textit{Jensen-Shannon} divergence measure (JSD) has these desirable properties which align with our problem~\cite{DBLP:journals/tit/Lin91, DBLP:conf/kdd/DhillonMK02, DBLP:journals/siamrev/Kieffer94}.\footnote{\ifnottechreport{More information on Jensen-Shannon can be found in our technical report including its formal definition.}{See Appendix~\ref{apx:JSD_desc} for additional details on the Jensen–Shannon divergence.}}
\begin{example}\label{exmp:domainsize}
Assume that the tables in Figure~\ref{fig:mtvexample}(a–b) contain only the tuples shown.
Consider the query table $Q$ in Figure~\ref{fig:mtvexample} and the unionable tables $T_1$ (Figure~\ref{fig:u1}) and $T_2$ (Figure~\ref{fig:u2_t2}), where in both cases attribute \texttt{Artist} is aligned with $Q.\texttt{Artist}$. Intuitively, the JSD between $Q.\texttt{Artist}$ and $T_1.\texttt{Artist}$ ($\mathrm{JSD}\!=\!1$)  is larger than between $Q.\texttt{Artist}$ and $T_2.\texttt{Artist}$ ($\mathrm{JSD}\!=\!0.82$) because in the former pair, the empirical distributions place all their probability mass on \emph{different} artists, whereas in the latter pair both distributions assign the same mass to \texttt{Leonardo da Vinci}. This shared mass makes the distributions closer and yields a smaller JSD.
\end{example}

We define $\mathrm{syn\_sim}(B_i, A_j)$ for small domains using one minus the Jensen Shannon Divergence of their distributions.

\begin{definition}\label{def:syn} {\bf (Attribute Syntactic Similarity)}  For a pair of aligned attributes $\langle B_i, A_j \rangle \in \mathcal{A}(Q,T)$, let  
$\mathcal{V}(B_i)$ and $\mathcal{V}(A_j)$ be their domains and $ \mathcal{D} = |\mathcal{V}(B_i)  \cup \mathcal{V}(A_j) |$. Furthermore, let $\mathbb{B}_i$ and $\mathbb{A}_j$ denote their discrete probability distribution  over  $\mathcal{V}(B_i)  \cup \mathcal{V}(A_j)$. 
 \vspace{-0.5em}
\begin{equation} \label{eq:syn}
\mathrm{syn\_sim}(B_i, A_j)= \begin{cases}  \frac{|\mathcal{V}(B_i) \cap \mathcal{V}(A_j)|}{|\mathcal{V}(B_i) \cup \mathcal{V}(A_j)|} & {\cal D} > s \\
1-\mathrm{JSD}(\mathbb{B}_i,\mathbb{A}_j) & {\cal D} \leq s 
\end{cases}
\end{equation}
\end{definition}

\subsection{Semantic Similarity of Attributes}\label{sec:sem}

To model the semantic similarity between pair of  aligned attributes, we compute cosine similarity  on their semantic representation as proposed in Fan et al.~\cite{DBLP:journals/pvldb/FanWLZM23}, which has achieved state-of-the-art performance in unionable \sTable discovery. For an attribute $A$, let $\mathrm{Sem(}A)$ be the \Starmie learned attribute embedding.

\begin{definition}\label{def:sem} {\bf (Attribute Semantic Similarity)} For a pair of aligned attributes $\langle B_i, A_j \rangle \in \mathcal{A}(Q,T)$, $B_i \in Q$, $A_j \in T$, the semantic similarity is defined as: 
\vspace{-0.5em}
\begin{equation} \label{eq:sem}
\mathrm{sem\_sim}(B_i, A_j)= \mathrm{cosSim}(\mathrm{Sem}(B_i), \mathrm{Sem}(A_j))
\end{equation}
\end{definition}

\subsection{Novelty-Aware Unionability}
We now define the {\em attribute novelty} as a score that when maximized,  the syntactic similarity is minimized and the semantic similarity is maximized.

\begin{definition} \label{def:att_based_score}
\textbf{(Attribute Novelty)} For a  pair of aligned attributes $\langle B_i, A_j \rangle \in \mathcal{A}(Q,T)$,  $B_i \in Q$, $A_j \in T$, the {\em attribute novelty}, $\mathrm{AttNovelty}(B_i, A_j)$, is defined as:  
 \vspace{-0.5em}
\begin{equation}\label{eq:attr_level_score}
(1\!-\!\mathrm{syn\_sim}(B_i, A_j))^b \!\times\! \mathrm{sem\_sim}(B_i, A_j)
    \end{equation}

where $b \geq 0$ is a hyperparameter.   \hfill \qed
\end{definition}
 
Attribute novelty is designed to capture both the unionability and novelty of attribute pairs. The score is maximized for pairs with high unionability and high novelty, and minimized when both are low.  The hyperparameter $b$ allows us to adjust the relative importance of the two signals.  When $b\!=\!1$, we are giving equal weight to the syntactic and semantic signals.  Lower $b$ gives more weight to semantic similarity, while larger $b$ gives more weight to syntactic dissimilarity (novelty).

The novelty of a table is then the sum of the attribute novelty over attributes aligned with a query table.

\begin{definition}\label{def:tablenovelty}\textbf{(Table Novelty)}
 Given a query table $Q$ and a unionable table $T$, the {\em table} novelty function  is defined as an aggregation over the attribute novelty of their aligned attributes:
 \vspace{-0.5em}
\begin{equation}\label{eq:aggregatescores}
\mathrm{TableNovelty}(Q, T)= \!\!\!\sum^{}_{ \langle B_i, A_j \rangle \in \mathcal{A}(Q,T) } \!\!\!\!\!\mathrm{AttNovelty}(B_i, A_j) 
 \end{equation}
\end{definition}

Intuitively, we prefer highly unionable tables that are semantically  similar to the query table that also contain syntactically different (novel) values. Therefore, we \textit{penalize} the semantic similarity between aligned attributes based on their syntactic similarity.

\begin{example}
Let us continue Example~\ref{exmp:domainsize} with $b\!=\!1$ and $s\!=\!5$, and assume that all aligned attributes have semantic similarity $1$. $T_1$ has an aligned counterpart for every attribute of $Q$. Except for \texttt{Medium}, these aligned pairs have large and disjoint domains, so their $\mathrm{AttNovelty}$ is maximal $1$. For \texttt{Medium}, the two distributions overlap on \texttt{Oil on canvas}, yielding JSD of $0.44$. Hence
\vspace{-0.6em}  
\[
\mathrm{TableNovelty}(Q, T_1)\!=\!1\!+\!1\!+\!1\!+\!0.44\!+\!1\!=\!4.44
\]
By contrast, $T_2$ has only two attributes aligned with $Q$. $T_2.\texttt{Artwork}$ is syntactically novel and semantically similar to $Q.\texttt{Artwork}$ and its $\mathrm{AttNovelty}$ is maximal $1$. $T_2.\texttt{Artist}$ shares probability mass with $Q.\texttt{Artist}$, giving a JSD of $0.82$ and thus lower syntactic novelty:
\vspace{-0.6em}
\[
\mathrm{TableNovelty}(Q, T_2)\!=\!1\!+\!0.82\!=\!1.82
\] 
\end{example}

\subsection{The \estimatename Solution}\label{sec: solution}
In summary, \estimatename takes as input a query table $Q$, a set $S$ of $k$ unionable tables, attribute alignments between $Q$ and each table $T \in S$, and the number of tables to return, $l$. It exhaustively scans all unionable tables $T\!\in\!S$, computes their novelty scores using $\mathrm{TableNovelty}(Q,T)$, sorts the candidates in descending order of novelty, and returns the top-$l$ most novel unionable tables (Algorithm~\ref{alg-ANTs}).

\SetKwComment{tcp}{\textasteriskcentered\ }{}
\definecolor{darkblue}{RGB}{0,0,139} 

\newcommand\mycommfont[1]{\small\color{darkblue}#1}
\SetCommentSty{mycommfont}

\begin{algorithm}
\caption{\estimatename} \label{alg-ANTs}
\KwIn{$Q$: query table; $S$: set of $k$ unionable tables; $A$: attribute alignments; $l$: number of tables to return}
\KwOut{$R$: list of $l$ most novel unionable tables}

$\text{scores} \gets \{\}$ \tcp{table to score map}

\ForEach{$T \in S$}{
  $\mathcal{A} \gets A[T]$ \tcp{get alignment}
  
$\text{scores}[T] \gets \text{getTableNovelty}(Q, T, \mathcal{A})$ \tcp{Def.~\ref{def:tablenovelty}}
}
$\text{scores} \gets \text{sortByVal}(\text{scores})$\tcp{descending sort}
$R \gets$ \text{top}($\text{scores}$, $l$)  \tcp{rank by novelty}

\Return $R$
\end{algorithm}

\normalsize
\section{Baselines and Compared Techniques}\label{sec:baselines}
In our evaluation, we consider five techniques: one baseline, three alternative \name variants, and a max–min diversification heuristic that serves as an approximate $\mathrm{nscore}$ baseline. Below, we summarize the key characteristics of each technique.
\subsubsection*{\Starmie} \Starmie~\cite{DBLP:journals/pvldb/FanWLZM23} is  a state-of-the-art  integrated framework designed for dataset discovery in data lakes, that can be used for table union search. It employs a contrastive learning approach to train attribute encoders entirely in an unsupervised fashion. These attribute encoders capture the  contextual semantics within tables. \Starmie computes a unionability score for attributes by measuring the cosine similarity between their embedding vectors, and aggregates these to derive an overall unionability score for tables based on their aligned attributes. \Starmie employs its own built-in alignment mechanism, which is based on maximum bipartite graph matching.	
We use \Starmie as a baseline for ranking the top  k  unionable tables based on their unionability scores relative to a given query.
The details of the hyperparameter settings used for \Starmie are provided in Section~\ref{sec:experiments}. 

\subsubsection*{Greedy with Marginal Contribution (\GMC)}\label{subsec:GMC}

While our work is the first to incorporate diversity (novelty) into table search, diversity in query results has been studied in prior research~\cite{DBLP:conf/icde/VieiraRBHSTT11}.  Viera et al.~\cite{DBLP:conf/icde/VieiraRBHSTT11} compare a number of approaches for query result diversification of which they show that an approach called \GMC achieves the best trade-off between result quality (precision and optimality) and execution time. We are interested in investigating how this method would perform if adapted to solve \name. Given a query (which may be a document, an image, or in our case a table), Vieira et al.~\cite{DBLP:conf/icde/VieiraRBHSTT11}  aim to  retrieve a set of $l$ elements that are highly similar to the query while ensuring that the selected set is mutually diverse, thereby forming a \textit{k-similar diversification set}. 
Let $S’ \subseteq S$, where $S$ is the set of candidate results. Let $\delta_{\text{sim}}(q, S’)$ denote the similarity between the query $q$ and the elements in $S’$, and let $\delta_{\text{div}}(S’)$ measure the diversity among the elements in $S’$. Given a parameter $\lambda \in [0,1]$, Vieira et al.~\cite{DBLP:conf/icde/VieiraRBHSTT11} formulate the objective as a \textit{Max-Sum optimization} problem to identify $R$:
    \vspace{-0.5em}
\begin{equation} \label{eq:GMC}
    R = \arg \max \limits_{S' \subseteq S, l=|S'|} \mathcal{F}(q,S')
\end{equation}
\noindent s.t. 
    \vspace{-0.8em}
$$\mathcal{F}(q,S')\!=\!(l\!-\!1)\!\times\!(\!1\!-\!\lambda)\!\times\!\delta_{sim}(q,S')\!+\!2\!\times\! \lambda\! \times\! \delta_{div}(S')$$

Since finding the optimal set $R$ is NP-hard, Vieira et al.~\cite{DBLP:conf/icde/VieiraRBHSTT11}  proposed \GMC, a deterministic approximation algorithm for identifying $R$. \GMC was evaluated against all known methods for diversifying query results that rely solely on relevance and diversity scores, and was shown to consistently outperform others in achieving higher values of the objective function $\mathcal{F}(q,S')$. For this reason, we compare our approach only against \GMC. Next, we describe how we have adapted \GMC for our setting.  

Given a query table $Q$ and a set of unionable tables $S$,  we employ the \GMC algorithm to select the $l$ most unionable and novel   tables. To do so, we supply \GMC with two functions: \textbf{1)} a function to measure the similarity (i.e., unionability) between the query table $Q$ and a data lake table $T\!\in\! S$ denoted as $\delta_{sim}(Q, T)\!: \!Q\!\times\!S\!\rightarrow\!\mathbb{R}^+$ for which we use $\mathrm{sem\_sim}$ (Definition~\ref{def:sem});
\textbf{2)} a function to quantify  the syntactic novelty between two data lake tables $T, T'\! \in \!S$  which is denoted as $\delta_{div}(T, T')\!:\!S\! \times\! S \!\rightarrow\! \mathbb{R}^+$ and for which we use $1\!-\!\mathrm{syn\_sim}$ (Definition~\ref{def:syn}).
Following the recommendation of Vieira et al.~\cite{DBLP:conf/icde/VieiraRBHSTT11}, we set $\lambda\! =\! 0.7$ in all experiments.

\subsubsection*{\GMM}\label{baseline:gmm}
\normalsize
We also used an approximation called \GMM, a greedy algorithm with a theoretical approximation guarantee for \textit{Max-Min diversification}. 
The Max-Min Diversification problem  selects a subset of items such that the minimum pairwise distance among them is maximized~\cite{DBLP:journals/ior/RaviRT94}. Formally, given a universe $\mathcal{U}$ of $k$ elements and a subset size $l$, the goal is to find a subset $S$ of size $l$ that maximizes $div(S)$:
    \vspace{-1em}
$$
S \!= \arg\max_{\substack{S' \subseteq \mathcal{U} \\ |S'| = l}}div(S')
$$
\normalsize
where $div(S) \!\! =\!\! \min \limits_{s, s’ \in S, s  \neq s'} \delta(s, s’)$ measures the smallest pairwise distance among elements of $S$. Ravi et al.~\cite{DBLP:journals/ior/RaviRT94}
 prove that Max-Min Diversification objective can be achieved by  \GMM with $\frac{1}{2}$-approximation guarantees, i.e., the $div(S)$ achieved is provably at least 50\% of the optimal (best possible) value. Therefore, \GMM provides an approximate $\mathrm{nscore}$ baseline focused \textit{only} on syntactic novelty (diversity).
We use the variant of \GMM\ proposed by Moumoulidou et al.~\cite{DBLP:conf/icdt/Moumoulidou0M21}, which allows initialization with a given element.  We use our previously defined syntactic distance function, $\delta_{div}$, to measure the pairwise syntactic  distance between tables (see Section~\ref{sec:ANTs}).

\subsubsection*{Semantic Representation for Syntactic Novelty (\semnov)}\label{subsec:semNovBaseline} 

Next, as an additional baseline, we investigate whether using the distance between table embeddings as a semantic novelty signal can result in  syntactic novelty, and we compare its effectiveness to that of \estimatename. Table embeddings such as TABBIE~\cite{DBLP:conf/naacl/IidaTMI21} capture the underlying semantics of tables and the distance between them serves as a proxy for semantic differences, and consequently for semantic novelty between tables. We propose to penalize the semantic similarity between two attributes by the semantic similarity between their corresponding table embeddings.
Note that in this case we do not distinguish between different domain sizes. 
Therefore, for an arbitrary  pair of aligned attributes $\langle B_i, A_j \rangle \in A(Q,T)$ where $B_i \in Q$, $A_j \in T$, $\mathrm{AttrSemNovelty}(B_i, A_j)$   will be:
\vspace{-0.5em}
\begin{equation} \label{eq:ants_bothdomain}
\begin{aligned}
\mathrm{AttrSemNovelty}(B_i, A_j)
&= (1-\mathrm{table\_sim}(Q, T))^{b} \\
&\times\, \mathrm{sem\_sim}(B_i, A_j)
\end{aligned}
\end{equation}
\noindent where  $\mathrm{table\_sim}(Q, T)$   denotes the cosine similarity between the vectors representing  tables    $Q$  and  $T$  extracted from the TABBIE trained model and takes values in  $[0,1]$. Again table novelty here is just the sum of the semantic $\mathrm{AttrSemNovelty}$ scores.

\subsubsection*{Entity Resolution (ER) for  Novel Table Search}\label{subsec:EntityResultion} 
\estimatename adopts an attribute-based strategy to reduce the time complexity of the novel table search problem. Alternatively, one could consider a tuple-based approximation approach. We apply an entity resolution technique  to estimate the degree of overlap, that is, the number of matching entities, between a query table and candidate tables in the data lake, and then rank the tables such that those with less (approximate tuple) overlap receive a higher score. An entity resolution process comprises two main steps: blocking and matching. When a query table arrives, its tuples are first blocked with those of each unionable table, and matching is then performed only within the shared blocks. As our work is based on syntactic novelty, we adopt a syntactic matching approach for the entity resolution task~\cite{DBLP:journals/pvldb/KondaDCDABLPZNP16}.  For blocking, we experimented with several techniques from the literature such as \blockone~\cite{DBLP:conf/icml/GuoSLGSCK20}, \blocktwo~\cite{DBLP:journals/vldb/SilvaALPA13}, and \blockthree~\cite{DBLP:journals/pvldb/Thirumuruganathan21} and adopted \blockone which was the fastest method.  
\ifnottechreport{A detailed explanation of our choice of these blocking techniques  is provided in the technical report.}{A detailed explanation of our choice of these blocking techniques is provided in Appendix~\ref{apnx:entityResultion}.} We refer to this system as \ERNov in our experiments.

\section{Evaluation Metrics}\label{sec:evaluationts}
We will use the search novelty score ($\mathrm{nscore}$) (Definition~\ref{df:nscore_resultset})
to evaluate the quality of NTS solutions.  However,  $\mathrm{nscore}$ is computationally expensive to compute, so we will  report it only for small benchmarks and employ several computationally cheaper   metrics.
Recall that the search novelty score decreases when the query table is unioned with the result table (Equation~\eqref{eq:resTable}). Thus, intuitively a technique is less desirable if it ranks an exact copy of the query table (when such a copy exists in the data lake) among the top results. We capture this tendency using the \textit{Blatant-Duplicate} metric, which measures, for each query, whether an exact copy of the query table appears in the top-$l$ results. Similarly, including tuples of the query table in any unionable table among the top-$l$ results (dilution) is detrimental to the search novelty score. Heuristically, we  prefer methods that prioritize  undiluted unionable tables and we quantify this behavior using the \textit{Syntactic Novelty Measure} (SNM).

\subsection{Blatant Duplicates}\label{sec:bltndup} 
The term blatant duplicate in Axiom~\ref{axm:blDup} refers to a  unionable table that is a duplicate of the query table. Motivated by this axiom, we introduce a metric to evaluate whether a system  ranks such blatant duplicates of query tables among the top-l results, when they exist within the data lake. 
Given a query  $Q$  and a set of tables, $R$, returned by a novel table search method  where $|R|\!=\!l$,  we define the metric as follows.
    \vspace{-0.5em}
\begin{equation}
\textit{blatant-duplicate}(Q,R)\! =\!
\begin{cases}
    1,  & \text{if } Q \!\in \!R, \\
    0 , &\text{else}.
\end{cases}
\end{equation}
To use this metric in our experiments, we include the query itself in the set of  $k$  unionable tables considered for \name.

\subsection{Measuring Syntactic Novelty}\label{sec:datadilution} 
Building on the concept of \textit{binary preference} (bpref)~\cite{DBLP:conf/sigir/Sakai07}, an established evaluation metric in the Information Retrieval community, we propose two new metrics to assess a search algorithm’s effectiveness in identifying novel tables from a data lake. Binary preference is designed for evaluation settings with incomplete relevance assessments. The bpref metric penalizes a system when a judged non-relevant document is ranked above a judged relevant one, while being unaffected by the retrieval of unjudged documents.

Similarly, for a query $Q$ and a set of $k$ unionable tables  $S\!=\!\{T_1, \cdots, T_k\}$, with their corresponding alignments $\bigA$, and a dilution degree $\Delta$ we construct a new set $\mathcal{S}=\!S\! \cup \!S_{dil}$ where $S_{dil}\!=\!\{\tau(T_i) | T_i \! \in \!S\}$.
We then evaluate how effectively a search algorithm ranks each unionable table $T_i$ above its counterpart $\tau (T_i)$ when ranking the elements in $\mathcal{S}$ and selecting top $l$ tables using a measure we call Syntactic Novelty Measure (SNM).

\begin{definition}\label{dfn:snm} {\bf ({S}yntactic {N}ovelty {M}easure (SNM))} Let $S$ be a set of $k$ unionable tables for a query $Q$, and $S_{dil}\!=\!\{\tau(T_i) | T_i \! \in \!S\}$ a set of diluted tables with degree $\Delta$.  Suppose a novel search algorithm returns an ordered list $S' \subset S \cup S_{dil}$  of size $l$. 
    \vspace{-0.8em}
\begin{equation}
    \textit{SNM}\!=\!1\!-\! \frac{|O| \! +\! |Y| }{l}
\end{equation}
    \vspace{-1.7em}
  s.t.   
\begin{equation*}
\begin{aligned}
O &= \{T_i \mid (\tau(T_i) \in S' \land T_i \notin S')  \\
  &\qquad \vee (T_i = Q \land \tau(T_i) \notin S' \land T_i \in S')\}, \\
Y &= \{T_i \mid (\mathcal{I}(S', T_i) > \mathcal{I}(S', \tau(T_i)))  \\
  &\qquad \vee (T_i = Q \land \mathcal{I}(S', T_i) < \mathcal{I}(S', \tau(T_i)))\}.
\end{aligned}
\end{equation*}

\normalsize $\mathcal{I}(S', T)$ denotes the position of table $T$ in the ordered list $S’$. A lower index indicates a better rank for the table.   Intuitively, $O$ denotes the set of diluted tables that appear without their corresponding original counterparts in the top l results while $Y$ represents the number of table pairs in the top $l$ where the diluted version is ranked lower (i.e., considered superior) to its original counterpart. To avoid rewarding systems that return blatant duplicates, we impose a penalty for retrieving blatant duplicates in the formulation of $O$ and $Y$.  Note that $0 \leq \mathrm{SNM} \leq 1$.  The minimum SNM value of 0 occurs when, for a given query, only diluted versions of the unionable tables are returned.
\end{definition}
SNM applies to ranking-based \name methods. As our \GMC method returns a set rather than a ranked list, we also define a rank-oblivious variant.

\begin{definition}\label{dfn:ssnm} \textbf{(Simplified SNM (SSNM))}  The Simplified SNM is defined as:
\vspace{-1.5em}
\begin{equation}
   \textit{SSNM}=1- \frac{|O|}{l}
\end{equation}
\vspace{-0.8em}
\end{definition}
\vspace{-1em}

\subsection{Additional  Metrics}
Given that \GMC (Section~\ref{sec:baselines}) is computationally expensive, we aim to investigate to what extent  \estimatename can optimize the objective defined for \GMC. 	As an additional evaluation metric, we compare the systems based on their $\mathcal{F}(Q, S’)$ or simply $\mathcal{F}$ (Equation~\eqref{eq:GMC}).

 \section{Experiments}\label{sec:experiments}
We implemented \estimatename, \ERNov, \GMM, and \GMC using Python 3.8.5, while \semnov required Python 3.6 due to library dependencies. All experiments were conducted on a server with five AMD EPYC 9554 64-core processors and 512 GB RAM.  We use the publicly available \Starmie codebase.\footnote{https://github.com/megagonlabs/starmie}

\subsection{Datasets}\label{subsubsec:datasets}
We begin with a detailed description of the three datasets and outline their characteristics and how we adapt them for our experimental setup. We have two main considerations in preparing datasets: i) We tackle the novel table search problem as a step on top of a unionable-table search system (see Definition~\ref{def-nts}). To  eliminate the influence of the initial  search system’s performance, in each dataset and for each query table, we only consider  $k$  unionable tables (as identified in the ground truth). We then apply various  approaches  solely to these unionable tables to find top $l$ final tables.  We set $k\!=\!20$;  ii) A key component of our evaluation is to systematically include both diluted versions of data lake tables and the query table itself (as a blatant duplicate) in the pool of unionable candidates.  In our dataset preparation, we set the \degreeOfDilution (\dod) to $0.4$ across all datasets. Note that all tables in the data lake have a diluted version.  Also note that we use the ground truth alignment for all experiments (except our ablation study of the influence of alignment quality in Section~\ref{subsec:alignment_intialresult_impact}) and all the tables are non-empty.  
 \vspace{0.5em}
 
\noindent \textbf{\textit{{T}able {U}nion {S}earch}}. The TUS dataset was synthesized from 10 base tables obtained from UK and Canadian open data~\cite{DBLP:journals/pvldb/NargesianZPM18}. A sequence of projection and selection operations were applied to these base tables to derive unionable tables. We use the same TUS version as that used by Khatiwada. et al.~\cite{santos23}.\footnote{\url{https://zenodo.org/records/14151800}} In this dataset, queries share common unionable tables, that is, a single data lake table may be unionable with multiple distinct query tables. Since diluting the data lake tables is part of our dataset preparation, we use a subset of queries that do not share unionable data lake tables, to simplify the implementation process. 
The resulting TUS dataset contains 42 query tables and 1008 unionable data lake tables, including 504 diluted tables. 

 \vspace{0.5em}
\noindent \textbf{\textit{Santos.}} Leveraging a benchmark creation technique similar to that proposed in Nargesian et al.~\cite{DBLP:journals/pvldb/NargesianZPM18}, Santos Small was also constructed  using open data~\cite{santos23}. 	Similar to TUS, we retain only a set of queries that do not share unionable tables with others. Our repurposed Santos Small dataset consists of 35 queries and 948 data lake tables, including 474 diluted versions. We refer to this dataset simply as Santos.

 \vspace{0.5em}\noindent \textbf{\textit{Ugen-v2}}.
 The most recent unionable table benchmark was generated by LLMs, \textbf{\textit{Ugen-v2}}~\cite{DBLP:conf/vldb/PalKSM24}.  This benchmark has been shown to be more challenging for all unionable table search solutions compared to hand-curated benchmarks like TUS and Santos. We created a smaller version of Ugen-v2, approximately 40\% of the  Ugen-v2 queries with the smallest number of tuples and 19 queries
 with two considerations. First, to eliminate any potential negative influence from automatic alignment,  we selected a subset of queries along with their corresponding unionable tables and manually aligned their attributes. Secondly, computing the search novelty score (Definition~\ref{df:nscore_resultset}) for our medium-sized datasets proved to be time-consuming. To enable computation of the nscore,
we created the smaller version of this challenging dataset. For brevity, we call our smaller version of this benchmark simply Ugen-v2.

\subsection{Experiment Setup}\label{subsec-setup}
We provide a detailed description of each system in our experimental comparison, including its configuration. 

\noindent \Starmie:
We use \Starmie~\cite{DBLP:journals/pvldb/FanWLZM23} for two purposes: (1) in offline mode, to generate contextualized attribute embeddings that are subsequently used by all methods; and (2) in online mode, as a baseline method for ranking unionable tables given a query table.
We represent attributes using embeddings generated by the \Starmie model, which are used to compute semantic similarity between attributes in  \estimatename, \semnov, and \GMC. \Starmie  was fine-tuned on the Santos 
dataset~\cite{DBLP:journals/pvldb/KhatiwadaSGM22}, and the resulting model was subsequently used to generate attribute embeddings with vector size $768$ for all datasets.\footnote{Starmie was fine-tuned on Santos arbitrarily, as it is expected to generalize across datasets without further tuning. Because all methods rely on the same fine-tuned model, the comparison is fair.} We configured the fine-tuning process with the following hyper-parameters:  $\mathrm{batch\;size}\!=\!64$, $\mathrm{learning\; rate}\! =\! 5e-5$, $\mathrm{max.\; seq.\; length}\! =\!128$,  $\mathrm{proj.\; dim.} \!=\! 768$,  $\mathrm{augment\_op}\! =\! \mathrm{drop\_col}$,  $\mathrm{sample\_meth}\! = \!\mathrm{tfidf\_entity}$, $\mathrm{table\_order} \!= \!\mathrm{column}$ and $\mathrm{num.\;of\;epochs}\! =\! 3$.  Refer to the official \Starmie codebase for hyperparameter details.
In the online mode, scalability is not a major concern when evaluating only the ground-truth unionable tables. Therefore, to rank the top-$l$ unionable tables, we employ a straightforward linear scan approach with a similarity threshold of $\tau \!\! =\!\!  0.7$ (see the Online Query Processing section in Fan et al.~\cite{DBLP:journals/pvldb/FanWLZM23}).

\noindent {TABBIE:}
We use the pre-trained TABBIE model~\cite{DBLP:conf/naacl/IidaTMI21} in inference mode from its public codebase.\footnote{\url{https://github.com/SFIG611/tabbie}} It represents each table with a $768$-dimensional embedding. Recall that TABBIE is used to create semantic embeddings for \semnov.

\noindent \ERNov: Experiments use the \blockone Python package (v1.4.2) with default settings and  $\mathrm{num\_neighbors}\!=\!5$\normalsize, determined empirically for fast blocking.\footnote{\url{https://pypi.org/project/scann/},} For entity matching, we use the rule-based technique shared in the publicly available \textit{py-entitymatching} \normalsize(version~0.4.2) codebase.\footnote{\url{https://pypi.org/project/py-entitymatching/}}  We compare  values of aligned attributes.  Values are considered similar if their Levenshtein similarity is $\geq 0.85$.

\subsection{Results for Effectiveness and Scalability}\label{subsec-eff}
\estimatename returns a ranked list of novel unionable tables. To evaluate ranking quality, we require at least two results and vary $l$ within $[2,10]$. 	Note that in the Ugen-v2 dataset, 79\% (15/19) of queries have fewer than 12 unionable tables, with two of them being blatant duplicates of  queries and their diluted versions and we account for this in the results.

\noindent \textbf{Blatant Duplicate Measure.} 
Introducing novelty through penalization, as in \estimatename and \semnov, yields better results (Table~\ref{tab:baltant_dup_avg}) on the Santos dataset, while \estimatename, \semnov, and \ERNov perform best on TUS (Table~\ref{tab:baltant_dup_avg}). On the Ugen-v2 dataset, \ERNov, which relies on syntactic similarity for matching, achieves the lowest duplicate rate. 

\begin{table}[!htbp]
\caption{\small  The average percentage of queries per dataset yielding a blatant duplicate within the top $l$ results ($l\!\in\![2,10]$). Lower values indicate better performance.  In all tables, \textbf{Bold} numbers show the best performance per row; the second best numbers \underline{underlin}\underline{ed}.}
\label{tab:baltant_dup_avg}
 \vspace{-0.5em}
\centering
{\setlength{\tabcolsep}{7pt} 
\renewcommand{\arraystretch}{1.1} 
\begin{tabular}{@{}r|ccccc} 
  & \Starmie  & \GMC & \ERNov & \semnov & \estimatename \\
\hline
\textit{Santos}          & 99.4\%  & 81.9\% & \underline{8.6\%} & {\textbf{0\%}} & {\textbf{0\%}} \\
\textit{TUS}             & 100\%   & 99.5\% & \textbf{0\%} & {\textbf{0\%}} & \textbf{0\%} \\
\textit{Ugen-v2}   & 100\%   & 98.8\% & \textbf{43.6\%} & {52.6\%} & \underline{43.9\%} \\
\hline
\textit{Ugen12}  & 100\%  & 100\% & \textbf{0\%} & \textbf{0\%} &\textbf{0\%}  \\
\textit{Ugen $l\!=\!2$}   & 100\%  &89.5\%  &\textbf{0\%}  & \textbf{0\%} & \textbf{0\%}\\
\hline
\end{tabular}
}
\end{table}

As previously noted, the Ugen-v2  dataset contains a limited number of unionable tables for 79\% of the queries. To enable a more equitable comparison, we restrict the evaluation to queries with at least 12 unionable tables (\textit{Ugen12}).   Under this condition, 
\estimatename, \semnov, and \ERNov have identical behavior.
Since all queries have at least four unionable tables, we also report results for $l\!=\!2$ (\textit{Ugen $l\!=\!2$}). Under this setting, \estimatename, \semnov, and \ERNov again show identical behavior and  outperform the other methods. Considering these conditions and \ERNov’s performance on Santos, \estimatename and \semnov remain the top performers overall.
\normalsize


\noindent \textbf{Syntactic Novelty Measure}. 
The SNM metric evaluates a system’s ability to rank an original table above its diluted (less novel) versions and applies only to rank-based methods; hence, \GMC is excluded. Figures~\ref{fig:snm_santos}, ~\ref{fig:snm_tus}, and ~\ref{fig:snm_ugen} report the average SNM across all queries for each dataset.  As shown, \estimatename consistently achieves superior performance across all datasets. Its best performance is observed on the TUS dataset, where it achieves the highest scores compared to the other datasets. The second  best performing system is \ERNov across all dataset. Furthermore, unlike the other two datasets, we observe a drop in SNM  for \estimatename and \ERNov on the Ugen-v2 dataset at $l\!=\!3$. A closer examination reveals that the scarcity of unionable tables contributes to the observed behaviors. 


\begin{figure}[!htbp]
\centering

\begin{subfigure}[t]{0.485\columnwidth}
  \includegraphics[width=\linewidth]{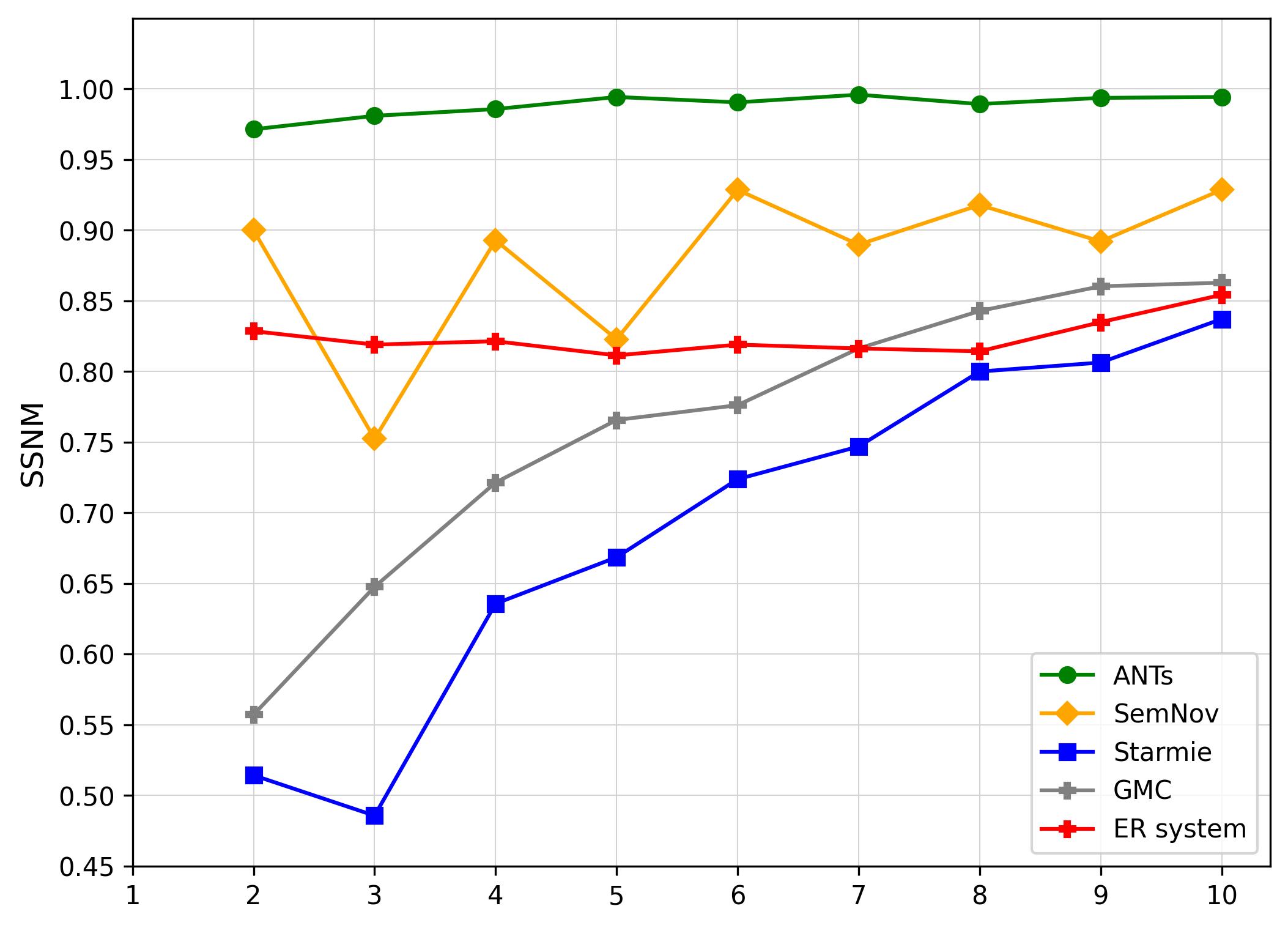}
  \caption{\small Santos}
  \label{fig:ssnm_santos}
\end{subfigure}
\hfill
\begin{subfigure}[t]{0.485\columnwidth}
  \includegraphics[width=\linewidth]{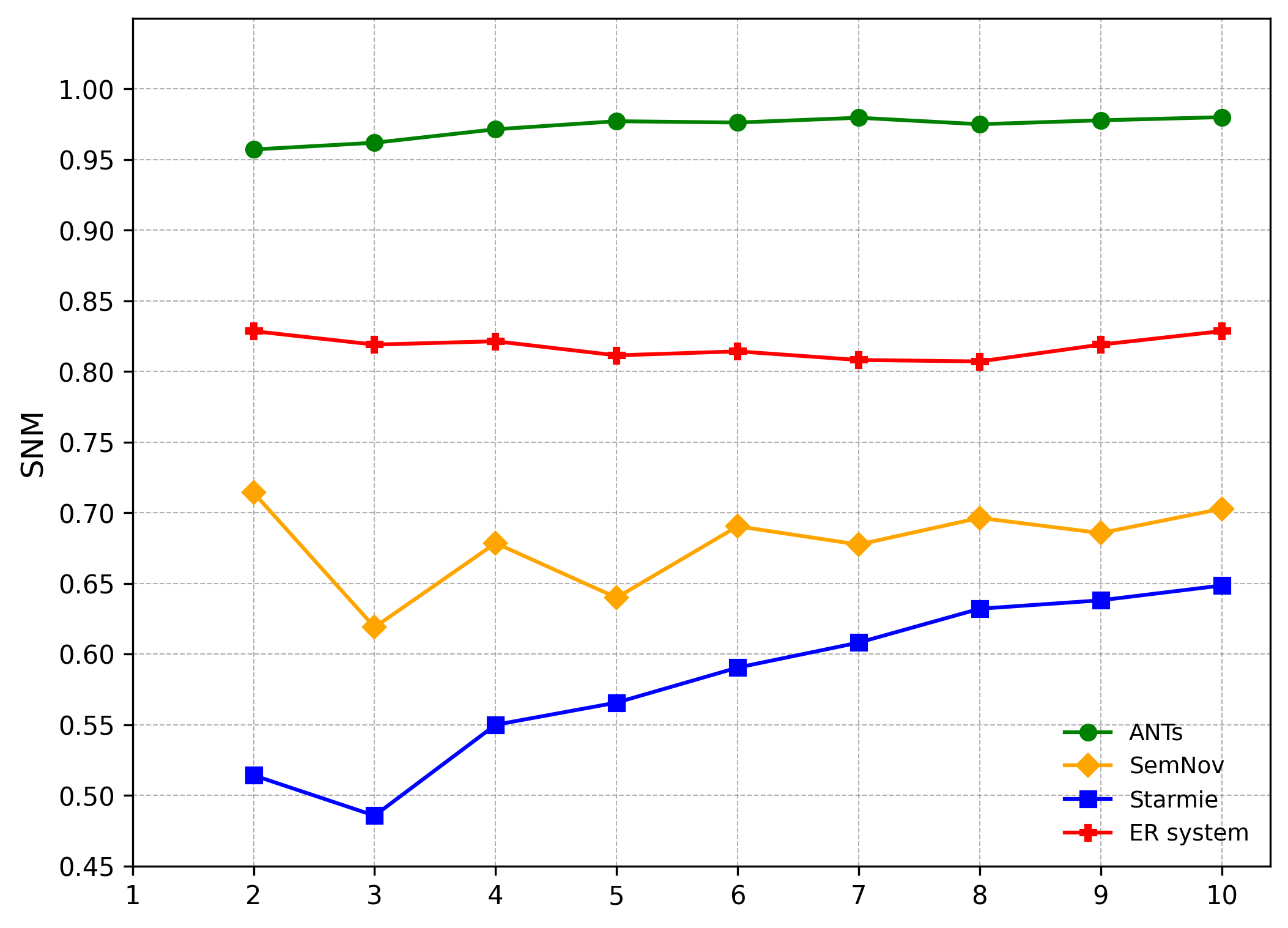}
  \caption{\small Santos}
  \label{fig:snm_santos}
\end{subfigure}

\begin{subfigure}[t]{0.485\columnwidth}
  \includegraphics[width=\linewidth]{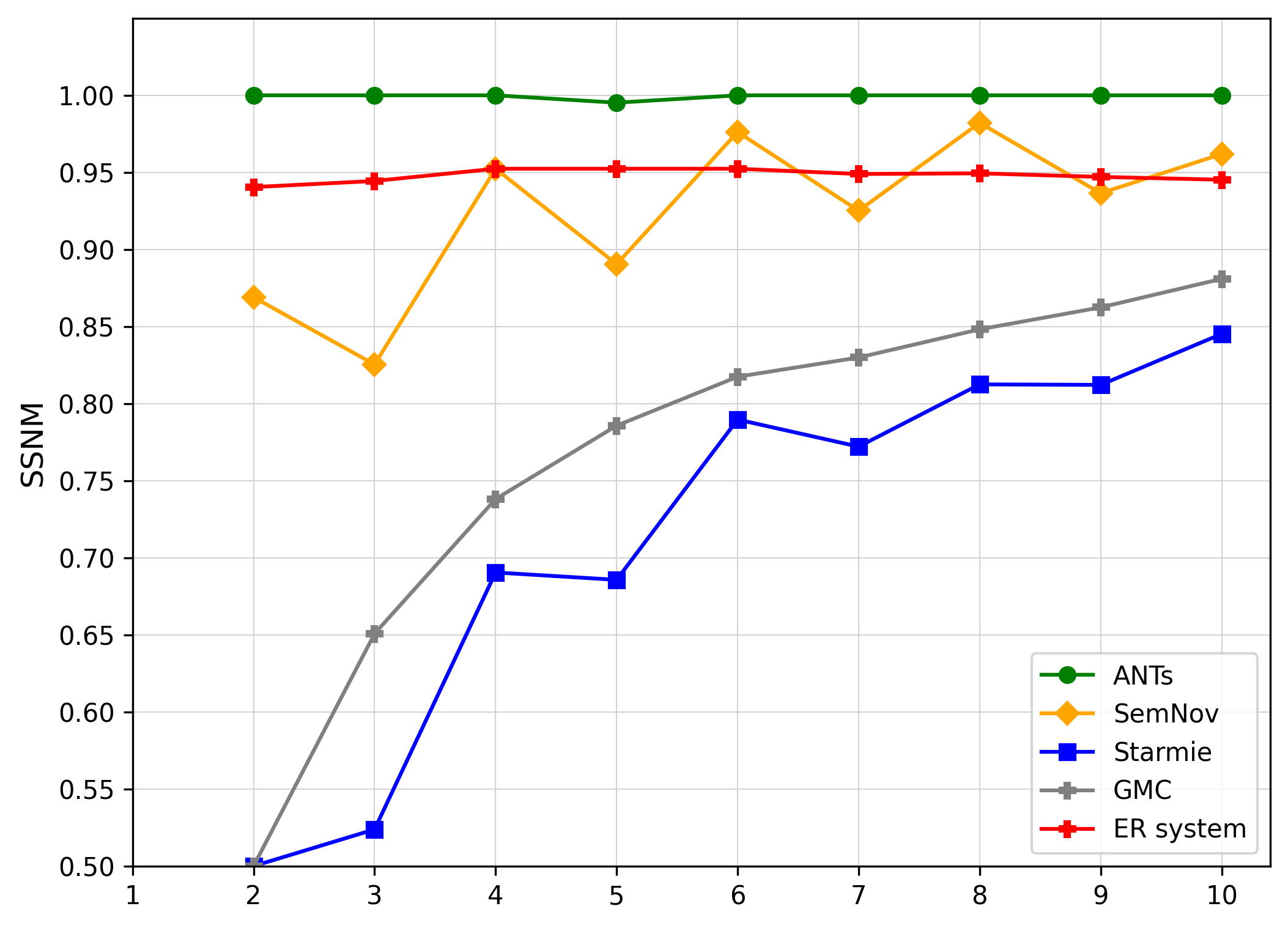}
  \caption{\small TUS}
  \label{fig:ssnm_tus}
\end{subfigure}
\hfill
\begin{subfigure}[t]{0.485\columnwidth}
  \includegraphics[width=\linewidth]{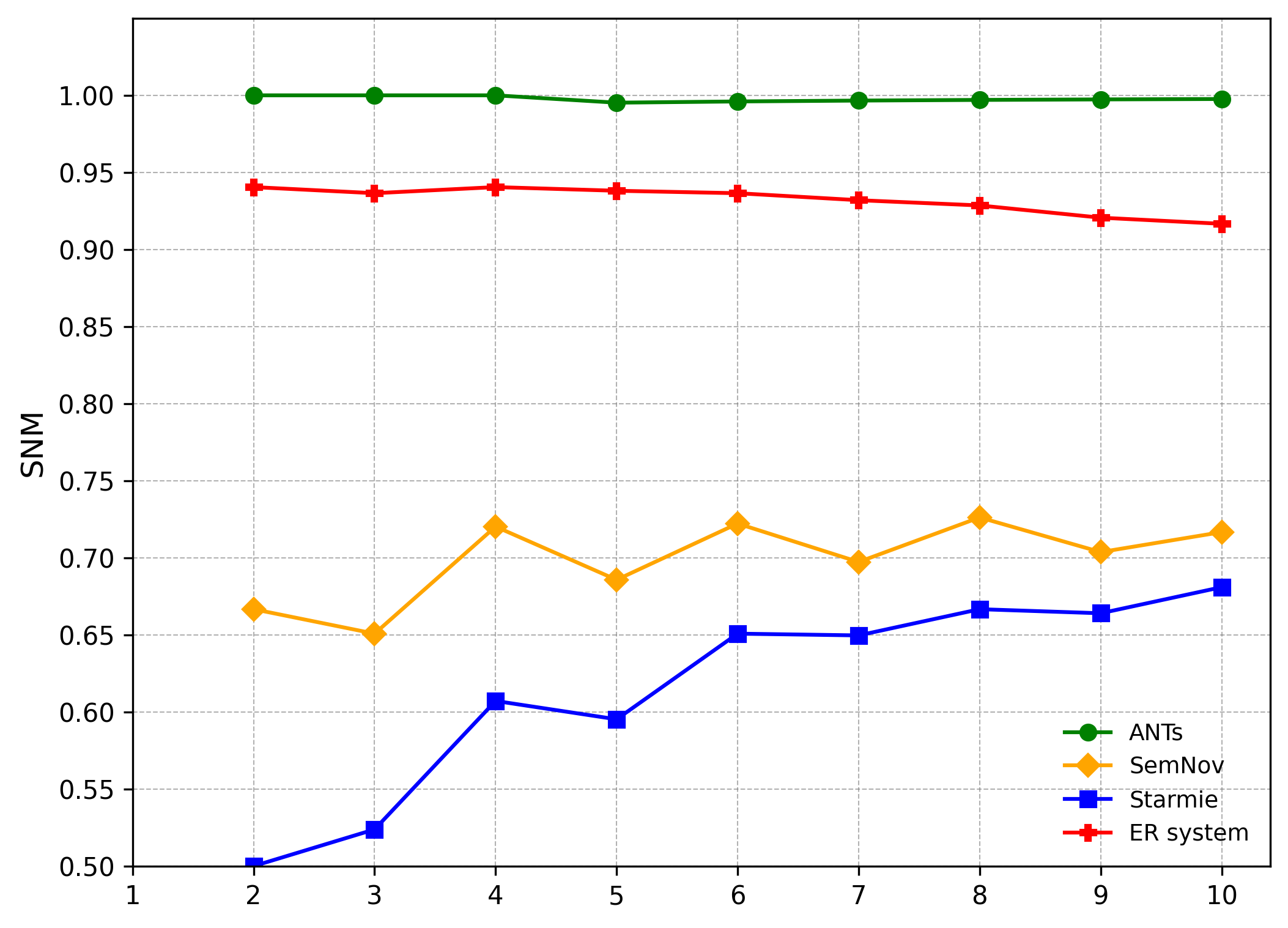}
  \caption{\small TUS}
  \label{fig:snm_tus}
\end{subfigure}

\begin{subfigure}[t]{0.485\columnwidth}
  \includegraphics[width=\linewidth]{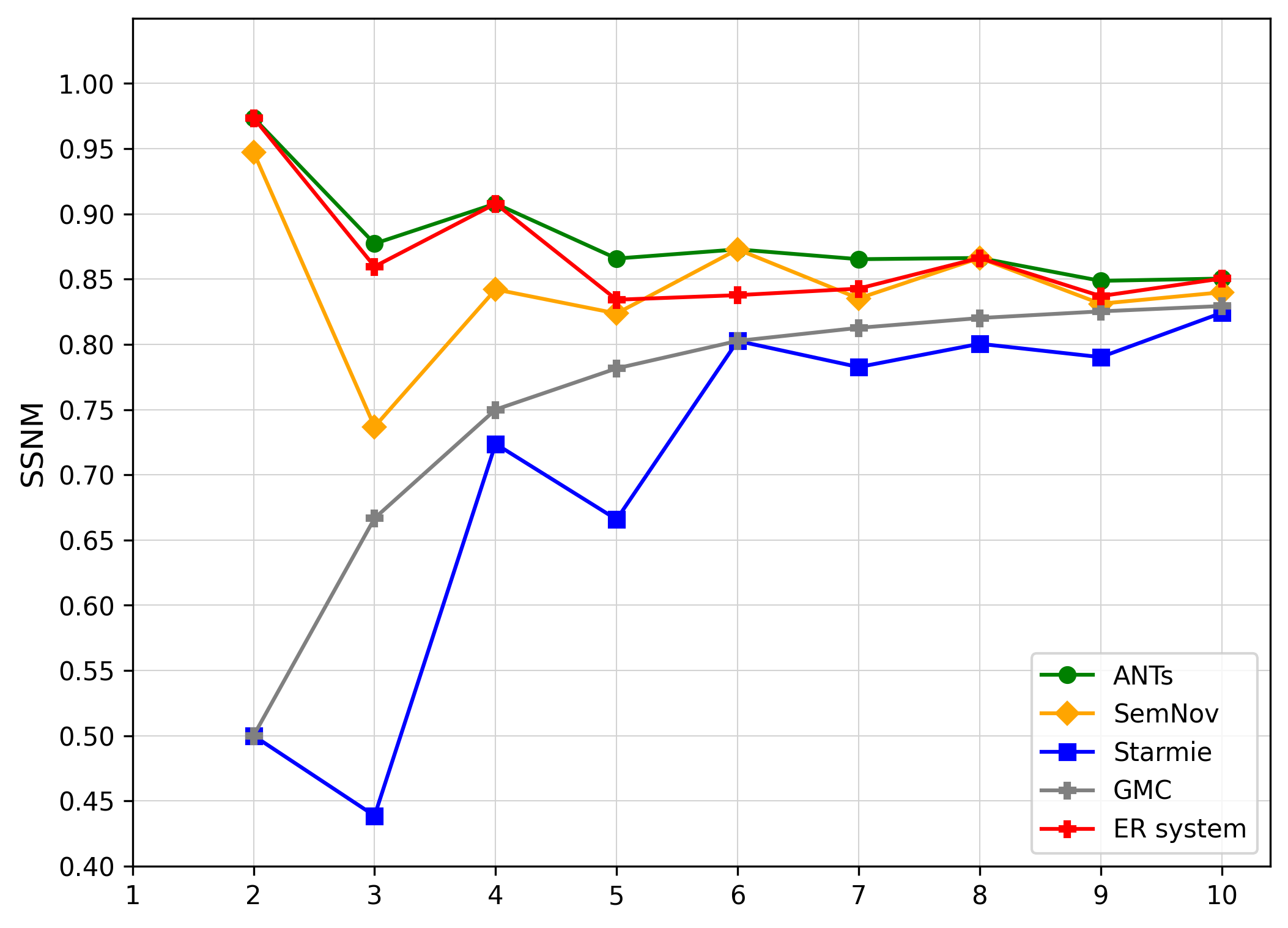} 
  \caption{\small Ugen-v2}
  \label{fig:ssnm_ugen}
\end{subfigure}
\hfill
\begin{subfigure}[t]{0.485\columnwidth}
  \includegraphics[width=\linewidth]{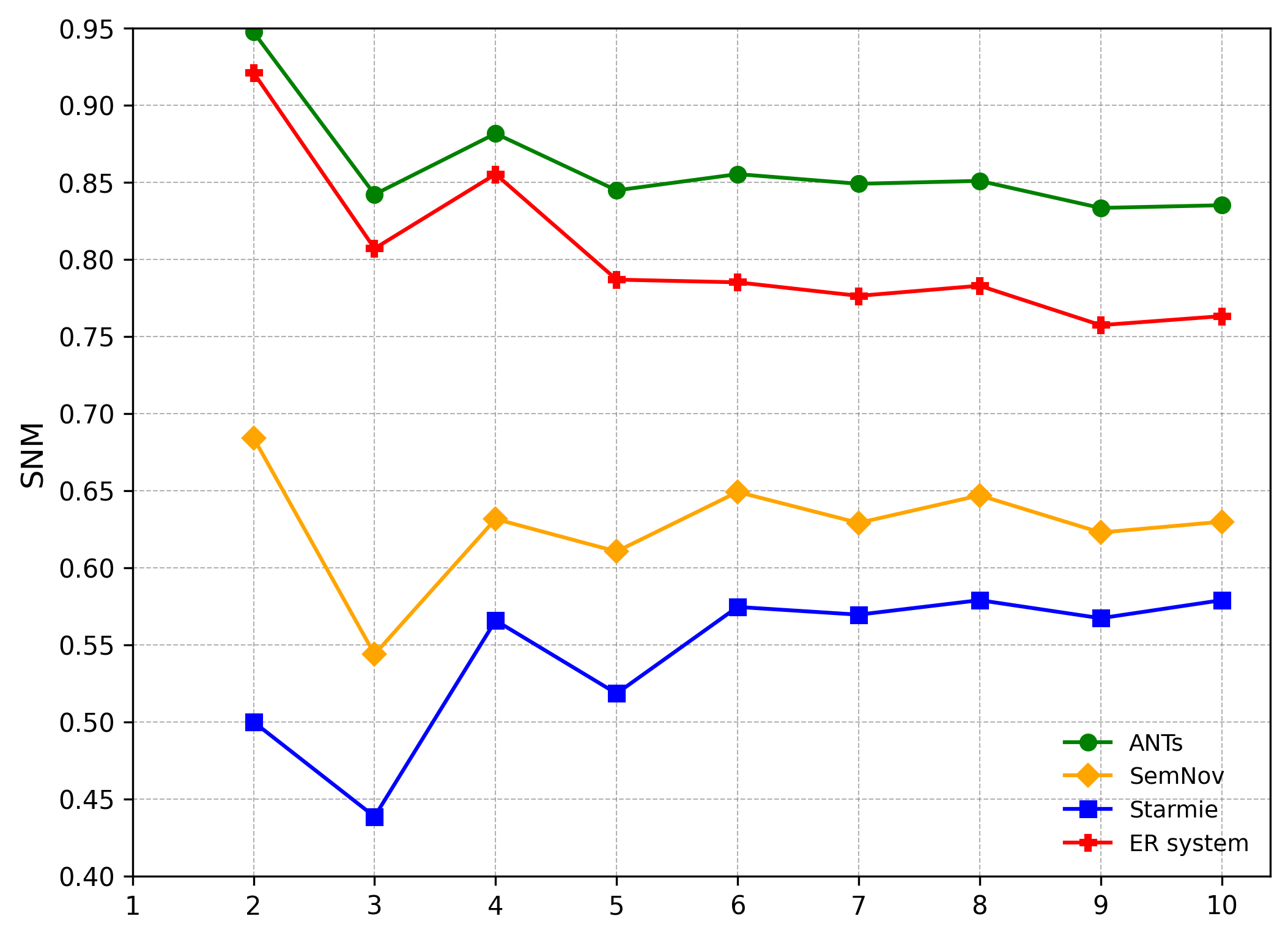} 
  \caption{\small Ugen-v2}
  \label{fig:snm_ugen}
\end{subfigure}
  \vspace{-0.5em}
\caption{\small System performance across datasets for SSNM (left) and SNM (right), $l\!\in\![2,10]$. Y-axis truncated.}
\label{fig:ssnm_vs_snm_all}
\end{figure}
\noindent \textbf{Simplified Syntactic Novelty Measure.}
Compared to SNM, SSNM provides less information, as it measures only the system’s ability to retrieve novel tables not their ranking. Figure~\ref{fig:ssnm_vs_snm_all} (left column (a), (c), (e)) shows the average SSNM across all queries for each dataset. \estimatename again outperforms others on TUS and Santos. 
Given its stable performance, \ERNov ranks as the second-best system. 
As $l$ increases, all methods converge, an expected result given the limited number of unionable tables per query (at most 20). For the Ugen-v2  dataset, a notable drop in most system's performance is observed,  due to the scarcity of unionable tables. 

\noindent \textbf{Search Novelty Score (nscore).}
Computing the search novelty  on our medium-sized datasets (TUS and Santos) proved computationally intensive, even with parallelization using 20 CPUs. 
Therefore, we limited the computation of  search novelty  to $l\! \in\! [2, 3]$ and focused solely on the Ugen-v2 dataset. 

Table~\ref{tb:nscore} demonstrates that \estimatename  outperforms all other methods, achieving the highest average novelty. We assessed statistical significance via paired t-tests for $l\!=\!2$ and $l\!=\!3$. 
Improvements of \estimatename over \Starmie, \GMC, and \semnov are statistically significant ($p\!\leq\!0.05$). \semnov ranks second and statistically significantly outperforms \Starmie\ and \GMC. \estimatename and \semnov achieve novelty scores comparable to \GMM, with no statistically significant differences observed. This indicates that the penalization-based methods achieve novelty scores comparable to those of a provably $\tfrac{1}{2}$-approximate baseline. 

\begin{table}[t]
\caption{\small Average novelty score across Ugen\_v2 queries.}
\label{tb:nscore}
 \vspace{-0.5em}
\centering
\small
{
\setlength{\tabcolsep}{3pt}%
\renewcommand{\arraystretch}{1.05}%
\resizebox{0.9\columnwidth}{!}{%
\begin{tabular}{@{}l| l|cccccc@{}} 
     & $l$ & \Starmie   & \GMC & \GMM & \ERNov & \semnov & \estimatename \\
 \hline
\multirow{2}{*}{$\mathrm{nscore}(T)$} 
  & $2$  & 0.0&   0.004 &  0.3702&0.3704&0.3829  & \textbf{0.3900} \\
  & $3$   & 0.0073 & 0.0065& 0.1816 &0.2383& 0.2412 & \textbf{0.2474} \\
\hline
\end{tabular}%
}
}
\end{table}
Another important observation is that \estimatename is the best-performing system overall in terms of the novelty score (Table~\ref{tb:nscore}). The next best methods are \semnov and \ERNov, which compete for second place on the novelty score, followed by \GMC, with \Starmie performing worst. This ranking is consistent with their relative behavior on SNM, SSNM, and Blatant-Duplicate, indicating that our evaluation metrics effectively capture the intended objectives.

\noindent \textbf{Additional Metrics}.  Table~\ref{tb:gmcmetrics_transposed} compares $\mathcal{F}$ values of \estimatename against the baselines across all benchmarks. Due to space limitations, we report only the average percentage of queries for which \estimatename achieves an equal or higher $\mathcal{F}$ value at $l \!=\! 10$. \ifnottechreport{The results for $l \in [2, 10]$ are provided in the technical report.}{The results for $l \in [2, 10]$ are provided in Table~\ref{tab:extendedofvalue} of Appendix.} \estimatename consistently underperforms \GMC on the Santos and TUS datasets across all levels of $l$, as well as on the Ugen-v2 dataset for $l \in [2, 4]$, highlighting \GMC’s effectiveness in achieving its optimization objective.
However, \estimatename consistently outperforms \semnov and \ERNov across all datasets and levels of $l$, achieving higher performance on at least 54.3\% of the queries. Given that \GMC is computationally expensive, \estimatename serves as a more efficient and effective alternative than others when runtime performance is a concern.

\begin{table}[!htbp]
\caption{\small Percentage of queries where \estimatename achieves equal or higher $\mathcal{F}$ at $l\!=\!10$. Values $\ge\!50$ are italicized.}
\label{tb:gmcmetrics_transposed}
 \vspace{-0.5em}
\centering
{
\setlength{\tabcolsep}{12pt} 
\renewcommand{\arraystretch}{1.1} 
\begin{tabular}{l|cccc}
& \Starmie  &  \GMC & \ERNov & \semnov \\
\hline
\textit{Santos} & 48.6\% &  2.9\% &\textit{54.3\%}& \textit{82.9\%} \\
\textit{TUS}    & \textit{66.7}\% &  2.4\% &\textit{73.8\%}& \textit{66.7\%} \\
\textit{Ugen-v2} & \textit{73.7\% }&  \textit{73.7\%} & \textit{84.2\% }&\textit{94.7\%} \\
\hline
\end{tabular}}
\end{table}

\begin{figure} [!htbp]
    \centering
 \includegraphics[width=0.75\columnwidth]{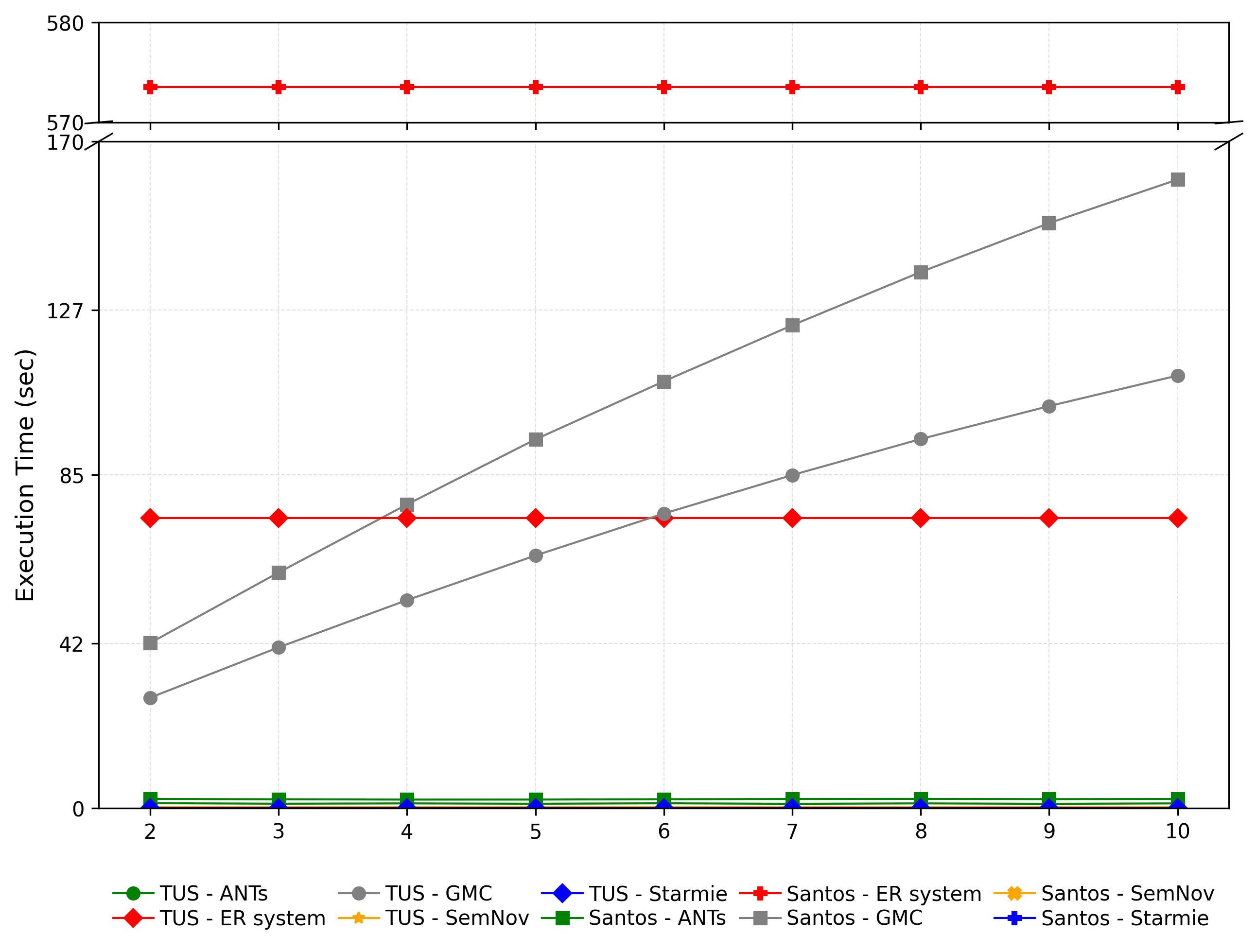}
 \vspace{-0.5em}
    \caption{\small Exec. time across datasets, averaged over all queries ($l\!\in\![2,10]$). Y-axis truncated.}
    \label{fig:exetime1}
\end{figure}
\noindent \textbf{Scalability.} We compare the runtime performance of various systems for the re-ranking task. The execution time for all systems is negligible ($\approx\mathbf{0}$) on Ugen-v2. Therefore, we report the average runtime over all queries for TUS and Santos in Figure~\ref{fig:exetime1}.  \estimatename and \semnov have small execution time across these datasets, less than $2.4$ seconds. In contrast, \GMC and \ERNov incur substantial overhead. This significant latency makes \GMC and \ERNov less suitable particularly for interactive data analysis scenarios over data lakes, where low-latency responses are essential for effective user interaction.

\subsection{Comparison with \DUST}\label{subsec:dust} 
In an upcoming paper, Khatiwada et al.~\cite{KhatiwadaSM26} study a related problem and propose the \DUST system.  Given a query \sTable $Q$, \DUST aims to identify $l$ unionable tuples that are maximally dissimilar to both $Q$ and to one another. Like our approach, their method builds on existing unionable search techniques, but unlike ours, it does not re-rank tables.  Rather,  it operates at the tuple level, returning a set of tuples drawn from multiple unionable tables that are the most diverse.
In addition,  Khatiwada et al.~\cite{KhatiwadaSM26} focus on semantic dissimilarity, whereas our work  
aims to develop a table reranking solution that maximizes syntactic dissimilarity (while maintaining high unionability or semantic similarity).

To adapt \DUST\ for our problem, we evaluate the novelty of its results (a set of diverse tuples which we model as a single table) by computing the search novelty score of this table and comparing this to the search novelty score of the top-1 table returned by \estimatename.  To do this, we find the size of  \estimatename' top ranked table ($l$), and run \DUST\ to find the most diverse $l$ tuples.  We would expect \DUST\ to out-perform \estimatename since it is drawing tuples from many tables. Our experiments are run on Ugen-v2 and  show that, although \DUST achieves a slightly higher novelty score than \estimatename, it does so only after retrieving many more tables and tuples and incurring substantially higher running time (Table~\ref{tb:dust-ants_main}).
We report the average size of the top-1 \estimatename table over all queries along with the total number of acquired tuples (and tables used in the search (\texttt{\#Acq.tup}).  \texttt{\#Acq.tup} is just the size of the Top-1 for \estimatename and is the sum of the tuples in $k$ table from which \DUST\ is drawing its diverse tuples.   Hence, in data-market settings with tuple-acquisition costs or in latency-sensitive scenarios, \estimatename offers a better cost–benefit trade-off, whereas when novelty is the sole objective and resources are unconstrained, \DUST is preferable.

\begin{table}[t]
\caption{\small The novelty score, the number of acquired tuples and tables, the final number of novel tuples, and the execution time are shown. The best values are shown in \textbf{bold}.}
\label{tb:dust-ants_main}
\centering
 \vspace{-0.5em}
\setlength{\tabcolsep}{1.5pt}%
\renewcommand{\arraystretch}{1.1}%
\begin{tabular}{l|ccccc}
 & $\mathrm{nscore}(T)$ & \#Acq.tup. & \#Nov.tup. & \#Acq.tbl. & Exec.time(s) \\
\hline
\estimatename & 0.4614          & \textbf{23} &\textbf{23} & \textbf{1} & $\approx\mathbf{0}$ \\
\hline
\DUST         & \textbf{0.4966} & 416         & \textbf{23} & 5         & 101 \\
\hline
\end{tabular}
\end{table}

\subsection{Ablation Studies: Constituent Techniques and Parameters}\label{subsec:alignment_intialresult_impact}
We report the main findings from three ablation studies on key components of our approach. 

\noindent{\em Alignment Quality:}  Using an automatic alignment algorithm~\cite{KhatiwadaSM26} (instead of the ground truth alignment) is found to be detrimental in 92\% of all (system X, dataset Y) settings for SSNM and in 89\% of settings for SNM. However, even with poor quality alignments, \estimatename continues to out-perform other \name approaches in 48\% and 89\% of all possible ($l$, dataset Y) settings for SSNM and SNM, respectively.

\noindent{\em Novelty of Search Result:}  Next, we study how the novelty of the initial search results affects the effectiveness of our novelty-aware reranker. Our experiments show that increasing the redundancy of the initial result  is detrimental to \estimatename's performance on all datasets in terms of SSNM, with a similar pattern observed for SNM.

\noindent{\em Hyperparameters:} An analysis of the ablation results on the domain-size threshold $s$ and penalization degree $b$ indicates that our method is robust to the choice of these parameters on the datasets considered: SNM and SSNM exhibit only limited variation across the tested values. Thus, \estimatename is not sensitive to the exact choice of $s$ and $b$: its performance remains stable even when these parameters are not set to their empirically optimal values, $s = 10$ and $b \in {4,5}$ across all datasets.

\section{\estimatename for Downstream Task}\label{sec:ANTsMLTask_summary}
We evaluate the impact of \estimatename, on a downstream rating-prediction task using the freely available IMDb Extensive Dataset.\footnote{\url{https://www.kaggle.com/datasets/simhyunsu/imdbextensivedataset}} We compare three configurations for constructing the training data:
(i) \Baseline, no data discovery is applied and we use only the query table;
(ii) \Starmie, we augment the query table with the tables returned by \Starmie; and
(iii) \estimatename,  we augment the query table with the tables returned by \Starmie but reranked by \estimatename.  We  build a regression model to predict the average user rating (\texttt{avg\_vote}) from movie metadata.\footnote{This task design is a modified version of a task used to show the efficacy of other data discovery approaches~\cite{DBLP:journals/pvldb/FanWLZM23}.}  The IMDb movie
table is randomly split into training (80\%) and test (20\%) sets. From the training set, we randomly select 1\% of the tuples and partition them into three distinct query tables, each retaining all predictive attributes (features). A synthetic data lake of unionable tables is created from the remaining training tuples following Nargesian et al.~\cite{DBLP:journals/pvldb/NargesianZPM18}. For each query, we additionally create a diluted variant of the unionable tables (with a dilution degree of 0.4) and manually insert the query table into the  data lake.
We use three standard regressors: Gradient Boosting Machine (LGBM), Random Forest (RF), and XGBoost (XGB), and evaluate them using standard regression metrics (R\textsuperscript{2} and RMSE).\footnote{
We implemented and evaluated all models using scikit-learn’s APIs~\url{https://scikit-learn.org}. } 
Training regressors on random selections of $k$ unionable tables,  we observe that LGBM consistently outperforms the others across all $k \! \in \![2,10]$ for both R\textsuperscript{2} and RMSE; we therefore adopt LGBM as the model for comparing table discovery strategies. On the diluted dataset, where redundancy in the training dataset is guaranteed, \estimatename consistently outperforms both \Starmie and \Baseline for all values of $k$ and for both R\textsuperscript{2} and RMSE. A paired t-test ($p\!\leq\!0.05$) shows that improvements over \Baseline are significant;  all improvements over \Starmie are statistically significant except for $k{=}5$. On the non-diluted dataset, for both R\textsuperscript{2} and RMSE, \estimatename significantly outperforms \Baseline, but no significant difference is observed between \estimatename and \Starmie. Overall, we conclude that in the presence of redundancy, reranking the \Starmie results by \estimatename is consistently beneficial for downstream machine learning tasks. In the non-diluted setting, the reranking step is not detrimental; thus, the novelty-aware reranker can be safely applied always  even when the redundancy level of the dataset is unknown.

\section{Related Work}\label{sec:related}
\noindent \textit{Table Discovery}.
Given a query table and an expansion strategy, discovering useful tables from data lakes has been extensively studied in recent years~\cite{DBLP:journals/vldb/ChapmanSKKIKG20,DBLP:conf/sigmod/Fan00M23}, with joinable and unionable expansions receiving the most attention. Various aspects of the problem have been addressed, including scalability and accuracy. Focusing on unionable table search, Nargesian et al.~\cite{DBLP:journals/pvldb/NargesianZPM18} define two tables as unionable if they share unionable attributes.
To assess  attribute unionability, they use  attribute values, their mappings to an ontology, along with semantic embeddings of values.  D3L extends the notion of attribute unionability~\cite{DBLP:conf/icde/BogatuFP020}. Khatiwada et al.~\cite{santos23} further improved unionable table search by modeling relationships between attributes.
Recognizing that  attribute understanding is critical for discovering unionable tables, recent research has increasingly focused on enhancing attribute representations. Hu et al.~\cite{DBLP:conf/acl/HuWQLSFKKWY23} and Fan et al.~\cite{DBLP:journals/pvldb/FanWLZM23} incorporate contextualized information to develop richer and more effective representations of attributes. 

\vspace{0.2em}
\noindent \textit{Novel Data Discovery}. A rich body of literature explores methods for retrieving novel results across various types of systems, categorized by the forms of user queries they support, including direct user queries (ad hoc information retrieval IR)~\cite{DBLP:conf/sigir/ClarkeKCVABM08,DBLP:journals/tkde/WuZMLHML24}, indirect user queries (recommender systems), queries expressed as quasi-SQL statements~\cite{DBLP:conf/vldb/MirzaeiR23}, or queries represented as tables~\cite{KhatiwadaSM26}. Some systems also allow users to express preferences regarding the trade-off between result quality and retrieval speed~\cite{DBLP:journals/pvldb/VieiraRBHSTT11,DBLP:conf/vldb/MirzaeiR23}.
For a comprehensive overview of related work on novel search in IR and recommender systems, we refer the reader to Wu et al.~\cite{DBLP:journals/tkde/WuZMLHML24}.
\vspace{0.2em}

\noindent \textit{Tuple Similarity.}
Our approach is grounded in syntactic similarity between tuples from the query table and candidate unionable tables. Prior work has proposed various techniques for comparing tuples in other contexts.  In the absence of schema information and other metadata, a framework is developed for computing similarity between database instances even when they contain missing data (i.e., null values)~\cite{DBLP:conf/edbt/GlavicMMPSV24}.  Glavic et al.~\cite{DBLP:conf/edbt/GlavicMMPSV24} developed a similarity measure based on matching between tuples of database instances. They distinguish between different types of null values, which they leverage in developing their measure. In contrast, our approach assumes that all null values are treated uniformly.

\section{Conclusions and Future Work}\label{sec:conclusion}
We addressed the problem of identifying novel unionable tables from data lakes. To this end, we proposed two general properties that any scoring mechanism should satisfy when ranking tables based on their novelty. We propose and compare several novel table search (\name) methods.  The first, \estimatename, ranks tables with respect to relevance (semantic similarity) and diversity (syntactic dissimilarity), that is the
retrieved tables are as relevant as possible to the
query, and, at the same time, the result set is as 
diverse as possible.  \estimatename is a new, scalable, approximate algorithm and we compare this to \GMC, a known slower approximate algorithm for query diversification using the same relevance and diversity functions as \estimatename.   \semnov, uses semantic dissimilarity (specifically the distance between table embeddings) instead of syntactic dissimilarity.  The fourth, ER, uses an entity-based dissimilarity, where the inverse of entity overlap (after applying entity-resolution) is used as table
diversity. Our experiments show that \estimatename outperforms the other \name methods.
We also demonstrate its positive impact on a downstream data-intensive task and conduct several ablation studies of its components.
Several aspects of this work remain open for future exploration.

\noindent \textit{Query Table Quality.} If a query \sTable contains 
poor quality data, a search algorithm may lack sufficient information to identify truly novel tables. Drawing inspiration from IR query reformulation techniques~\cite{DBLP:conf/doceng/KassaieKT25}, one approach to enhance the quality of a query \sTable  is to enrich the query using hypothetical tuples or attributes generated by an oracle (e.g., LLMs). We plan to investigate the impact of query expansion in \name. 

\vspace{0.2em}
\noindent \textit{Novelty-Aware Table Embeddings.} In this work, we 
computed relevance using the similarity of attribute embeddings. This was used as a proxy for table unionability.
As a direction for future work, we aim to incorporate diversity (or novelty) directly into the model’s objective function for creating attribute embeddings, allowing us to address \name through a unified and end-to-end process.

\vspace{0.2em}
\noindent \textit{Benchmark.}
The benchmarks used in our experiments were not originally designed for the \name problem, and as a result, they may not fully capture the capabilities of our framework. As part of future work, we plan to develop new open benchmarks tailored specifically to the \name task. 

\section{Acknowledgment}

We acknowledge the
support of the Canada Excellence Research Chairs (CERC) program. Nous remercions le Chaires d’excellence en recherche du
Canada (CERC) de son soutien.

\section{AI-Generated Content Acknowledgment}
All ideas, methods, and results presented in this paper were conceived and developed entirely by the authors. The use of AI tools did not contribute to the formulation of any scientific ideas, methodologies, analyses, or conclusions. An AI-based writing assistant was used occasionally to suggest alternative wording and correct spelling or grammatical errors.


\bibliographystyle{IEEEtran}   
\bibliography{ref/core,ref/table,ref/other}

@article{santos23,
	author    = {Aamod Khatiwada and
	               Grace Fan and
	               Roee Shraga and
	               Zixuan Chen and
	               Wolfgang Gatterbauer and
	               Renée J. Miller and
	               Mirek Riedewald},
	title     = {SANTOS: Relationship-based Semantic Table Union Search},
	journal = {{Proc. ACM Manag. Data}},
    volume    = {1},
    number    = {1},
    pages     ={1--25},
	year      = {2023}
}

@article{DBLP:journals/pvldb/NargesianZPM18,
  author       = {Fatemeh Nargesian and
                  Erkang Zhu and
                  Ken Q. Pu and
                  Ren{\'{e}}e J. Miller},
  title        = {Table Union Search on Open Data},
  journal      = {Proc. {VLDB} Endow.},
  volume       = {11},
  number       = {7},
  pages        = {813--825},
  year         = {2018}}

@inproceedings{DBLP:conf/acl/HuWQLSFKKWY23,
  author       = {Xuming Hu and
                  Shen Wang and
                  Xiao Qin and
                  Chuan Lei and
                  Zhengyuan Shen and
                  Christos Faloutsos and
                  Asterios Katsifodimos and
                  George Karypis and
                  Lijie Wen and
                  Philip S. Yu},
  title        = {Automatic Table Union Search with Tabular Representation Learning},
  booktitle    = {{ACL}},
  pages        = {3786--3800},
  year         = {2023},
  doi          = {10.18653/V1/2023.FINDINGS-ACL.233},
  bibsource    = {dblp computer science bibliography, https://dblp.org}
}

@inproceedings{DBLP:conf/vldb/MirzaeiR23,
  author       = {Hamed Mirzaei and
                  Davood Rafiei},
  title        = {Table Union Search with Preferences},
  booktitle   = {Proc. {VLDB} Workshops},
  year         = {2023},
  bibsource    = {dblp computer science bibliography, https://dblp.org}
}

@inproceedings{DBLP:conf/edbt/GlavicMMPSV24,
  author       = {Boris Glavic and
                  Giansalvatore Mecca and
                  Ren{\'{e}}e J. Miller and
                  Paolo Papotti and
                  Donatello Santoro and
                  Enzo Veltri},
  title        = {Similarity Measures For Incomplete Database Instances},
  booktitle    = {{EDBT}},
  pages        = {461--473},
  year         = {2024},
  doi          = {10.48786/EDBT.2024.40},
  bibsource    = {dblp computer science bibliography, https://dblp.org}
}

@inproceedings{DBLP:conf/icde/SantosBMF22,
	author    = {A{\'{e}}cio S. R. Santos and
	Aline Bessa and
	Christopher Musco and
	Juliana Freire},
	title     = {A Sketch-based Index for Correlated Dataset Search},
	booktitle = {ICDE},
	pages     = {2928--2941},
	year      = {2022}
}

@article{DBLP:journals/pvldb/CasteloRSBCF21,
	author    = {Sonia Castelo and
	R{\'{e}}mi Rampin and
	A{\'{e}}cio S. R. Santos and
	Aline Bessa and
	Fernando Chirigati and
	Juliana Freire},
	title     = {Auctus: {A} Dataset Search Engine for Data Discovery and Augmentation},
	journal   = {Proc. {VLDB} Endow.},
	volume    = {14},
	number    = {12},
	pages     = {2791--2794},
	year      = {2021}
}

@inproceedings{DBLP:conf/icde/DongT0O21,
	author    = {Yuyang Dong and
	Kunihiro Takeoka and
	Chuan Xiao and
	Masafumi Oyamada},
	title     = {Efficient Joinable Table Discovery in Data Lakes: {A} High-Dimensional
	Similarity-Based Approach},
	booktitle =  {ICDE},
	pages     = {456--467},
	year      = {2021}
}

@inproceedings{DBLP:conf/icde/KoutrasSIPBFLBK21,
	author    = {Christos Koutras and
	George Siachamis and
	Andra Ionescu and
	Kyriakos Psarakis and
	Jerry Brons and
	Marios Fragkoulis and
	Christoph Lofi and
	Angela Bonifati and
	Asterios Katsifodimos},
	title     = {Valentine: Evaluating Matching Techniques for Dataset Discovery},
	booktitle = {ICDE},
	pages     = {468--479},
	year      = {2021}
}

@inproceedings{DBLP:conf/sigmod/ZhangI20,
	author    = {Yi Zhang and
	Zachary G. Ives},
	title     = {Finding Related Tables in Data Lakes for Interactive Data Science},
	booktitle = {SIGMOD},
	pages     = {1951--1966},
	year      = {2020}
}

@inproceedings{DBLP:conf/sigmod/ZhuDNM19,
  author       = {Erkang Zhu and
                  Dong Deng and
                  Fatemeh Nargesian and
                  Ren{\'{e}}e J. Miller},
  title        = {{JOSIE:} Overlap Set Similarity Search for Finding Joinable Tables
                  in Data Lakes},
  booktitle = {SIGMOD},
  pages        = {847--864},
  year         = {2019},
  doi          = {10.1145/3299869.3300065},
  bibsource    = {dblp computer science bibliography, https://dblp.org}
}

@article{DBLP:journals/pvldb/KhatiwadaSGM22,
  author       = {Aamod Khatiwada and
                  Roee Shraga and
                  Wolfgang Gatterbauer and
                  Ren{\'{e}}e J. Miller},
  title        = {Integrating Data Lake Tables},
  journal      = {Proc. {VLDB} Endow.},
  volume       = {16},
  number       = {4},
  pages        = {932--945},
  year         = {2022},
  doi          = {10.14778/3574245.3574274},
  bibsource    = {dblp computer science bibliography, https://dblp.org}
}

@article{DBLP:journals/tkde/WuZMLHML24,
  author       = {Haolun Wu and
                  Yansen Zhang and
                  Chen Ma and
                  Fuyuan Lyu and
                  Bowei He and
                  Bhaskar Mitra and
                  Xue Liu},
  title        = {Result Diversification in Search and Recommendation: {A} Survey},
  journal      = {{IEEE} Trans. Knowl. Data Eng.},
  volume       = {36},
  number       = {10},
  pages        = {5354--5373},
  year         = {2024},
  doi          = {10.1109/TKDE.2024.3382262},
  bibsource    = {dblp computer science bibliography, https://dblp.org}
}

@book{DBLP:books/fm/GareyJ79,
  author       = {M. R. Garey and
                  David S. Johnson},
  title        = {Computers and Intractability: {A} Guide to the Theory of NP-Completeness},
  year         = {1979},
  isbn         = {0-7167-1044-7},
  bibsource    = {dblp computer science bibliography, https://dblp.org}
}

@inproceedings{DBLP:conf/sigir/ClarkeKCVABM08,
  author       = {Charles L. A. Clarke and
                  Maheedhar Kolla and
                  Gordon V. Cormack and
                  Olga Vechtomova and
                  Azin Ashkan and
                  Stefan B{\"{u}}ttcher and
                  Ian MacKinnon},
  title        = {Novelty and diversity in information retrieval evaluation},
  booktitle    = {{SIGIR}},
  pages        = {659--666},
  year         = {2008},
  doi          = {10.1145/1390334.1390446},
  bibsource    = {dblp computer science bibliography, https://dblp.org}
}

@article{DBLP:journals/pvldb/VieiraRBHSTT11,
  author       = {Marcos R. Vieira and
                  Humberto Luiz Razente and
                  Maria Camila Nardini Barioni and
                  Marios Hadjieleftheriou and
                  Divesh Srivastava and
                  Caetano Traina Jr. and
                  Vassilis J. Tsotras},
  title        = {DivDB: {A} System for Diversifying Query Results},
  journal      = {Proc. {VLDB} Endow.},
  volume       = {4},
  number       = {12},
  pages        = {1395--1398},
  year         = {2011},
  bibsource    = {dblp computer science bibliography, https://dblp.org}
}

@article{DBLP:journals/fcsc/YuLDF16,
  author       = {Minghe Yu and
                  Guoliang Li and
                  Dong Deng and
                  Jianhua Feng},
  title        = {String similarity search and join: a survey},
  journal      = {Frontiers Comput. Sci.},
  volume       = {10},
  number       = {3},
  pages        = {399--417},
  year         = {2016},
  doi          = {10.1007/S11704-015-5900-5},
  bibsource    = {dblp computer science bibliography, https://dblp.org}
}

@inproceedings{DBLP:conf/naacl/IidaTMI21,
  author       = {Hiroshi Iida and
                  Dung Thai and
                  Varun Manjunatha and
                  Mohit Iyyer},
  title        = {{TABBIE:} Pretrained Representations of Tabular Data},
  booktitle    = {{NAACL-HLT}},
  pages        = {3446--3456},
  year         = {2021},
  doi          = {10.18653/V1/2021.NAACL-MAIN.270},
  bibsource    = {dblp computer science bibliography, https://dblp.org}
}

@article{DBLP:journals/program/Porter80,
  author       = {Martin F. Porter},
  title        = {An algorithm for suffix stripping},
  journal      = {Program},
  volume       = {14},
  number       = {3},
  pages        = {130--137},
  year         = {1980},
  doi          = {10.1108/EB046814},
  bibsource    = {dblp computer science bibliography, https://dblp.org}
}

@inproceedings{DBLP:conf/vldb/PalKSM24,
  author       = {Koyena Pal and
                  Aamod Khatiwada and
                  Roee Shraga and
                  Ren{\'{e}}e J. Miller},
  title        = {{ALT-GEN:} Benchmarking Table Union Search using Large Language Models},
   booktitle      = {{Proc. VLDB Workshops}},
  year         = {2024},
  bibsource    = {dblp computer science bibliography, https://dblp.org}
}

@inproceedings{DBLP:conf/icde/BogatuFP020,
	author    = {Alex Bogatu and
	Alvaro A. A. Fernandes and
	Norman W. Paton and
	Nikolaos Konstantinou},
	title     = {Dataset Discovery in Data Lakes},
	booktitle = {ICDE},
	pages     = {709--720},
	year      = {2020}
}

@inproceedings{DBLP:conf/www/BrickleyBN19,
	author    = {Dan Brickley and
	Matthew Burgess and
	Natasha F. Noy},
	title     = {Google Dataset Search: Building a search engine for datasets in an
	open Web ecosystem},
	booktitle =  {WWW},
	pages     = {1365--1375},
	year      = {2019}
}

@article{DBLP:journals/vldb/NeuhofFPNANK24,
  author       = {Franziska Neuhof and
                  Marco Fisichella and
                  George Papadakis and
                  Konstantinos Nikoletos and
                  Nikolaus Augsten and
                  Wolfgang Nejdl and
                  Manolis Koubarakis},
  title        = {Open benchmark for filtering techniques in entity resolution},
    journal      = {Proc. {VLDB} Endow.},
  volume       = {33},
  number       = {5},
  pages        = {1671--1696},
  year         = {2024},
  doi          = {10.1007/S00778-024-00868-7},
  bibsource    = {dblp computer science bibliography, https://dblp.org}
}

@inproceedings{DBLP:conf/icdt/Moumoulidou0M21,
  author       = {Zafeiria Moumoulidou and
                  Andrew McGregor and
                  Alexandra Meliou},
  title        = {Diverse Data Selection under Fairness Constraints},
  booktitle    = {{ICDT}},
  volume       = {186},
  pages        = {13:1--13:25},
  year         = {2021},
  doi          = {10.4230/LIPICS.ICDT.2021.13},
  bibsource    = {dblp computer science bibliography, https://dblp.org}
}

@article{DBLP:journals/pvldb/KondaDCDABLPZNP16,
  author       = {Pradap Konda and
                  Sanjib Das and
                  Paul Suganthan G. C. and
                  AnHai Doan and
                  Adel Ardalan and
                  Jeffrey R. Ballard and
                  Han Li and
                  Fatemah Panahi and
                  Haojun Zhang and
                  Jeffrey F. Naughton and
                  Shishir Prasad and
                  Ganesh Krishnan and
                  Rohit Deep and
                  Vijay Raghavendra},
  title        = {Magellan: Toward Building Entity Matching Management Systems},
  journal      = {Proc. {VLDB} Endow.},
  volume       = {9},
  number       = {12},
  pages        = {1197--1208},
  year         = {2016},
  doi          = {10.14778/2994509.2994535},
  bibsource    = {dblp computer science bibliography, https://dblp.org}
}

@inproceedings{KhatiwadaSM26,
  author       = {Aamod Khatiwada and
                  Roee Shraga and
                  Ren{\'{e}}e J. Miller},
  title        = {Diverse Unionable Tuple Search: Novelty-Driven Discovery in Data Lakes},
  booktitle    = {{EDBT}},
  pages        = {42--55},
  year         = {2026},
  doi          = {10.48786/EDBT.2026.04},
  bibsource    = {dblp computer science bibliography, https://dblp.org}
}

@inproceedings{DBLP:conf/icml/GuoSLGSCK20,
  author       = {Ruiqi Guo and
                  Philip Sun and
                  Erik Lindgren and
                  Quan Geng and
                  David Simcha and
                  Felix Chern and
                  Sanjiv Kumar},
  title        = {Accelerating Large-Scale Inference with Anisotropic Vector Quantization},
  booktitle    = {{ICML}},
  volume       = {119},
  pages        = {3887--3896},
  year         = {2020},
  bibsource    = {dblp computer science bibliography, https://dblp.org}
}

@article{DBLP:journals/vldb/SilvaALPA13,
  author       = {Yasin N. Silva and
                  Walid G. Aref and
                  Per{-}{\AA}ke Larson and
                  Spencer Pearson and
                  Mohamed H. Ali},
  title        = {Similarity queries: their conceptual evaluation, transformations,
                  and processing},
   journal      = {Proc. {VLDB} Endow.},
  volume       = {22},
  number       = {3},
  pages        = {395--420},
  year         = {2013},
  doi          = {10.1007/S00778-012-0296-4},
  bibsource    = {dblp computer science bibliography, https://dblp.org}
}

@article{DBLP:journals/corr/abs-2411-07267,
  author       = {Jiayao Zhang and
                  Yuran Bi and
                  Mengye Cheng and
                  Jinfei Liu and
                  Kui Ren and
                  Qiheng Sun and
                  Yihang Wu and
                  Yang Cao and
                  Raul Castro Fernandez and
                  Haifeng Xu and
                  Ruoxi Jia and
                  Yongchan Kwon and
                  Jian Pei and
                  Jiachen T. Wang and
                  Haocheng Xia and
                  Li Xiong and
                  Xiaohui Yu and
                  James Zou},
  title        = {A Survey on Data Markets},
  journal      = {CoRR},
  volume       = {abs/2411.07267},
  year         = {2024},
  doi          = {10.48550/ARXIV.2411.07267},
  eprinttype    = {arXiv},
  eprint       = {2411.07267},
  bibsource    = {dblp computer science bibliography, https://dblp.org}
}

@article{DBLP:journals/isr/MehtaDJM21,
  author       = {Sameer Mehta and
                  Milind Dawande and
                  Ganesh Janakiraman and
                  Vijay S. Mookerjee},
  title        = {How to Sell a Data Set? Pricing Policies for Data Monetization},
  journal      = {Inf. Syst. Res.},
  volume       = {32},
  number       = {4},
  pages        = {1281--1297},
  year         = {2021},
  doi          = {10.1287/ISRE.2021.1027},
  bibsource    = {dblp computer science bibliography, https://dblp.org}
}

@article{DBLP:journals/pvldb/Thirumuruganathan21,
  author       = {Saravanan Thirumuruganathan and
                  Han Li and
                  Nan Tang and
                  Mourad Ouzzani and
                  Yash Govind and
                  Derek Paulsen and
                  Glenn Fung and
                  AnHai Doan},
  title        = {Deep Learning for Blocking in Entity Matching: {A} Design Space Exploration},
  journal      = {Proc. {VLDB} Endow.},
  volume       = {14},
  number       = {11},
  pages        = {2459--2472},
  year         = {2021},
  doi          = {10.14778/3476249.3476294},
  bibsource    = {dblp computer science bibliography, https://dblp.org}
}

@article{DBLP:journals/vldb/LiLSDT23,
  author       = {Yuliang Li and
                  Jinfeng Li and
                  Yoshi Suhara and
                  AnHai Doan and
                  Wang{-}Chiew Tan},
  title        = {Effective entity matching with transformers},
    journal      = {Proc. {VLDB} Endow.},
  volume       = {32},
  number       = {6},
  pages        = {1215--1235},
  year         = {2023},
  doi          = {10.1007/S00778-023-00779-Z},
  bibsource    = {dblp computer science bibliography, https://dblp.org}
}

@article{DBLP:journals/debu/MillerNZCPA18,
	author    = {Ren{\'{e}}e J. Miller and
	Fatemeh Nargesian and
	Erkang Zhu and
	Christina Christodoulakis and
	Ken Q. Pu and
	Periklis Andritsos},
	title     = {Making Open Data Transparent: Data Discovery on Open Data},
	journal   = {{IEEE} Data Eng. Bull.},
	volume    = {41},
	number    = {2},
	pages     = {59--70},
	year      = {2018}
}

@article{DBLP:journals/siamrev/Kieffer94,
  author       = {John Kieffer},
  title        = {Elements of Information Theory (Thomas M. Cover and Joy A. Thomas)},
  journal      = {{SIAM} Rev.},
  volume       = {36},
  number       = {3},
  pages        = {509--511},
  year         = {1994},
  doi          = {10.1137/1036124},
  bibsource    = {dblp computer science bibliography, https://dblp.org}
}

@inproceedings{DBLP:conf/icde/VieiraRBHSTT11,
  author       = {Marcos R. Vieira and
                  Humberto Luiz Razente and
                  Maria Camila Nardini Barioni and
                  Marios Hadjieleftheriou and
                  Divesh Srivastava and
                  Caetano Traina Jr. and
                  Vassilis J. Tsotras},
  title        = {On query result diversification},
  booktitle = {ICDE},
  pages        = {1163--1174},
  year         = {2011},
  doi          = {10.1109/ICDE.2011.5767846},
  bibsource    = {dblp computer science bibliography, https://dblp.org}
}

@inproceedings{DBLP:conf/sigir/Sakai07,
  author       = {Tetsuya Sakai},
  title        = {Alternatives to Bpref},
  booktitle    = {{SIGIR}},
  pages        = {71--78},
  year         = {2007},
  doi          = {10.1145/1277741.1277756},
  bibsource    = {dblp computer science bibliography, https://dblp.org}
}

@inproceedings{DBLP:conf/kdd/DhillonMK02,
  author       = {Inderjit S. Dhillon and
                  Subramanyam Mallela and
                  Rahul Kumar},
  title        = {Enhanced word clustering for hierarchical text classification},
  booktitle    = {{SIGKDD}},
  pages        = {191--200},
  year         = {2002},
  doi          = {10.1145/775047.775076},
  bibsource    = {dblp computer science bibliography, https://dblp.org}
}

@article{DBLP:journals/tit/Lin91,
  author       = {Jianhua Lin},
  title        = {Divergence measures based on the Shannon entropy},
  journal      = {{IEEE} Trans. Inf. Theory},
  volume       = {37},
  number       = {1},
  pages        = {145--151},
  year         = {1991},
  doi          = {10.1109/18.61115},
  bibsource    = {dblp computer science bibliography, https://dblp.org}
}

@article{DBLP:journals/pvldb/FanWLZM23,
  author       = {Grace Fan and
                  Jin Wang and
                  Yuliang Li and
                  Dan Zhang and
                  Ren{\'{e}}e J. Miller},
  title        = {Semantics-aware Dataset Discovery from Data Lakes with Contextualized
                  Column-based Representation Learning},
  journal      = {Proc. {VLDB} Endow.},
  volume       = {16},
  number       = {7},
  pages        = {1726--1739},
  year         = {2023},
  doi          = {10.14778/3587136.3587146},
  bibsource    = {dblp computer science bibliography, https://dblp.org}
}

@article{DBLP:journals/tods/Codd79,
  author       = {E. F. Codd},
  title        = {Extending the Database Relational Model to Capture More Meaning},
  journal      = {{ACM} Trans. Database Syst.},
  volume       = {4},
  number       = {4},
  pages        = {397--434},
  year         = {1979},
  doi          = {10.1145/320107.320109}
}

@inproceedings{DBLP:conf/sigmod/Fan00M23,
  author       = {Grace Fan and
                  Jin Wang and
                  Yuliang Li and
                  Ren{\'{e}}e J. Miller},
  title        = {Table Discovery in Data Lakes: State-of-the-art and Future Directions},
  booktitle    = {{SIGMOD}},
  pages        = {69--75},
  year         = {2023},
  doi          = {10.1145/3555041.3589409},
}

@article{DBLP:journals/vldb/ChapmanSKKIKG20,
  author       = {Adriane Chapman and
                  Elena Simperl and
                  Laura Koesten and
                  George Konstantinidis and
                  Luis{-}Daniel Ib{\'{a}}{\~{n}}ez and
                  Emilia Kacprzak and
                  Paul Groth},
  title        = {Dataset search: a survey},
    journal      = {Proc. {VLDB} Endow.},
  volume       = {29},
  number       = {1},
  pages        = {251--272},
  year         = {2020},
  doi          = {10.1007/S00778-019-00564-X},
}

@article{DBLP:journals/ior/RaviRT94,
  author       = {S. S. Ravi and
                  Daniel J. Rosenkrantz and
                  Giri Kumar Tayi},
  title        = {Heuristic and Special Case Algorithms for Dispersion Problems},
  journal      = {Oper. Res.},
  volume       = {42},
  number       = {2},
  pages        = {299--310},
  year         = {1994},
  doi          = {10.1287/OPRE.42.2.299},
  bibsource    = {dblp computer science bibliography, https://dblp.org}
}

@inproceedings{DBLP:conf/doceng/KassaieKT25,
  author       = {Besat Kassaie and
                  Andrew Kane and
                  Frank Wm. Tompa},
  title        = {Exploiting Query Reformulation and Reciprocal Rank Fusion in Math-Aware
                  Search Engines},
  booktitle    = {{DocEng}},
  pages        = {7:1--7:10},
  year         = {2025},
  doi          = {10.1145/3704268.3742687},
  bibsource    = {dblp computer science bibliography, https://dblp.org}
}

@article{rousseeuw1987silhouettes,
  title={Silhouettes: a graphical aid to the interpretation and validation of cluster analysis},
  author={Rousseeuw, Peter J},
  journal={J. Comput. Appl. Math.},
  volume={20},
  pages={53--65},
  year={1987}
}

@misc{behrouz_babaki_2017_831850,
  author       = {Behrouz Babaki},
  title        = {COP-Kmeans version 1.5},
  year         = 2017,
  doi          = {10.5281/zenodo.831850}
}

@inproceedings{DBLP:conf/icml/WagstaffCRS01,
  author       = {Kiri Wagstaff and
                  Claire Cardie and
                  Seth Rogers and
                  Stefan Schr{\"{o}}dl},
 
  title        = {Constrained K-means Clustering with Background Knowledge},
  booktitle    = {{ICML}},
  pages        = {577--584},
  year         = {2001},
  bibsource    = {dblp computer science bibliography, https://dblp.org}
}

@article{DBLP:journals/corr/abs-2004-00584,
  author       = {Yuliang Li and
                  Jinfeng Li and
                  Yoshihiko Suhara and
                  AnHai Doan and
                  Wang Chiew Tan},
  title        = {Deep Entity Matching with Pre-Trained Language Models},
  journal      = {CoRR},
  volume       = {abs/2004.00584},
  year         = {2020},
  eprinttype    = {arXiv},
  eprint       = {2004.00584},
  bibsource    = {dblp computer science bibliography, https://dblp.org}
}

@inproceedings{DBLP:conf/edbt/BrunnerS20,
  author       = {Ursin Brunner and
                  Kurt Stockinger},
  title        = {Entity Matching with Transformer Architectures - {A} Step Forward
                  in Data Integration},
  booktitle    = {{EDBT}},
  pages        = {463--473},
  year         = {2020},
  doi          = {10.5441/002/EDBT.2020.58},
  bibsource    = {dblp computer science bibliography, https://dblp.org}
}

\ifnottechreport{}{\newpage
\appendix
\subsection{Proofs}\label{apx:proofs}
In this section, we present two proofs supporting the theorems made in the main text. 
\newenvironment{claimbdup}{\par\noindent\textbf{Theorem~\ref{clm:bdupl}.} \itshape}{\par}

\begin{claimbdup}
The Search Novelty Score  (Definition~\ref{df:nscore_resultset}) satisfies the blatant duplicate axiom and the dilution axiom.
\end{claimbdup}
\begin{proof}
Let $T$ be the search result table constructed as described in Definition~\ref{df:nscore_resultset}. For each tuple $t \in T$, we incorporate its tuple novelty computed according to Definition~\ref{def:novelty_support} into the overall novelty score of the result, as defined in Equation~\eqref{eq:table_nov_support}. By construction, all newly added tuples originate from the  query table have already an identical counter part tuple in $T$,  consequently:
\begin{compactitem}
   \item[i)] For newly introduced tuples, all resulting $N(t)$ evaluate to zero, since all tuple pairs involve an identical counter part tuple from the query table already in $T$: an increase in the denominator ($y$), with the numerator remaining unchanged in Equation~\eqref{eq:table_nov_support}.
   \item[ii)] For existing tuples, in fact  we add already-accounted-for tuples therefore this does not affect the minimum in (Definition~\ref{def:novelty_support}) while it increases the number of tuples  in (Definition~\ref{def:table_nov_support}) (i.e., a larger $y$). 
\end{compactitem}
Consequently, the addition of query tuples leads to a lower nscore. Similarly, we can show that the $nscore$ function (Definition~\ref{df:nscore_resultset}) satisfies the Dilution axiom.
\end{proof}



\newenvironment{claimbnp}{\par\noindent\textbf{Theorem~\ref{clm:nphard}.} \itshape}{\par}

\begin{claimbnp}
   Finding on optimal $R$ in Definition~\ref{def:objective_main} is NP-Hard. 
\end{claimbnp} 
\begin{proof}
If we represent the set $S$ as a complete graph, the problem resembles the induced subgraph problem with a property $\Pi$ and size $l$ \cite{DBLP:books/fm/GareyJ79}, where $\Pi$ in our case requires that the subgraph of size $l$ has the maximum nscore compared to all other subgraphs of the same size. Since  verifying  property $\Pi$ is exponential, our problem is NP-Hard. 
\end{proof}

\subsection{Jensen-Shannon divergence}\label{apx:JSD_desc}
We  briefly review  the Jensen-Shannon divergence~\cite{DBLP:conf/kdd/DhillonMK02}. Let the attributes $B_j$ and $A_i$ be  discrete random variables that take on values from the set $\mathcal{D}$ with probability distribution $\mathbb{B}_i$ and $\mathbb{A}_j$  respectively. The relative entropy or Kullback-Leibler (KL) divergence  between tow distributions $\mathbb{B}_i$ and $\mathbb{A}_j$ is computed as: 
\vspace{-0.5em}
$$\mathrm{KL}(\mathbb{B}_i,\mathbb{A}_j) \! =\! \sum_{v \in \mathcal{D}}^{} \mathbb{B}_i(v) \times \log \frac{\mathbb{B}_i(v)}{\mathbb{A}_j(v)} $$
The Kullback-Leibler (KL) divergence has several issues that make it unsuitable for our use case. It is not a proper distance measure as it is asymmetric and does not satisfy the triangle inequality. Additionally, it is unbounded for missing values: if $\mathbb{B}_i(v)\!=\!0 $ and $\mathbb{A}_j(v)\! \neq \! 0$, $KL(\mathbb{B}_i,\mathbb{A}_j) \!= \!-\infty$, making it impractical for our problem where probability mass can be missing in either of distributions.
To address these limitations, we consider the Jensen-Shannon divergence (JSd), which is derived from KL divergence but has desirable properties: it is symmetric, bounded within $[0,1]$, and can serves as a proper distance metric:
 \vspace{-0.5em}
$$\mathrm{JSd}(\mathbb{B}_i, \mathbb{A}_j)\!=\!\frac{1}{2} \times (\mathrm{KL}(\mathbb{B}_i, \mathbb{M}), \mathrm{KL}(\mathbb{A}_j, \mathbb{M}))$$ 
where $\mathbb{M}\! = \!\frac{\mathbb{B}_i+\mathbb{A}_j}{2}$ is a mixture distribution of  $\mathbb{B}_i$ and $\mathbb{A}_j$.  The Jensen-Shannon distance (JSD) between two probability distribution is the square root of JSd: 
\vspace{-0.5em}
$$\mathrm{JSD}(\mathbb{B}_i, \mathbb{A}_j)\!=\!\sqrt{\mathrm{JSd}(\mathbb{B}_i, \mathbb{A}_j)}$$

\subsection{Ablation Study: Effects of Dataset Scarcity}\label{apdx:ablationstudey}

In this section, we study the effect of dataset scarcity on our results.
\subsubsection{Syntactic Novelty Measure}
A closer examination reveals that, in the Ugen-v2 small dataset, 80\% of the queries for which \estimatename performs poorly at $l\!=\!3$ have only four unionable tables in the ground truth, two of which are blatant duplicates of the query table (including its diluted version). These low-quality matches negatively impact the SNM score. This scarcity of unionable tables contributes to the observed performance degradation. To test this hypothesis, we further restricted the Ugen-v2 small dataset to queries with at least 12 unionable tables, which resulted in trends consistent with the other datasets (Figure~\ref{fig:snmfig3_filtered}). 

\subsubsection{Simplified Syntactic Novelty Measure}
Again, for the Ugen-v2 Small dataset, an unusual drop in the performance of \estimatename is observed, is attributed to the limited number of unionable tables. To validate this hypothesis, we report SSNM scores for queries with at least 12 unionable tables. The results, shown in Figure~\ref{fig:ssnmfig3extendedand_filtered}, support our assumption and confirm the impact of data sparsity on performance.

\begin{figure}[!t]
  \centering

  \begin{subfigure}[t]{0.48\linewidth}
    \centering
    \includegraphics[width=\linewidth]{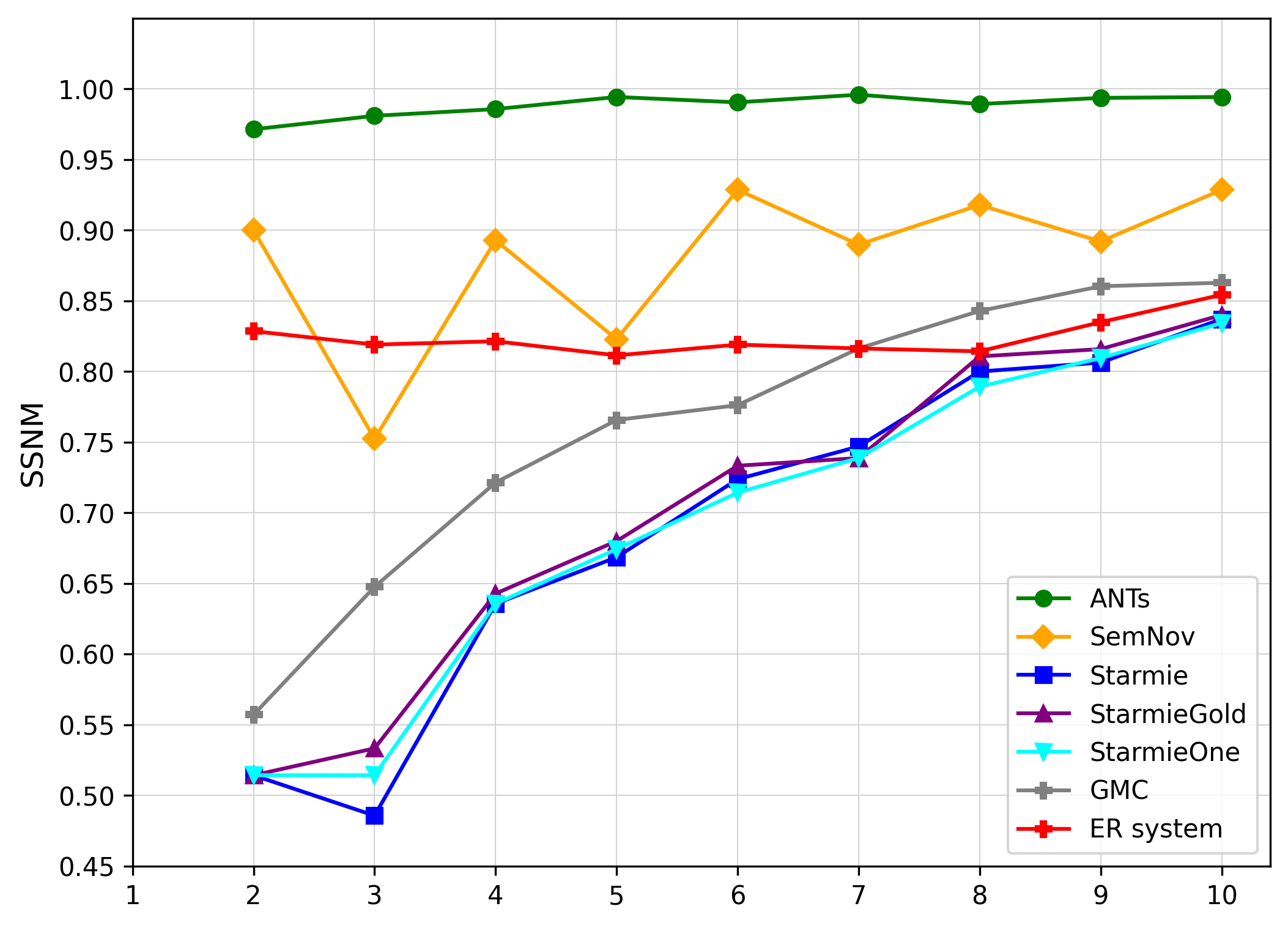}
    \caption{Santos (SSNM)}
    \label{fig:ssnm_santosex}
  \end{subfigure}
  \hfill
  \begin{subfigure}[t]{0.48\linewidth}
    \centering
    \includegraphics[width=\linewidth]{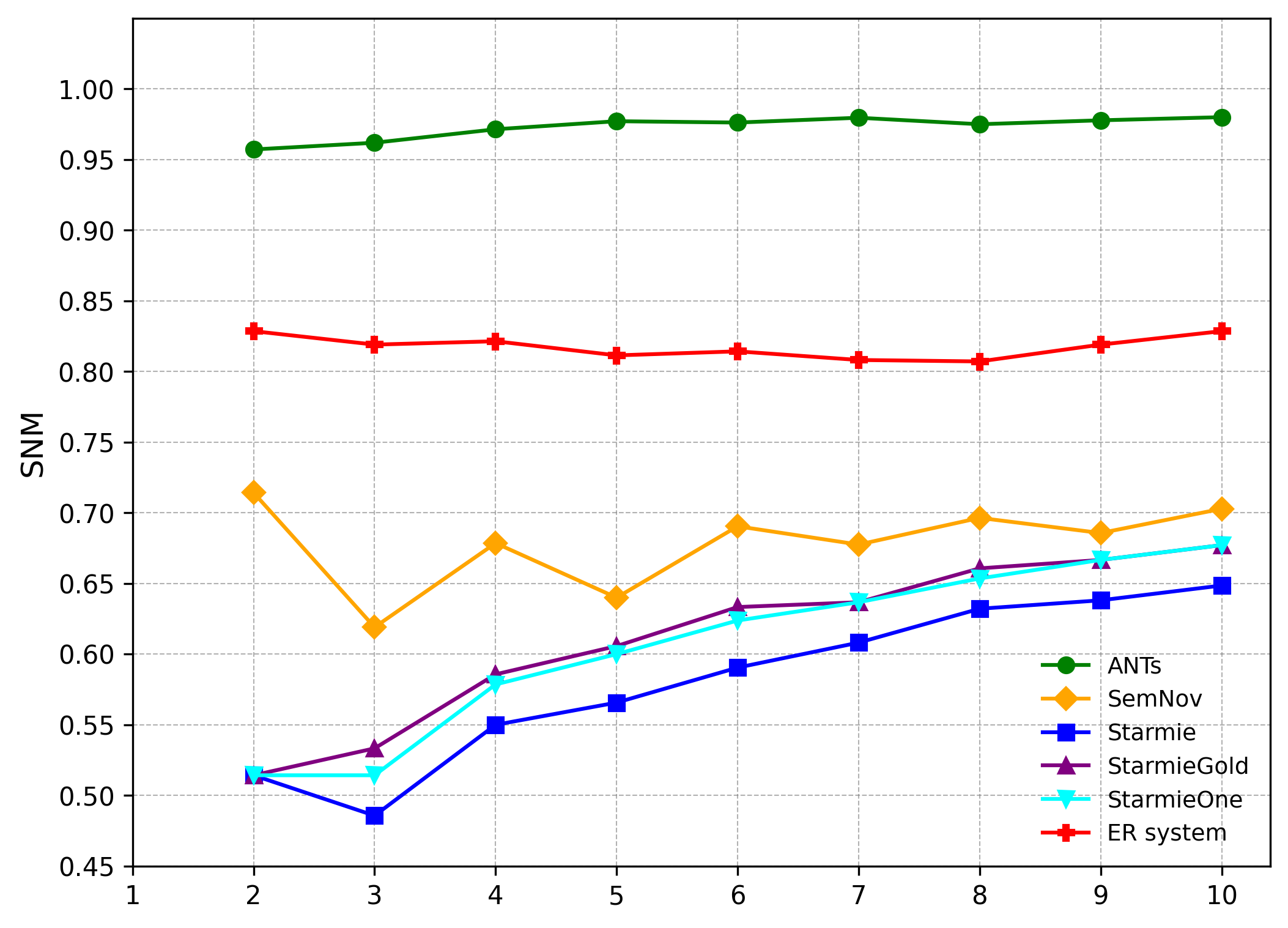}
    \caption{Santos (SNM)}
    \label{fig:snm_santosex}
  \end{subfigure}

  \vspace{0.6em}

  \begin{subfigure}[t]{0.48\linewidth}
    \centering
    \includegraphics[width=\linewidth]{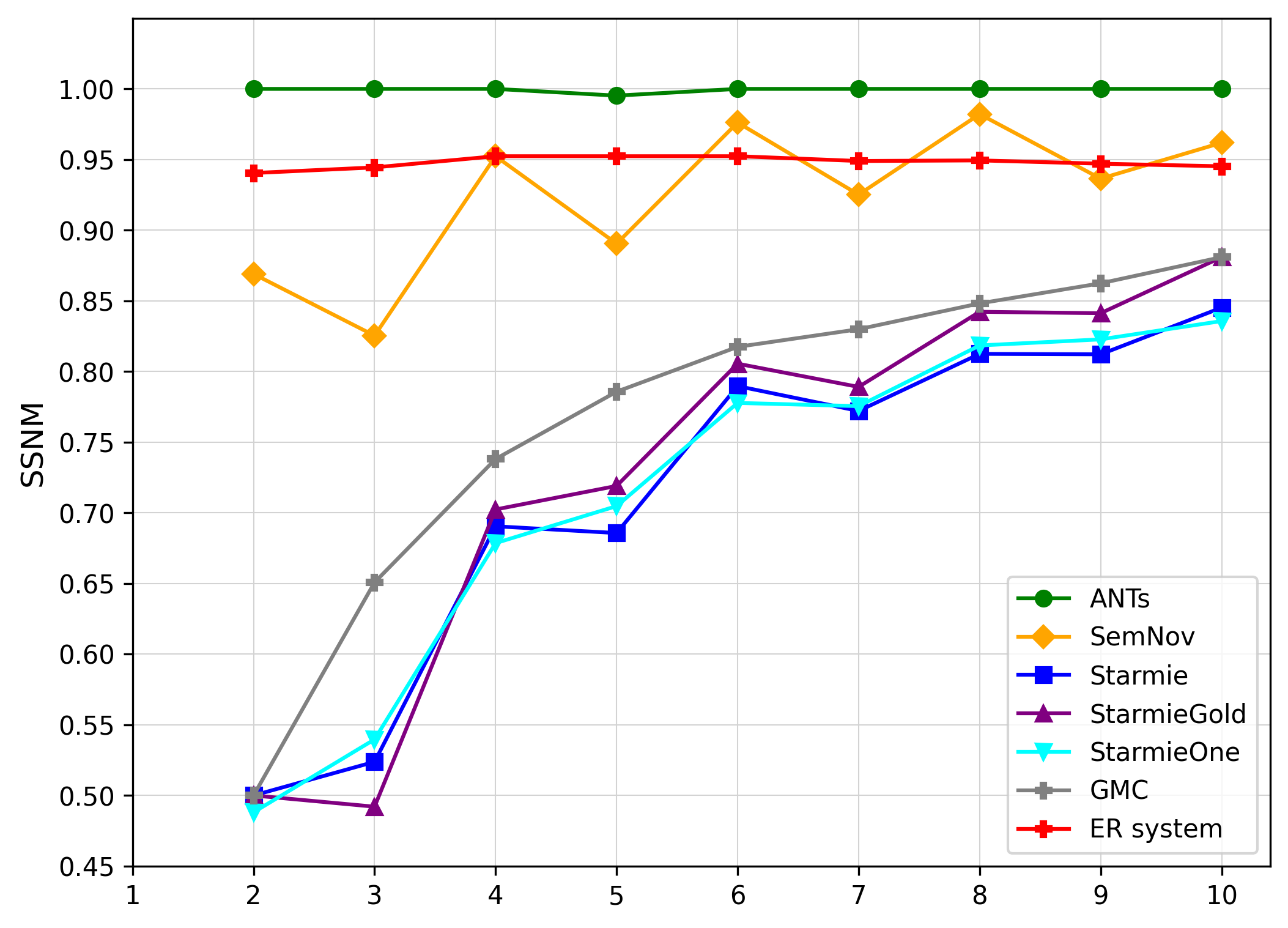}
    \caption{TUS (SSNM)}
    \label{fig:ssnm_tusex}
  \end{subfigure}
  \hfill
  \begin{subfigure}[t]{0.48\linewidth}
    \centering
    \includegraphics[width=\linewidth]{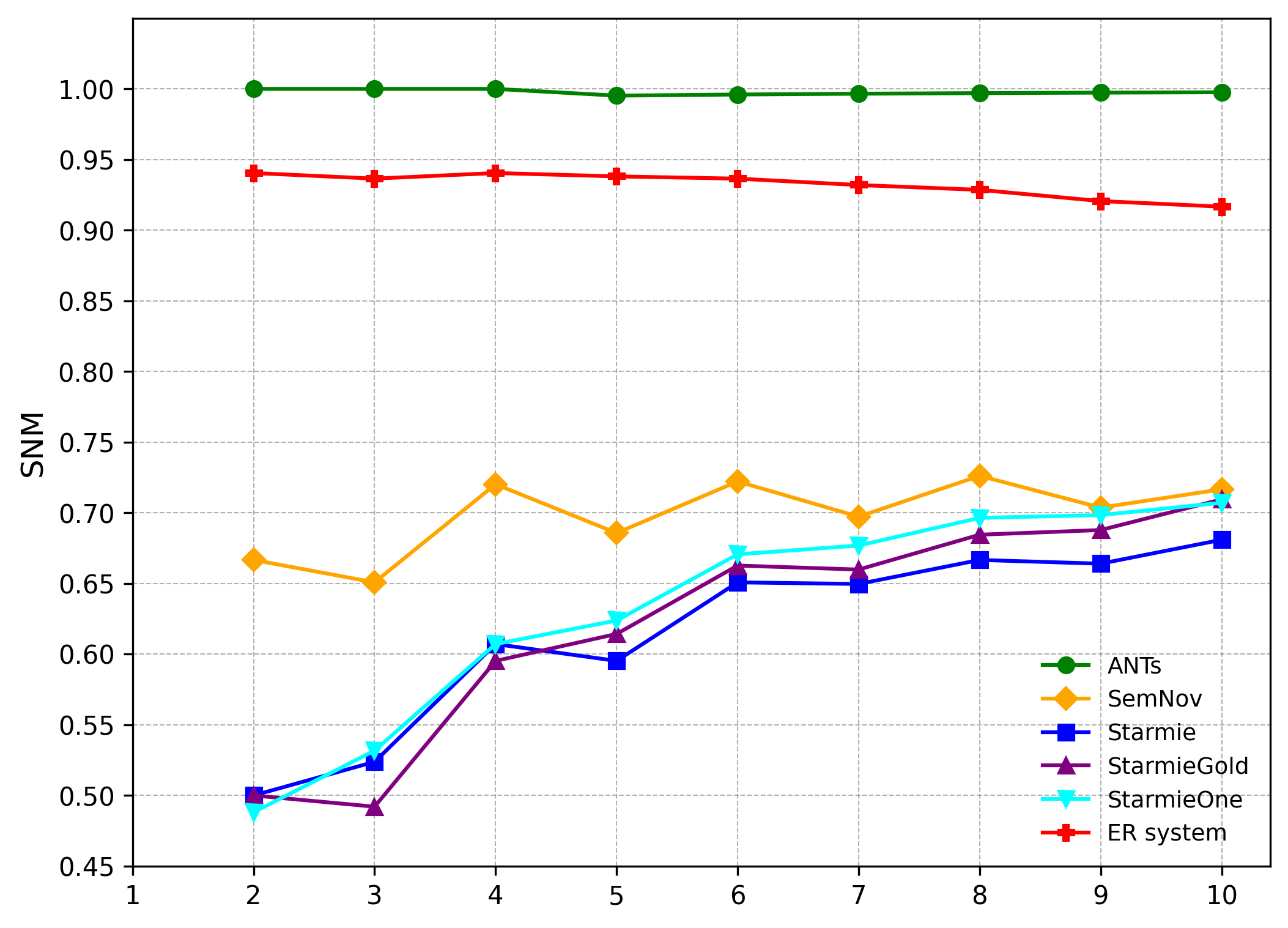}
    \caption{TUS (SNM)}
    \label{fig:snm_tusex}
  \end{subfigure}

  \vspace{0.6em}

  \begin{subfigure}[t]{0.48\linewidth}
    \centering
    \includegraphics[width=\linewidth]{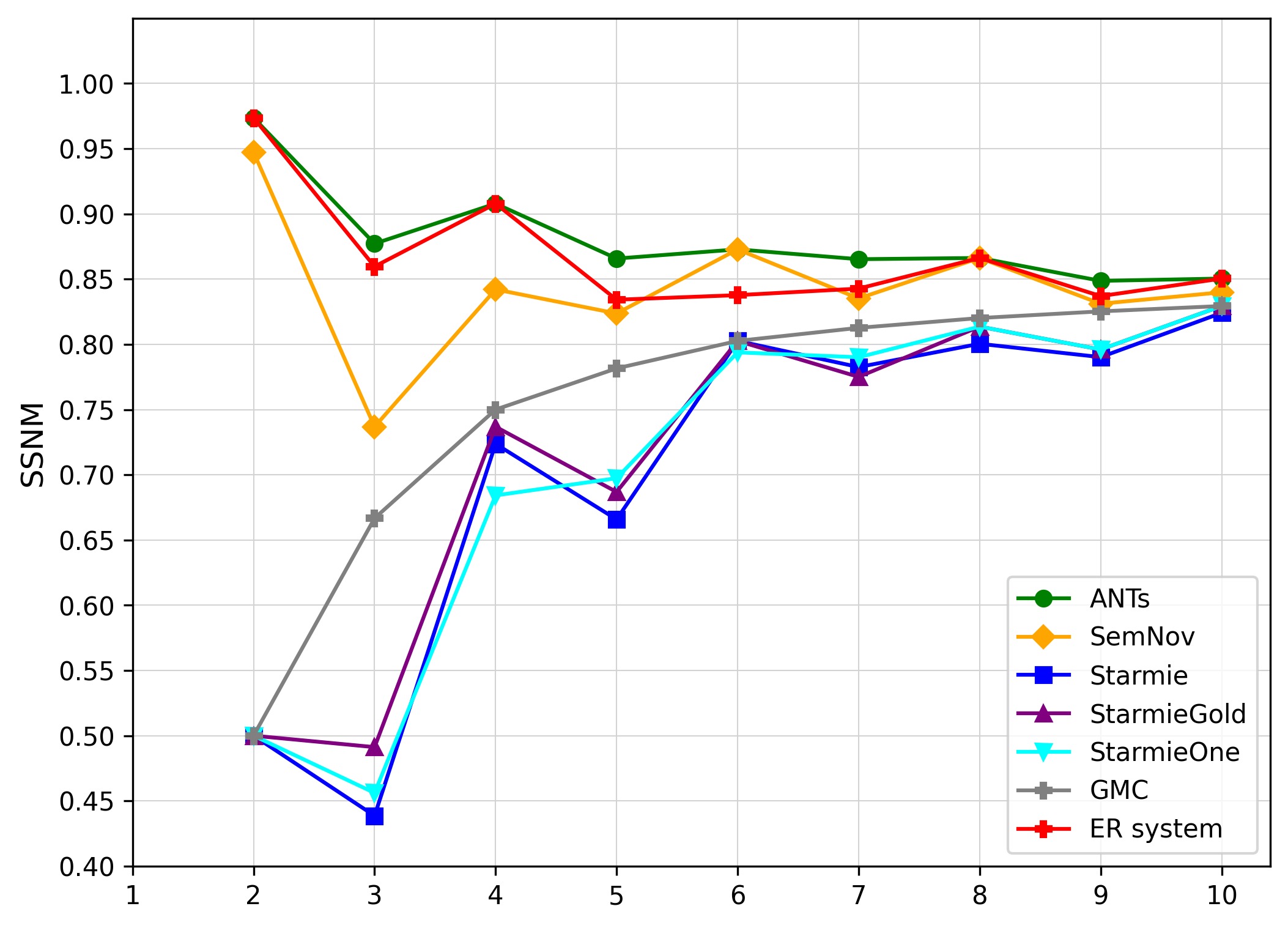}
    \caption{Ugen-v2 (SSNM)}
    \label{fig:ssnm_ugenex}
  \end{subfigure}
  \hfill
  \begin{subfigure}[t]{0.48\linewidth}
    \centering
    \includegraphics[width=\linewidth]{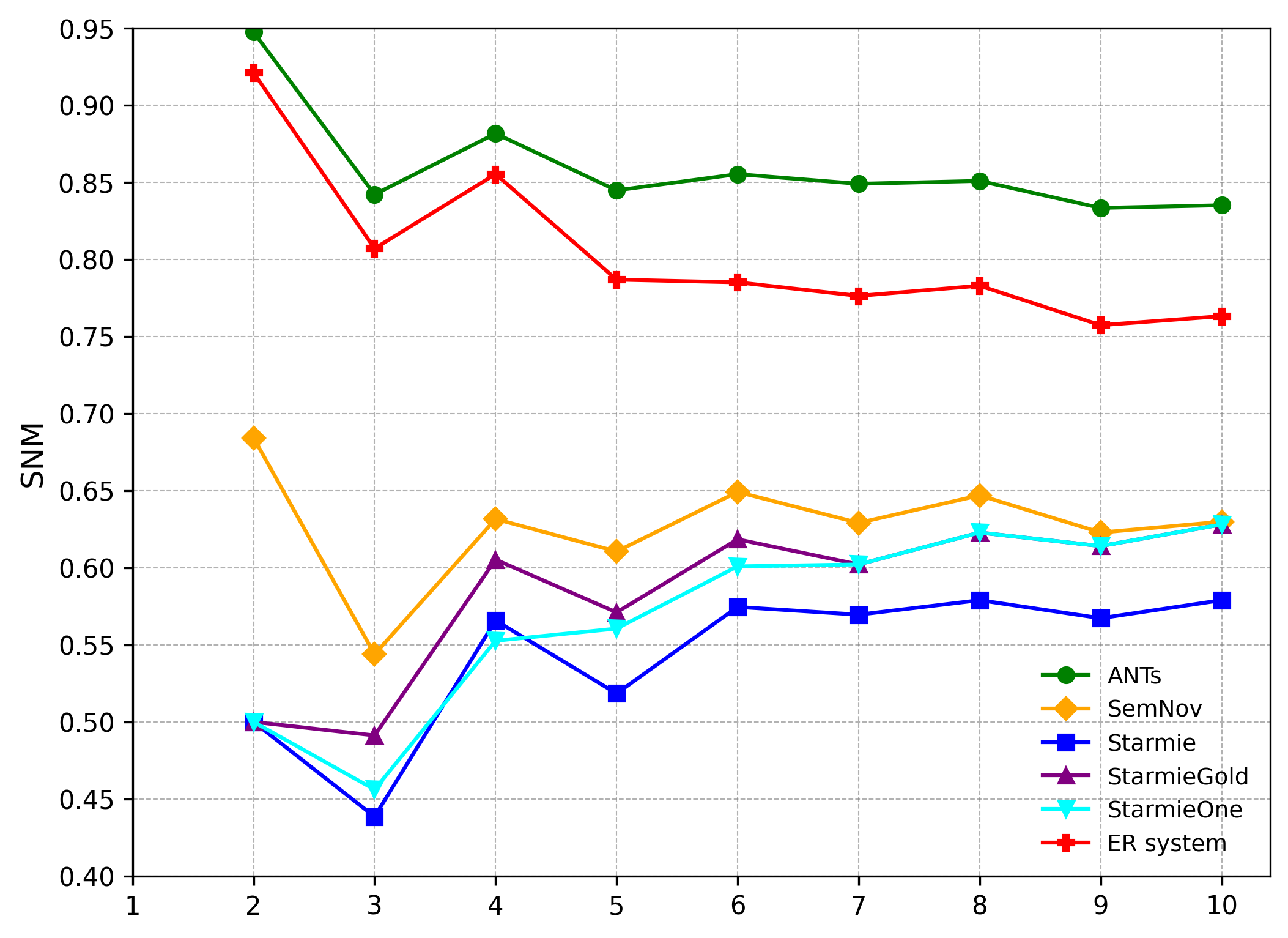}
    \caption{Ugen-v2 (SNM)}
    \label{fig:snm_ugenex}
  \end{subfigure}

  \caption{\small Comparison of SNM (right column) and SSNM (left column) across the Santos, TUS, and Ugen-v2 Small datasets.}
  \label{fig:snm_ssnm_grid_twocol}
\end{figure}
\subsection{Entity Resolution for Novel Table Search}\label{apnx:entityResultion}

\estimatename adopts an attribute-based strategy to reduce the time complexity of the novel table search problem (Definition~\ref{def:objective_main}). Alternatively, one could consider a tuple-based approximation approach. For instance, entity resolution (also known as record linkage or entity matching) techniques could be applied to estimate the degree of overlap, that is, the number of identical entities, between a query table and candidate tables in the data lake, and then rank the tables such that those with less overlap receive a higher score. Next, we study the application of entity resolution techniques in solving the novel table search problem. 
  
Entity resolution focuses on determining whether two tuples refer to the same real-world entity.
 In this line of work, methods are developed   to identify pairs of matching entities from two or more data sources, such as tables~\cite{DBLP:journals/vldb/LiLSDT23, DBLP:journals/corr/abs-2004-00584, DBLP:conf/edbt/BrunnerS20,DBLP:journals/pvldb/Thirumuruganathan21,DBLP:journals/pvldb/KondaDCDABLPZNP16}. Exhaustive entity resolution incurs quadratic time complexity with respect to the number of tuples. To improve scalability, blocking techniques are commonly adopted to cluster similar tuples and limit comparisons to candidate pairs within each block, effectively reducing the computational overhead during the matching phase.
   In performing entity resolution, it is necessary to determine the appropriate blocking and matching approaches to employ. When a query table arrives, it first undergoes a blocking phase, where similar tuples between the query table and each unionable table are identified and grouped into blocks. The subsequent matching step is then restricted to tuple pairs that fall within the same block.
  
\noindent \textbf{Blocking}. Recently, \cite{DBLP:journals/vldb/NeuhofFPNANK24} conducted an extensive experimental study of various blocking techniques across multiple well-known datasets, comparing their performance along several dimensions. We base our choice of blocking method on their findings to select the most appropriate technique for our setting. In selecting blocking techniques, we consider the following key characteristics. 
We favor self-supervised blocking strategies that require neither domain knowledge nor labeled data, as both are impractical in data lake settings due to the absence of metadata and the massive scale of tuples. Consequently, we operate in a \textit{schema-agnostic} setting that does not depend on determining suitable attributes for blocking; instead, each tuple is represented as tokens. Most importantly, the technique must be computationally efficient, particularly in the context of reranking unionable table search results, which is inherently an online process. For instance, we employed the Autoencoder component from the DeepBlocker framework, which has been shown to be effective in prior studies~\cite{DBLP:journals/pvldb/Thirumuruganathan21, DBLP:journals/vldb/NeuhofFPNANK24}. This approach aligns well with our setup, as it does not require labeled data. However, despite its effectiveness in achieving high precision, we found it to be computationally expensive. In our tests on the Santos dataset, it was proved to be extremely slow and was therefore excluded from our experiments. Specifically, we encountered two major issues: the process consistently required more than our available 512 GB of RAM due to its high memory consumption, and it was very slow, taking over two hours to complete the blocking step for a single query, which led us to halt its execution.
The most efficient and effective approach proposed by Neuhof et al. ~\cite{DBLP:journals/vldb/NeuhofFPNANK24} is Random Join, for which no publicly available implementation could be found. 
Furthermore, we focus on the techniques that does not require non-trivial fine tuning and configurations. Considering all these aspects, we restrict our analysis to two blocking techniques for which publicly available implementations exist:
\begin{enumerate}
    \item k-Nearest-Neighbor Join (\blocktwo)~\cite{DBLP:journals/vldb/SilvaALPA13}
\blocktwo was extremely slow; we terminated the KNN (k = 5) run after 16.85 hours, by which time it had processed only four query tables on the Santos dataset (averaging 4.2 hours per table).
    \item \blockone~\cite{DBLP:conf/icml/GuoSLGSCK20}  This technique was the fastest, required minimal parameter tuning, and was less memory-intensive among the available implementations, which is why we selected it.
\end{enumerate}

\iftechreport
\subsection{Normalization} \label{subsec:Normalization}

  We describe the syntactic normalization processes that were applied to attributes at two stages of our work. We apply this normalization process to attribute values prior to  domain size estimation and syntactic similarity computation. Note that attribute normalization was not performed during the fine-tuning or inference stages of the \Starmie model.

In the absence of any information about how the tables were created, we treat all tables and attributes uniformly; therefore, the semantics of the normalization function remain identical across all attributes. Given an arbitrary table $T \in$ \Tschema, for every tuple $t \in  T $ where $t= \langle v_1,\dots ,v_n \rangle$, and for every $v_k \in t$, The following steps are performed in the sequence outlined below:
\begin{compactitem}
\item[i)]\textit{Tokenization.}
$v_k$ is tokenized, i.e.,  splitted using spaces, periods, underscores, and dashes as delimiters. We used NLTK’s off-the-shelf regular-expression tokenizers.\footnote{\url{https://www.nltk.org/api/nltk.tokenize.regexp.html}}
\item [ii)]\textit{Lexical Standardization.}
The tokens are then converted to lowercase and processed using stemming~\cite{DBLP:journals/program/Porter80}.
\item [iii)]\textit{Reconstruction.}
Finally, the normalized tokens are merged back together, preserving their original order and using space character as separators.
\end{compactitem}

For example, applying three steps, the string “IT-Hardware Purchases” is normalized to “it hardwar purchas”. The input to the normalization function is a string, and the output is also a string. When computing the Jensen-Shannon distance, we represent the attribute as a multiset of normalized strings to preserve the frequencies; however, for Jaccard similarity computations, we represent the attribute as a set of normalized strings. 

\begin{table}[!htbp]
\caption{\small Each cell displays the average percentage (\%) of queries in each dataset that yield a blatant duplicate within the top $l$ results, average computed over the range $l \in [2,10]$ and across all queries in each dataset.}
\label{tab:baltant_dup_avg_extended}
\vspace{-1em}
\centering
{
\setlength{\tabcolsep}{.15pt} 
\renewcommand{\arraystretch}{1.25}  
\begin{tabular}{rcccccc>{\columncolor{white}}c} 
 	& \Starmie & \StarmieGold & \StarmieOne & \GMC & \ERNov & \semnov & \estimatename \\
\hline
\textit{Santos}  & 99.4 & 99 & 99 &81.9 & 8.6 &\textbf{0} & \textbf{0} \\
\textit{TUS}     & 100 & 100 & 100 & 99.5 & \textbf{0}&\textbf{0} & \textbf{0} \\
\textit{Ugen-v2} & 100 & 100 & 100 & 98.8 & \textbf{43.6} & {52.6} & {43.9} \\
\hline
\end{tabular}
}
\end{table}

\begin{table}[!htbp]
\caption{\small Each cell presents the average search novelty score across  queries of Ugen\_v2 small.} \label{tb:nscore_extended}

\centering 
{
\setlength{\tabcolsep}{2pt} 
\renewcommand{\arraystretch}{1.25}  
\begin{tabular}{l|cccccc>{\columncolor{lightblue}}c} 
 	& \Starmie  &\StarmieGold & \StarmieOne & \GMC & \ERNov & \semnov & \estimatename \\
 \hline
$l=2$   & 0.0 & 0.0 & 0.0& 0.0040 & 0.3715& \underline{0.3840}  & \textbf{0.3911} \\
\hline

$l=3$  & 0.0074 & 0.0076  & 0.0074 & 0.0066 & 0.2387 & \underline{0.2417}  & \textbf{0.2478} \\
\hline
\end{tabular}

}\end{table}

\subsection{Extended Results}\label{apndx: effectiveResults}
We extended our experiments to include two additional baselines, \StarmieOne and \StarmieGold. In this section, we revisit the previously reported results, now incorporating these two baselines into the comparison.
\StarmieGold is similar to \Starmie, but replaces its built-in alignment with ground-truth alignment. Attributes are represented using \Starmie embeddings and aligned based on ground-truth mappings. The unionability score is then computed as the sum of cosine similarities between aligned attribute pairs, and retrieved tables are ranked accordingly. We also introduce a variant called \StarmieOne, which
computes the unionability score as the average
of cosine similarities instead of the sum as done in \Starmie.  
Our analysis shows that the different variants of \Starmie{} have little to no impact on the evaluated metrics (see Figures~\ref{fig:snm_ssnm_grid_twocol} and~\ref{fig:extendedTime}, and Tables~\ref{tab:baltant_dup_avg_extended} and~\ref{tb:nscore_extended}).

\begin{figure} [!htbp]
    \centering
 \includegraphics[width=0.95\columnwidth]{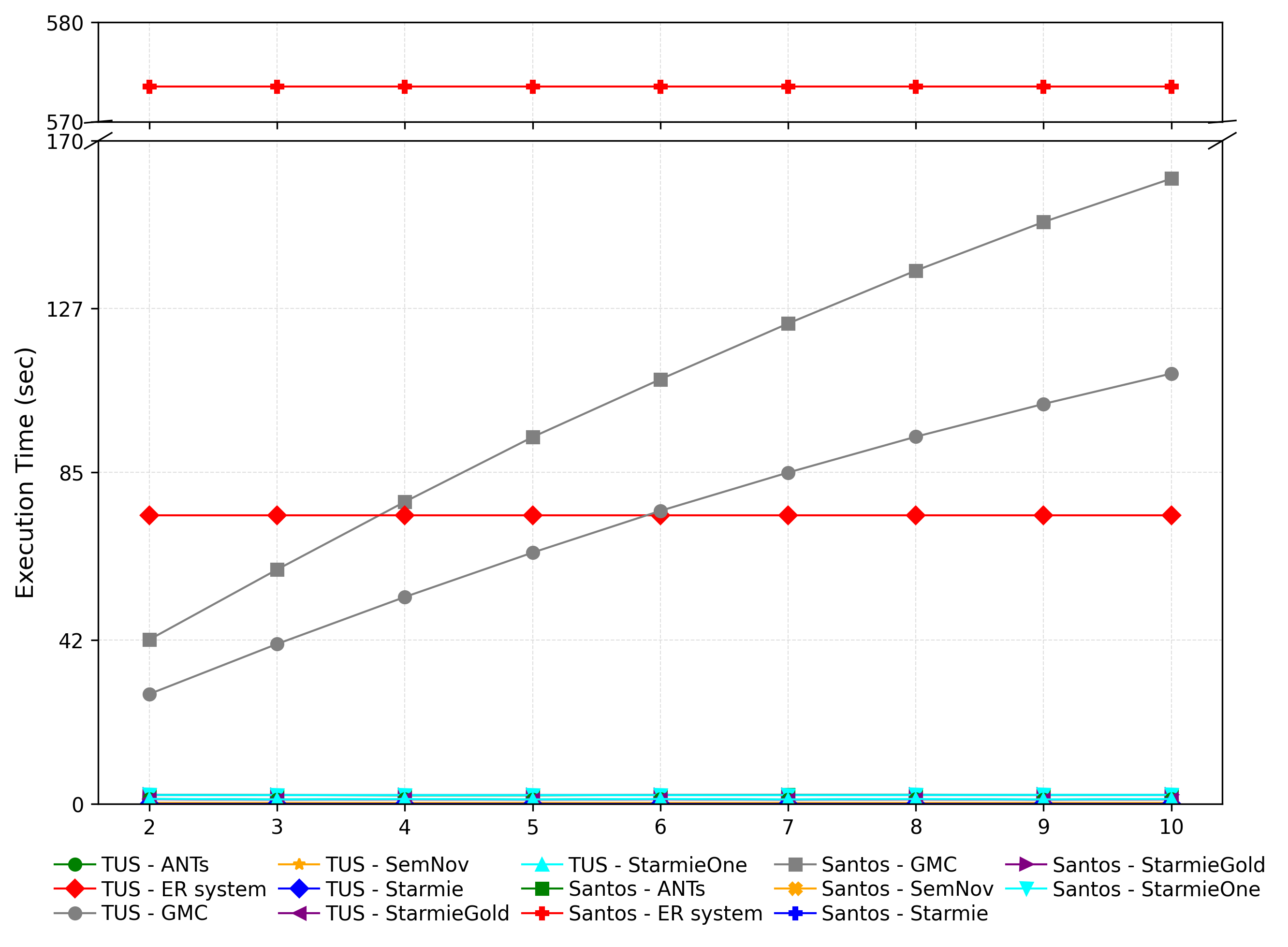}
    \caption{\small  The graph is the complete version of Figure~\ref{fig:exetime1}.}\label{fig:extendedTime}
\end{figure}
\subsection{Complete Results}\label{apndx: completeResults}
We  report the results for $\mathcal{F}$ values for all $l \in [2,10]$ across all datasets and systems in Table~\ref{tab:extendedofvalue}.
\begin{table}[!htbp]
\caption{\small Average percentage of queries for which \estimatename achieves equal or higher $\mathcal{F}$ values evaluated over $l$ values from 2 to 10. Cases where the percentage exceeds 50\% are highlighted. For readability, we omit the percentage sign and write each value as $X$ instead of $X\%$ in the table.}  \label{tab:extendedofvalue}
\vspace{-1em}
\centering 
{
\setlength{\tabcolsep}{3pt} 
\begin{tabular}{l|cccccccccc} 
 & 2	& 3	&4&5&6&7&8&9&10 \\
 \hline
$\textit{\textbf{Santos}}$\\
\Starmie &	\cellcolor{lightblue}68.6  & \cellcolor{lightblue}54.3  &\cellcolor{lightblue}54.3&48.6&\cellcolor{lightblue}54.3&\cellcolor{lightblue}51.4&48.6&45.7&48.6 \\
\StarmieGold&	\cellcolor{lightblue}68.6  & 48.6  &\cellcolor{lightblue}51.2&42.9&42.9&42.9&40&40&40 \\
\StarmieOne &	\cellcolor{lightblue}68.6  & 48.6  &\cellcolor{lightblue}54.3&48.6&\cellcolor{lightblue}57.1&\cellcolor{lightblue}51.4&48.6&48.6&\cellcolor{lightblue}51.4 \\
\GMC  &5.7 & 2.9 & 0&0&0&2.9&5.7&0&2.9\\
\ERNov  &\cellcolor{lightblue}65.7 & \cellcolor{lightblue}62.9& \cellcolor{lightblue}62.9&\cellcolor{lightblue}65.7&\cellcolor{lightblue}68.6&\cellcolor{lightblue}62.9&\cellcolor{lightblue}65.7&\cellcolor{lightblue}60&\cellcolor{lightblue}54.3\\
\semnov &\cellcolor{lightblue}94.3  & \cellcolor{lightblue}91.4  &\cellcolor{lightblue}91.4 &\cellcolor{lightblue}94.3&\cellcolor{lightblue}91.4&\cellcolor{lightblue}85.7&\cellcolor{lightblue}88.6&\cellcolor{lightblue}85.7&\cellcolor{lightblue}82.9 \\
\hline
$\textit{\textbf{TUS}}$\\
\Starmie &	\cellcolor{lightblue}69  & \cellcolor{lightblue}69  &\cellcolor{lightblue}66.7&\cellcolor{lightblue}66.7&\cellcolor{lightblue}66.7&\cellcolor{lightblue}66.7&\cellcolor{lightblue}66.7&\cellcolor{lightblue}66.7&\cellcolor{lightblue}66.7 \\
\StarmieGold &	\cellcolor{lightblue}69  & \cellcolor{lightblue}69   &\cellcolor{lightblue}69&\cellcolor{lightblue}61.9&\cellcolor{lightblue}59.5&\cellcolor{lightblue}57.1&42.9&35.7&38.1 \\
\StarmieOne &\cellcolor{lightblue}69   & \cellcolor{lightblue}69   &\cellcolor{lightblue}66.7&\cellcolor{lightblue}66.7&\cellcolor{lightblue}66.7&\cellcolor{lightblue}66.7 &\cellcolor{lightblue}69&\cellcolor{lightblue}71.4&\cellcolor{lightblue}69 \\
\GMC  &	14.3 & 2.4 & 2.4&0&2.4&4.8&4.8&2.4&2.4\\
\ERNov  &\cellcolor{lightblue}71.4 & \cellcolor{lightblue}64.3& \cellcolor{lightblue}61.9&\cellcolor{lightblue}66.7&\cellcolor{lightblue}81&\cellcolor{lightblue}78.6&\cellcolor{lightblue}88.1&\cellcolor{lightblue}85.7&\cellcolor{lightblue}73.8\\
\semnov &	\cellcolor{lightblue}83.3  & \cellcolor{lightblue}88.1  &\cellcolor{lightblue}88.1  &\cellcolor{lightblue}78.6&\cellcolor{lightblue}73.8&\cellcolor{lightblue}71.4&\cellcolor{lightblue}73.8&\cellcolor{lightblue}66.7&\cellcolor{lightblue}66.7\\
\hline
$\textit{\textbf{Ugen-v2 small}}$\\
\Starmie &5.3	 & 0 & 36.8&36.8&\cellcolor{lightblue}57.9&\cellcolor{lightblue}57.9&\cellcolor{lightblue}63.2&\cellcolor{lightblue}63.2&\cellcolor{lightblue}73.7\\
\StarmieGold &	\cellcolor{lightblue}52.3 & 0 & 36.8&31.6&\cellcolor{lightblue}57.9&\cellcolor{lightblue}57.9&\cellcolor{lightblue}63.2&\cellcolor{lightblue}63.2&\cellcolor{lightblue}68.4 \\
\StarmieOne &	5.3 & 0 & 26.3 &36.8&\cellcolor{lightblue}57.9&\cellcolor{lightblue}57.9&\cellcolor{lightblue}63.2&\cellcolor{lightblue}63.2&\cellcolor{lightblue}73.7\\
\GMC  &	0 & 0 & 31.6&36.8&\cellcolor{lightblue}57.9&\cellcolor{lightblue}57.9&\cellcolor{lightblue}63.2&\cellcolor{lightblue}63.2&\cellcolor{lightblue}73.7\\
\ERNov  &\cellcolor{lightblue}78.9 & \cellcolor{lightblue}89.5& \cellcolor{lightblue}84.2&\cellcolor{lightblue}84.2&\cellcolor{lightblue}84.2&\cellcolor{lightblue}89.5&\cellcolor{lightblue}84.2&\cellcolor{lightblue}89.5&\cellcolor{lightblue}84.2\\
\semnov & \cellcolor{lightblue}78.9  & \cellcolor{lightblue}78.9  &\cellcolor{lightblue}84.2 &\cellcolor{lightblue}78.9&\cellcolor{lightblue}94.7&\cellcolor{lightblue}84.2&\cellcolor{lightblue}84.2&\cellcolor{lightblue}84.2&\cellcolor{lightblue}89.5 \\
\end{tabular}
}
\end{table}



\subsection{\estimatename Versus \DUST}\label{sec:ANTsVersuDust}
We compare \estimatename with \DUST, a recent technique for unionable tuple search proposed by Khatiwada et al.\cite{KhatiwadaSM26}. We evaluate how novel the results returned by \DUST are by computing the search novelty score (Definition~\ref{df:nscore_resultset}), compare these scores to those of \estimatename. As reported in Section~\ref{sec:experiments}, computing the novelty score is prohibitively time-consuming, so we focus this analysis on the Ugen-v2 dataset. For both systems, we set the number of unionable tables in the initial search result to $k \!=\! 20$; For each query, we first obtain from \estimatename the most unionable table and denote its number of tuples by $m$. We then ask \DUST to retrieve the $m$ most unionable tuples for the same query. We compute and compare the novelty of the result tables using the \textit{nscore} metric. On average over all queries, the most novel tables returned by \estimatename contain $m \!=\! 23$ tuples. Table~\ref{tb:dust-ants_main} reports, for each method, the average novelty score, the average total number of retrieved tables and tuples required to attain that score, the resulting number of novel tuples, and the corresponding running time. While \DUST achieves a higher novelty score than \estimatename ($7.6$\% gain), it does so only after retrieving, on average, five times more tables and eighteen times more tuples than \estimatename, and incurring a substantially higher execution time. We therefore conclude that, in data-market settings where an acquisition cost is attached to each tuple, or in latency-sensitive scenarios, \estimatename is a much better choice.


\subsection{\estimatename for Machine Learning Task}\label{sec:ANTsMLTask}

The goal of data discovery is to prepare data for downstream tasks. In this section, we demonstrate how using our novelty-aware data discovery method, \estimatename, affects the performance of a machine learning model. We arbitrarily select a representative downstream task, i.e., rating prediction, and a freely available dataset, the IMDb Extensive Dataset.\footnote{\url{https://www.kaggle.com/datasets/simhyunsu/imdbextensivedataset}} We then train models under three configurations: no data discovery, relying only on the query table (\Baseline); augmenting the query table with tables returned by \Starmie and augmenting it with tables  returned by \Starmie and reranked by  \estimatename.\footnote{We first retrieve the top-$30$ tables using \Starmie and then apply \estimatename to rerank the result.} Next, we explain the experimental setup in detail.

\subsubsection{Data Preparation}
We use  movie metadata (title, year, genre, director, actors, duration, budget, etc.) where the average user rating for each movie is given by the \texttt{avg\_vote} attribute. The ML task is to train a regression model to predict \texttt{avg\_vote} from the metadata features. The IMDb movie table is randomly split into training (80\%) and test (20\%) datasets. From the training set, we randomly select 1\% of the tuples and partition them into three distinct query tables, each retaining all predictive attributes (features). We then construct a synthetic data lake from the remaining training tuples, following the procedure proposed by Nargesian et al.~\cite{DBLP:journals/pvldb/NargesianZPM18}. This process yields three query tables and $50$ unionable tables in the data lake. To ensure that the training dataset contains sufficient redundant information, we run our experiments using both the original data lake and a diluted variant with dilution degree $0.4$, where the query table is included in the data lake,  yielding a total of $101$ unionable tables. We rely on standard evaluation metrics for regression models namely, $R^2$ Scores (higher is better) and RMSE Scores (lower is better).\footnote{
We created and  evaluated all models using the standard  APIs provided by scikit-learn \url{https://scikit-learn.org}.}
\subsubsection{Model Selection} First, we select three well-known regression models: Light Gradient Boosting Machine (LGBM), Random Forest (RF), and XGBoost (XGB). For each query table, we train these models on the union of randomly selected tables from the non-diluted data lake (including the query table itself) and report performance averaged over the three queries. As shown in Table~\ref{tab:random_nondiluted_r2_rmse}, LGBM consistently outperforms the other models across all values of $k$ for both $R^2$ and RMSE. We hypothesize that this performance gain is due to the way LGBM internally handles \nullval, which is particularly advantageous in our setting, where the construction of the result table (Equation~\eqref{eq:resTable}) introduces  \nullval values. Therefore, we select LGBM as our regression model for comparing the performance of different table discovery techniques.
\begin{table}[!t]
\caption{R$^2$ and RMSE of all three models trained on the {Random} data selection from non-diluted data lake. The best value of each metric across models for each $k$ is shown in \textbf{boldface}. }
\label{tab:random_nondiluted_r2_rmse}
\centering

\setlength{\tabcolsep}{2pt}
\renewcommand{\arraystretch}{1.3}
\begin{tabular}{c
                ccc|
                ccc}
\hline
 & \multicolumn{3}{c}{R$^2$} & \multicolumn{3}{c}{RMSE} \\
 \hline
$k$ & LGBM & XGB & RF & LGBM & XGB & RF \\
\hline
2  & \textbf{0.4447} & 0.2714 & 0.2360 & \textbf{0.9194} & 1.0506 & 1.0752 \\
3  & \textbf{0.4671}& 0.2790 & 0.2524 & \textbf{0.9007} & 1.0454 & 1.0638 \\
4  & \textbf{0.4577} & 0.2694 & 0.2164 & \textbf{0.9086} & 1.0526 & 1.0878 \\
5  & \textbf{0.4554}& 0.2702 & 0.2077 & \textbf{0.9105} & 1.0516 & 1.0933 \\
6  & \textbf{0.4565} & 0.2697 & 0.2075 & \textbf{0.9095} & 1.0517 & 1.0930 \\
7  & \textbf{0.4366} & 0.2155 & 0.1567 & \textbf{0.9262} & 1.0923 & 1.1307 \\
8  & \textbf{0.4322} & 0.2553 & 0.2338 & \textbf{0.9299} & 1.0649 & 1.0802 \\
9  & \textbf{0.4335} & 0.2530 & 0.2358 & \textbf{0.9288} & 1.0666 & 1.0788 \\
10 & \textbf{0.4352} & 0.2559 & 0.2377 & \textbf{0.9274} & 1.0645 & 1.0774 \\
\hline
\end{tabular}
\end{table}

\begin{table*}[!t]
\caption{RMSE (lower is better) and $R^2$ (higher is better) by $k$ on the diluted dataset.  The values are averaged over the results for three query tables.
Gain is the relative change of \estimatename vs.\ Starmie. For each metric, the best value in a row is \textbf{boldfaced}. }
\label{tab:rmse_r2_merged_diluted}
\centering
\small
\setlength{\tabcolsep}{4pt}
\renewcommand{\arraystretch}{1.1}
\begin{tabular}{c|cccc|cccc}
\hline
\noalign{\vskip 0.25ex} 
 & \multicolumn{4}{c|}{RMSE} & \multicolumn{4}{c}{$R^2$} \\
$k$
 & \Baseline & \Starmie & \estimatename & Gain
 & \Baseline & \Starmie & \estimatename & Gain \\
\hline
\noalign{\vskip 0.1ex} 
 2 
 & 0.8975 & 0.8544 & \textbf{0.8373} & $-2.00\%$
 & 0.4711 & 0.5206 & \textbf{0.5396} & $+3.65\%$ \\
 3 
 & -- & 0.8553 & \textbf{0.8324} & $-2.68\%$
 & -- & 0.5196 & \textbf{0.5450} & $+4.89\%$ \\
 4 
 & -- & 0.8471 & \textbf{0.8364} & $-1.26\%$
 & -- & 0.5288 & \textbf{0.5406} & $+2.24\%$ \\
 5 
 & -- & 0.8507 & \textbf{0.8380} & $-1.49\%$
 & -- & 0.5248 & \textbf{0.5389} & $+2.69\%$ \\
 6 
 & -- & 0.8512 & \textbf{0.8383} & $-1.52\%$
 & -- & 0.5242 & \textbf{0.5385} & $+2.72\%$ \\
 7 
 & -- & 0.8524 & \textbf{0.8378} & $-1.71\%$
 & -- & 0.5229 & \textbf{0.5391} & $+3.10\%$ \\
 8 
 & -- & 0.8504 & \textbf{0.8381} & $-1.45\%$
 & -- & 0.5252 & \textbf{0.5388} & $+2.59\%$ \\
 9 
 & -- & 0.8514 & \textbf{0.8384} & $-1.53\%$
 & -- & 0.5240 & \textbf{0.5384} & $+2.74\%$ \\
 10 
 & -- & 0.8532 & \textbf{0.8405} & $-1.49\%$
 & -- & 0.5219 & \textbf{0.5361} & $+2.71\%$ \\
\hline
\end{tabular}
\end{table*}
\subsubsection{Results}
Tables~\ref{tab:rmse_r2_merged_diluted} and~\ref{tab:rmse_r2_merged} summarize the performance of LGBM on the diluted and non-diluted training sets. On the diluted dataset, where redundancy is guaranteed, \estimatename consistently outperforms both \Starmie and \Baseline for all values of $k$ and for both $R^2$ and RMSE. We assess significance using paired t-tests ($\textit{p-value}\!\leq\! 0.05$); improvements over \Baseline are significant;  all improvements over \Starmie are statistically significant except for $k{=}5$. On the non-diluted dataset, for both $\textit{R}^2$ and RMSE, \estimatename significantly outperforms \Baseline, but no significant difference is observed between \estimatename and \Starmie.

Overall, we conclude that in the presence of redundancy, reranking the results of \Starmie with \estimatename is clearly preferable. Even in the non-diluted setting, the reranking step is not detrimental; thus, the novelty-aware reranker can always be safely applied, even when the redundancy level of the dataset is unknown.

\begin{table*}[!t]
\caption{RMSE (lower is better) and $R^2$ (higher is better) by $k$ on the non-diluted dataset. The values are averaged over the results for three query tables. 
Gain is the relative change of \estimatename vs.\ Starmie. 
 For each metric, the best value in a row is \textbf{boldfaced}.}
\label{tab:rmse_r2_merged}
\centering
\small
\setlength{\tabcolsep}{4pt}
\renewcommand{\arraystretch}{1.1}
\begin{tabular}{c|cccc|cccc}
\hline
\noalign{\vskip 0.25ex} 
 & \multicolumn{4}{c|}{RMSE} & \multicolumn{4}{c}{$R^2$} \\
$k$
 & \Baseline & \Starmie & \estimatename & Gain
 & \Baseline & \Starmie & \estimatename & Gain \\
\hline
 2 
 & 0.8975 & 0.8378 & \textbf{0.8373} & $-0.06\%$
 & 0.4711 & 0.5391 & \textbf{0.5396} & $+0.09\%$ \\
 3 
 & -- & 0.8359 & \textbf{0.8324} & $-0.42\%$
 & -- & 0.5412 & \textbf{0.5450} & $+0.71\%$ \\
 4 
 & -- & \textbf{0.8359} & 0.8364 & $+0.06\%$
 & -- & \textbf{0.5411} & 0.5406 & $-0.10\%$ \\
 5 
 & -- & 0.8380 & \textbf{0.8380} & $0.00\%$
 & -- & \textbf{0.5389} & 0.5389 & $+0.01\%$ \\
 6 
 & -- & \textbf{0.8375} & 0.8383 & $+0.10\%$
 & -- & \textbf{0.5394} & 0.5385 & $-0.17\%$ \\
 7 
 & -- & \textbf{0.8378} & 0.8378 & $0.00\%$
 & -- & \textbf{0.5391} & 0.5391 & $-0.003\%$ \\
 8 
 & -- & 0.8388 & \textbf{0.8381} & $-0.08\%$
 & -- & 0.5380 & \textbf{0.5388} & $+0.14\%$ \\
 9 
 & -- & 0.8392 & \textbf{0.8394} & $+0.02\%$
 & -- & \textbf{0.5375} & 0.5373 & $-0.04\%$ \\
 10 
 & -- & 0.8396 & \textbf{0.8418} & $+0.26\%$
 & -- & \textbf{0.5371} & 0.5347 & $-0.44\%$ \\
\hline
\end{tabular}
\end{table*}

\subsection{Ablation Study: Impact of Alignment Quality}\label{subsec:alignmnetAbl}
In this study, we examine how alignment correctness affects the effectiveness of our methods by replacing the ground-truth alignment with an automatically generated one. 
Leveraging the codebase from a recent related work~\cite{KhatiwadaSM26}, we employ a clustering‐based approach to align attributes between a query table and data lake tables.\footnote{\url{https://github.com/northeastern-datalab/dust}} Similar to this work~\cite{KhatiwadaSM26}, our objective is to form clusters that each contain exactly one query column paired with its corresponding aligned attributes from the data lake tables. Based on domain knowledge, we enforce the constraint that no two columns from the same table may be clustered together (i.e., aligned to a single query attribute). For attribute  representations, we utilize the column‐level RoBERTa model, the best-performing model in~\cite{KhatiwadaSM26} and apply the Silhouette coefficient~\cite{rousseeuw1987silhouettes} to determine the optimal number of clusters. With a slight modification to their code, we adopt K-means clustering augmented with a \textit{cannot-link} constraint~\cite{DBLP:conf/icml/WagstaffCRS01} using Cosine distance, which yields improved results. The  K-means clustering algorithm’s codebase is freely available for download~\cite{behrouz_babaki_2017_831850}. We use this alignment method for \GMC, \semnov, \ERNov, and \estimatename.   \Starmie, however, includes its own built-in alignment procedure. We have alignment ground truth for all datasets and are therefore able to report the alignment algorithm’s precision, recall, and F-measure, as shown in Table~\ref{tb:dustalignemnt}. As anticipated, replacing manual alignment with an automatic alignment algorithm consistently harms performance. In terms of SSNM (Figure~\ref{fig:ssnm_vs_snm_all_dustAlign}), it reduces the performance of \estimatename, \semnov, and \GMC in 11 out of 12 (system, dataset) configurations (92\%). Similarly, it degrades the SNM of \estimatename and \semnov in 89\% of the configurations (Figure~\ref{fig:snm_dustalignment_filtered_}). However, even under these conditions, \estimatename remains the best-performing system in 48\% and 89\% of (dataset, $l$) settings for SSNM and SNM, respectively, when restricted to Ugen-v2 queries with more than 12 unionable tables.
Interestingly, \ERNov exhibits improved performance on the Santos dataset, for both SSNM and SNM.

\begin{table}[!htbp]
\caption{\small Effectiveness of the Alignment Algorithm.} \label{tb:dustalignemnt}

\centering \small
{
\setlength{\tabcolsep}{4.5pt} 
\renewcommand{\arraystretch}{1.1}  
\begin{tabular}{l|ccc} 
 	& Santos  &  TUS &Ugen-v2  \\
 \hline
Precision  & 0.80 &  0.88  & 0.96 \\

 Recall  & 0.81 &0.62 & 0.95  \\
  F-measure & 0.78 &0.70& 0.95  \\
\hline
\end{tabular}}
\end{table}
\begin{table}[!htbp]
\caption{\small  The average percentage of queries per dataset yielding a blatant duplicate within the top $l$ results ($l\!\in\![2,10]$), using automatic alignment. Lower values indicate better performance.  In all tables, \textbf{Bold} numbers show the best performance per row; \underline{underlin}\underline{ed} numbers indicate  the second best.}
\label{tab:baltant_dup_avg_auto_align}
\centering

{\setlength{\tabcolsep}{7pt} 
\renewcommand{\arraystretch}{1.1} 
\begin{tabular}{@{}r|ccccc} 
  & \Starmie  & \GMC & \ERNov & \semnov & \estimatename \\
\hline
\textit{Santos}          & 99.4\%  & 89.2\% & \underline{3.5\%} & {4.1\%} & {3.5\%} \\
\textit{TUS}             & 100\%   & 100\% & \textbf{2.4\%} & {\textbf{3.2\%}} & \textbf{5.3\%} \\
\textit{Ugen-v2}   & 100\%   & 69\% & \textbf{84\%} & {85\%} & \underline{84.8\%} \\
\hline
\textit{Ugen12}  & 100\%  & 22.2\% & \textbf{61\%} & \textbf{61\%} &\textbf{61\%}  \\
\textit{Ugen $l\!=\!2$}   &100\%  &100\%  &\textbf{47.4\%}  & \textbf{52.6\%} & \textbf{47.4\%}\\
\hline
\end{tabular}
}
\end{table}

\begin{figure}[!htbp]
\centering
\captionsetup[subfigure]{aboveskip=0pt, belowskip=-1pt}

\begin{subfigure}[t]{0.485\columnwidth}
  \includegraphics[width=\linewidth]{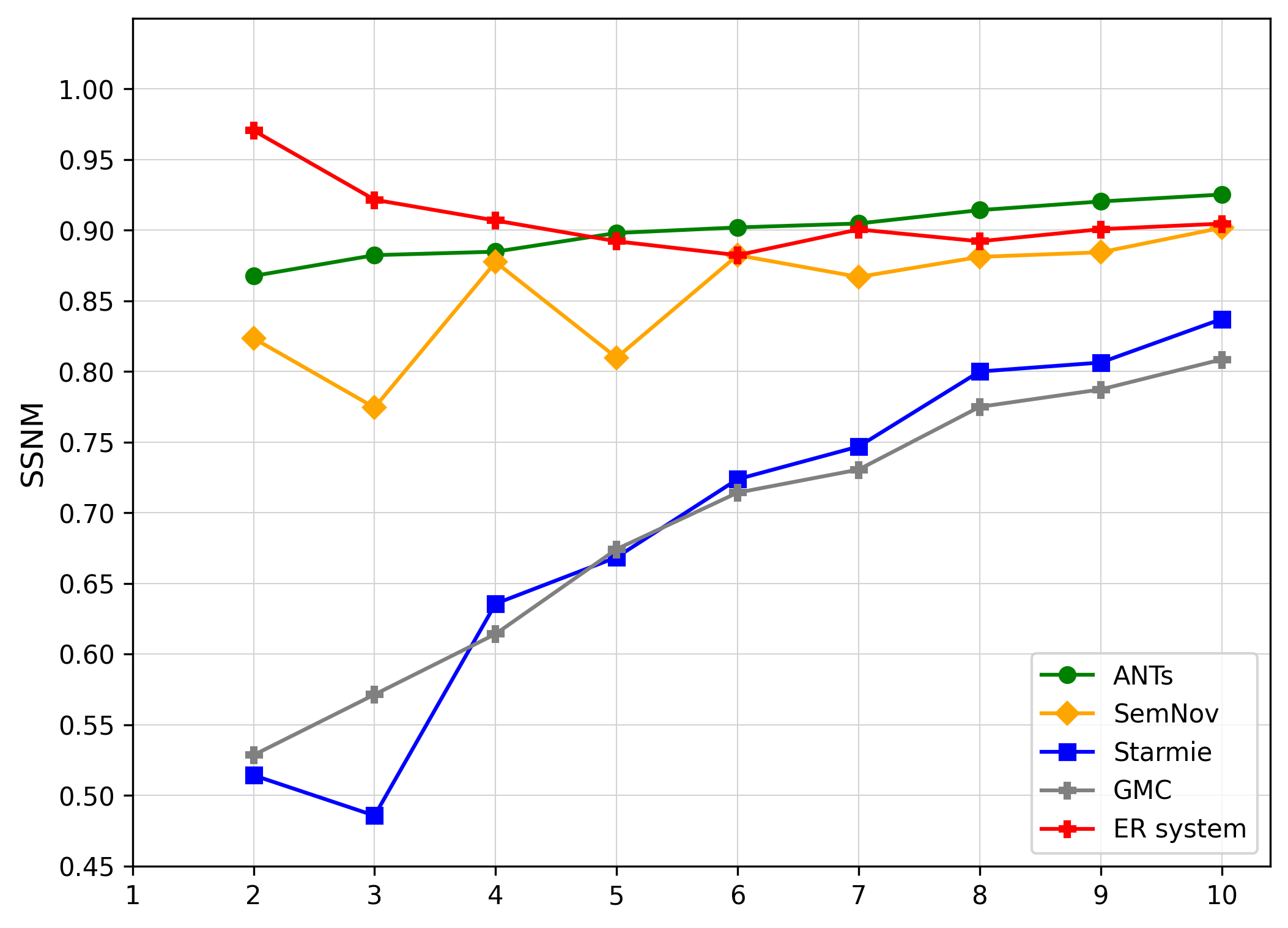}
  \caption{\small Santos}
  \label{fig:ssnm_santos_dustAlign}
\end{subfigure}
\hfill
\begin{subfigure}[t]{0.485\columnwidth}
  \includegraphics[width=\linewidth]{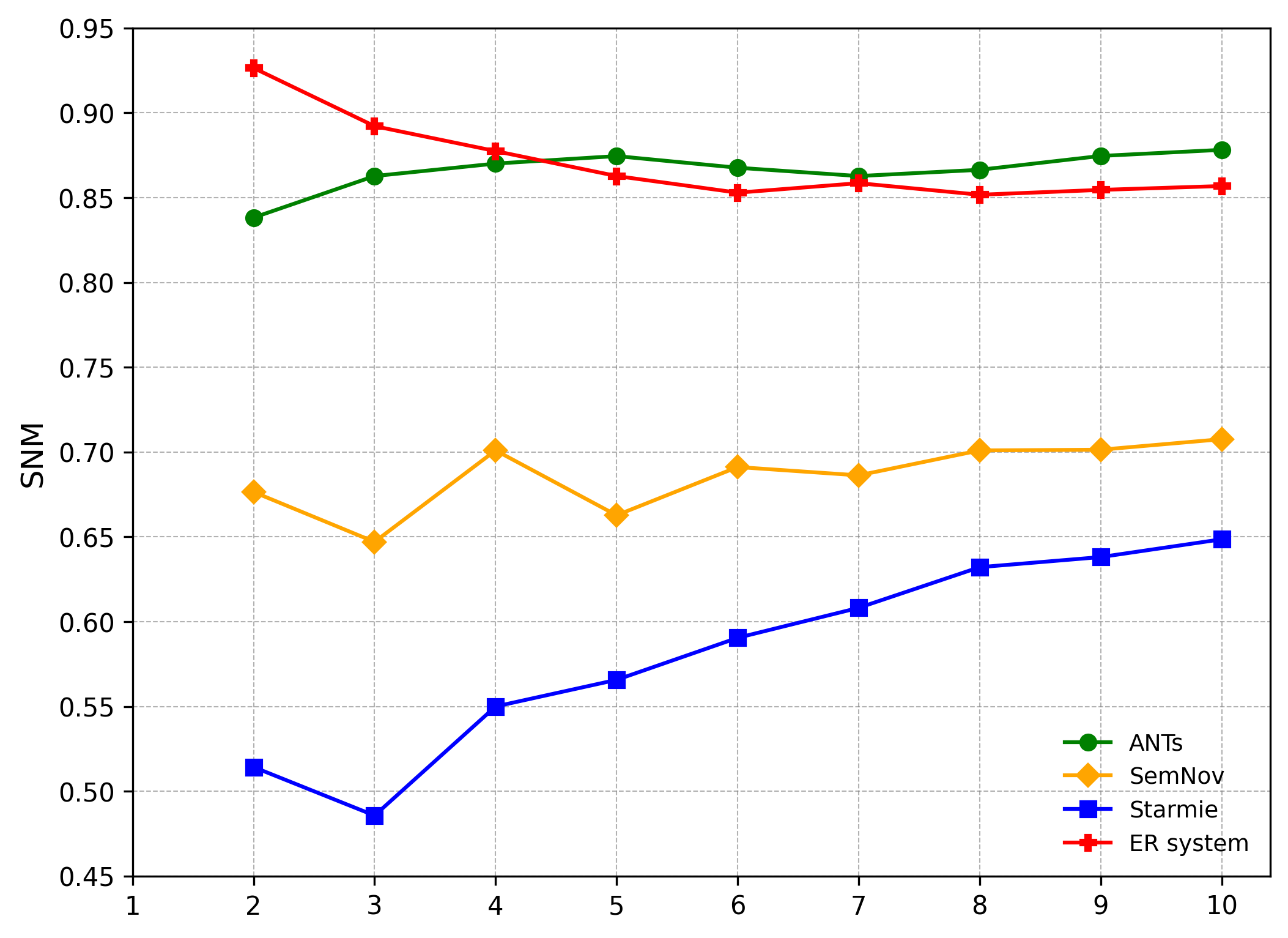}
  \caption{\small Santos}
  \label{fig:snm_santos_dustAlign}
\end{subfigure}

\begin{subfigure}[t]{0.485\columnwidth}
  \includegraphics[width=\linewidth]{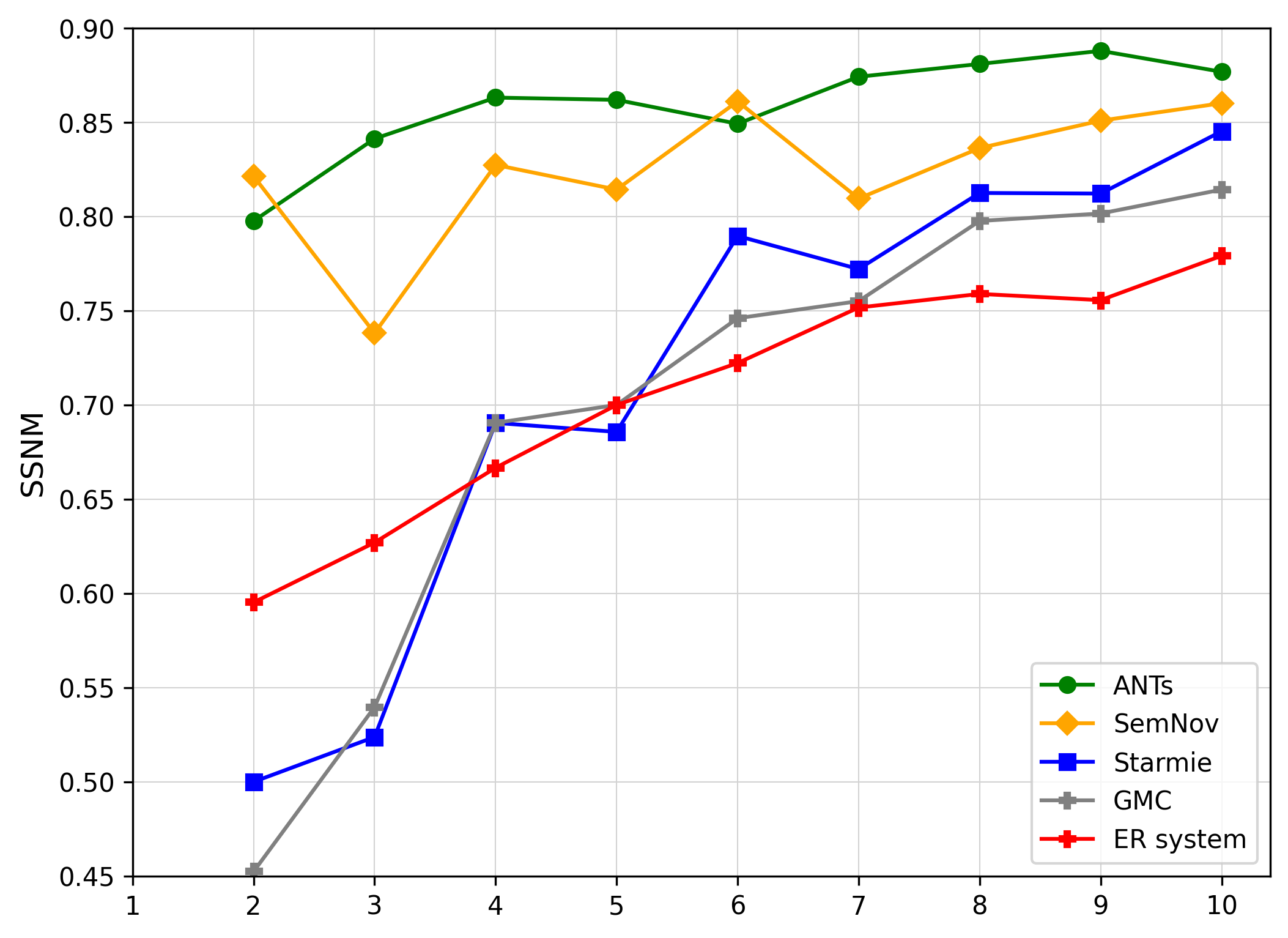}
  \caption{\small TUS}
  \label{fig:ssnm_tus_dustAlign}
\end{subfigure}
\hfill
\begin{subfigure}[t]{0.485\columnwidth}
  \includegraphics[width=\linewidth]{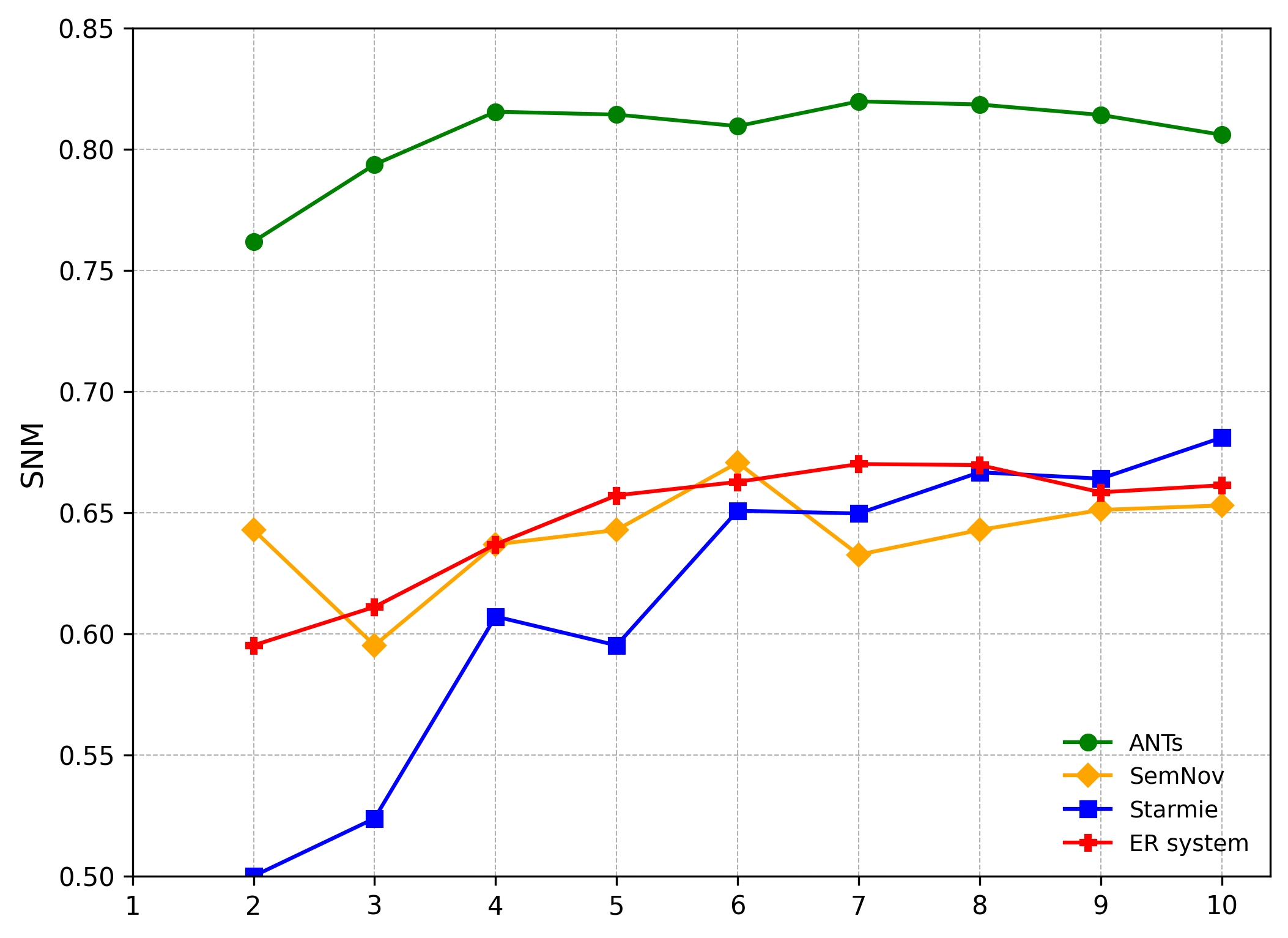}
  \caption{\small TUS}
  \label{fig:snm_tus_dustAlign}
\end{subfigure}

\begin{subfigure}[t]{0.485\columnwidth}
  \includegraphics[width=\linewidth]{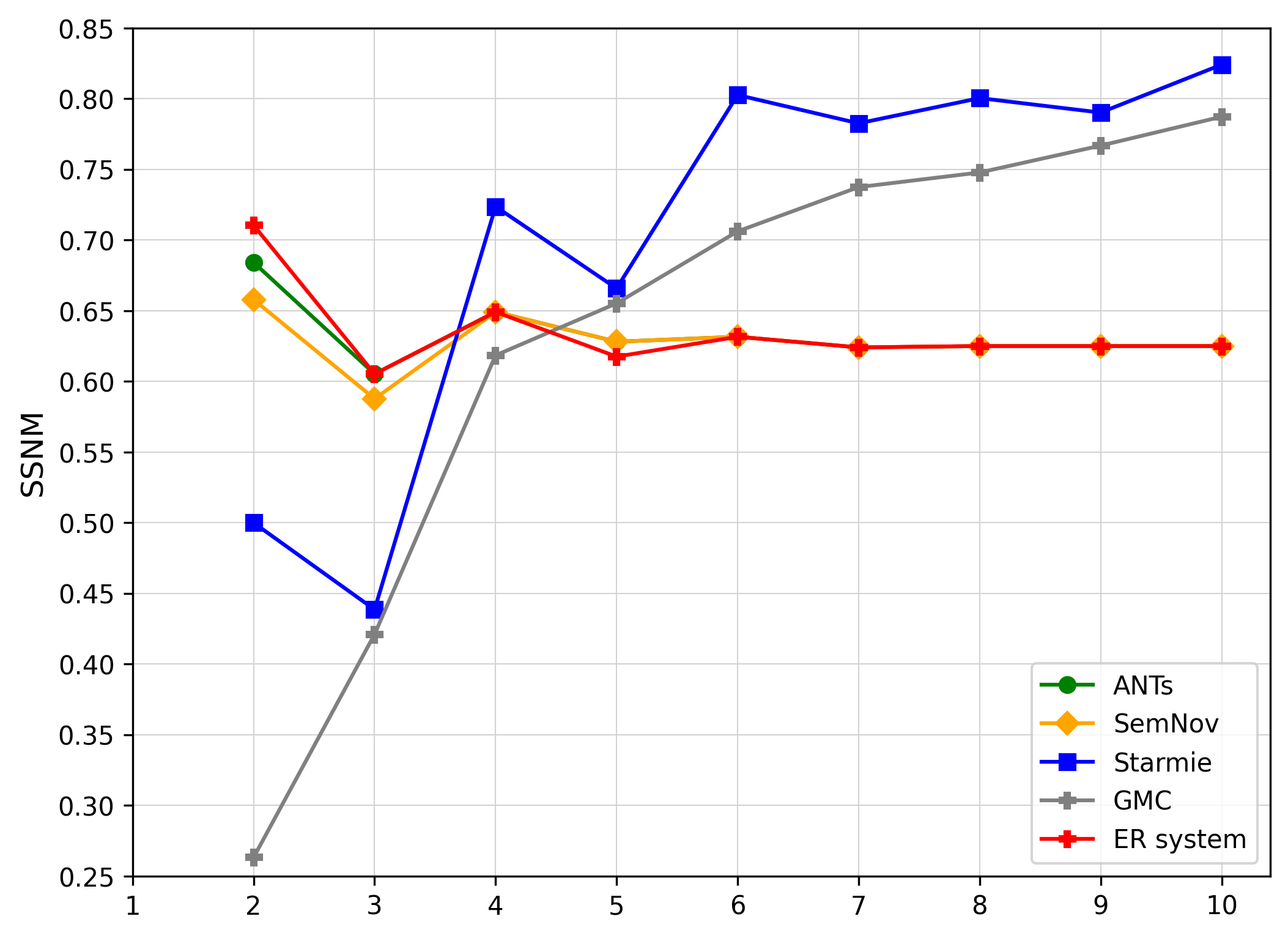} 
  \caption{\small Ugen-v2}
  \label{fig:ssnm_ugen_dustAlign}
\end{subfigure}
\hfill
\begin{subfigure}[t]{0.485\columnwidth}
  \includegraphics[width=\linewidth]{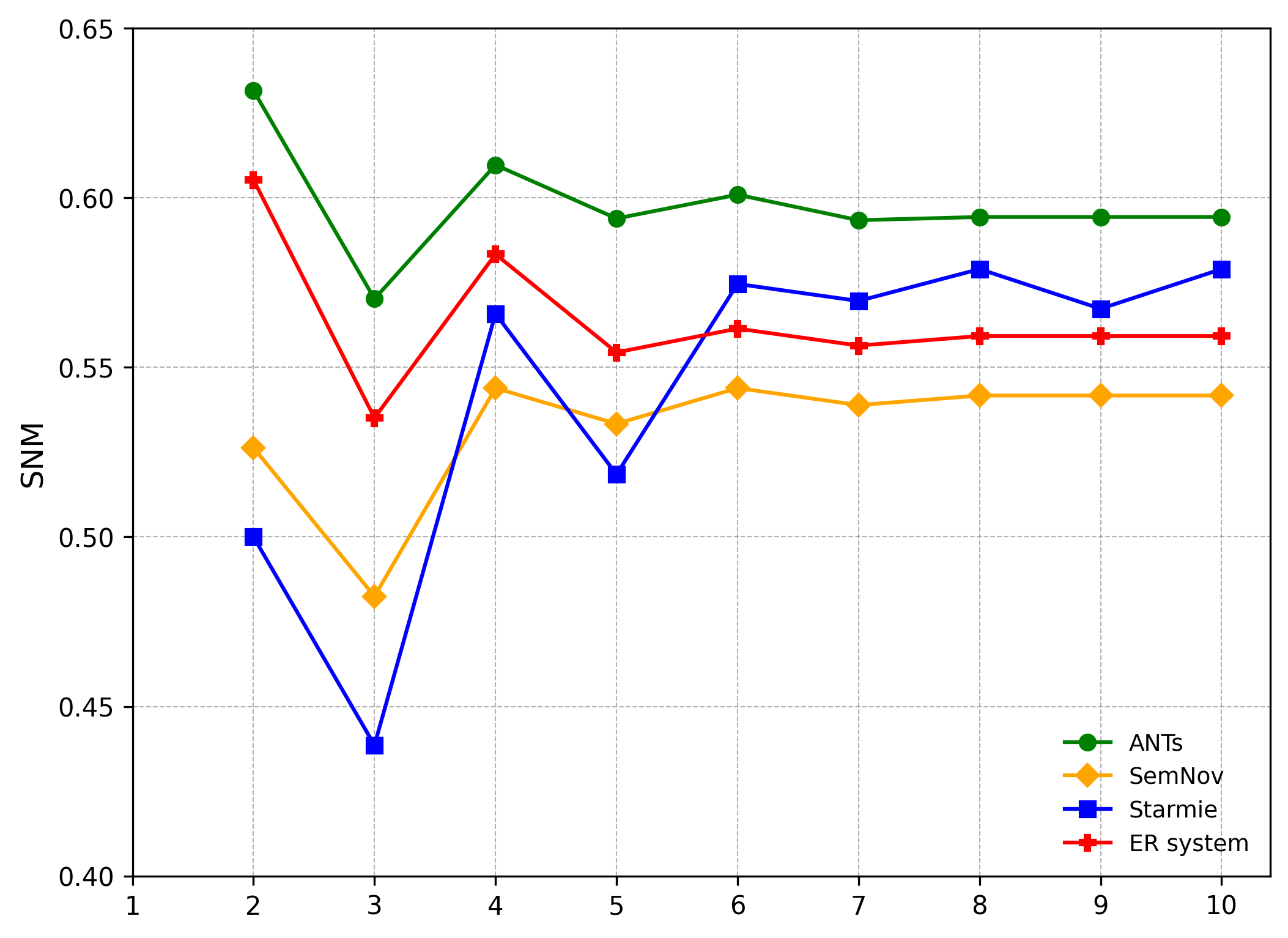} 
  \caption{\small Ugen-v2}
  \label{fig:snm_ugen_dustAlign}
\end{subfigure}

\caption{\small System performance (using the alignment algorithm) across datasets for SSNM (left) and SNM (right), for $l \in [2,10]$. Y-axis truncated.}
\label{fig:ssnm_vs_snm_all_dustAlign}
\end{figure}

\begin{figure}[!t]
  \centering

    \begin{subfigure}[t]{0.48\columnwidth}
    \centering
    \includegraphics[width=\linewidth]{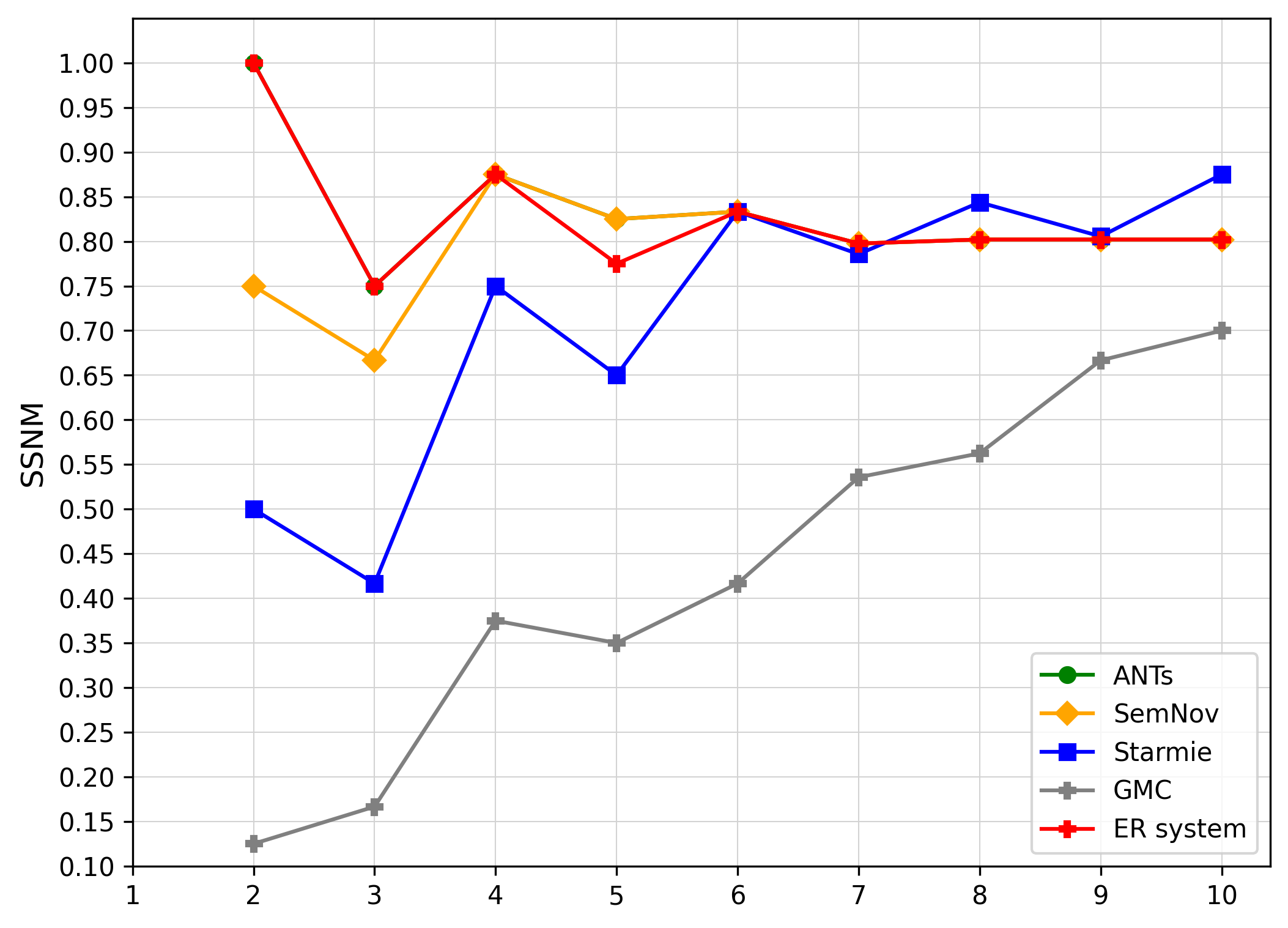}
    \caption{}
    \label{fig:ssnmfig3extendedand_filtered_dust}
  \end{subfigure}
  \hfill
  \begin{subfigure}[t]{0.48\columnwidth}
    \centering
    \includegraphics[width=\linewidth]{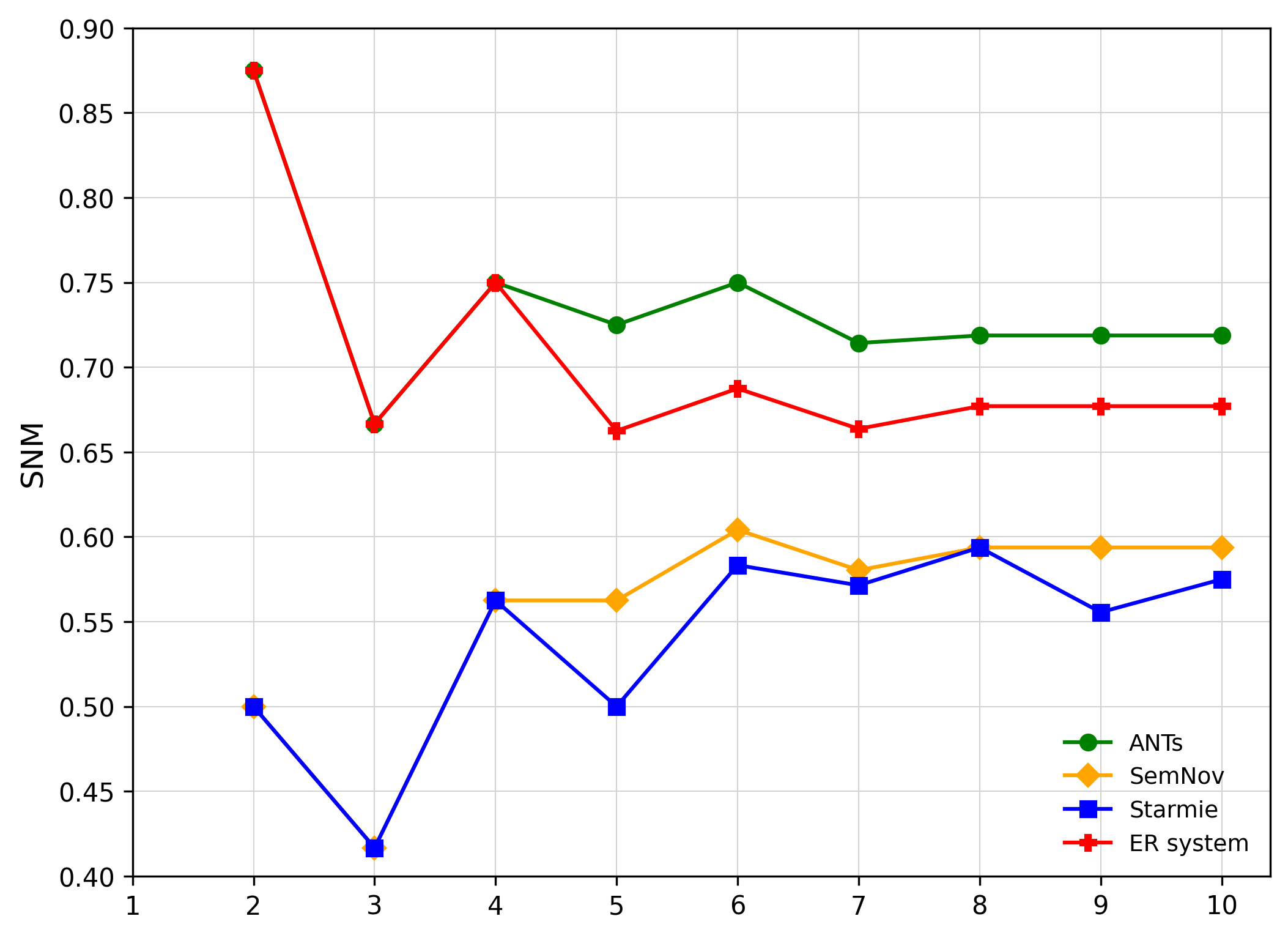}
    \caption{}
    \label{fig:snmfig3_filtered_dust}
  \end{subfigure}

  \caption{\small Results on Ugen-v2 queries with more than 12 unionable tables (Ugen12);   automatic alignment is used. }
  \label{fig:snm_dustalignment_filtered_}
\end{figure}

\begin{figure}[!t]
  \centering
   \begin{subfigure}[t]{0.48\columnwidth}
    \centering
    \includegraphics[width=\linewidth]{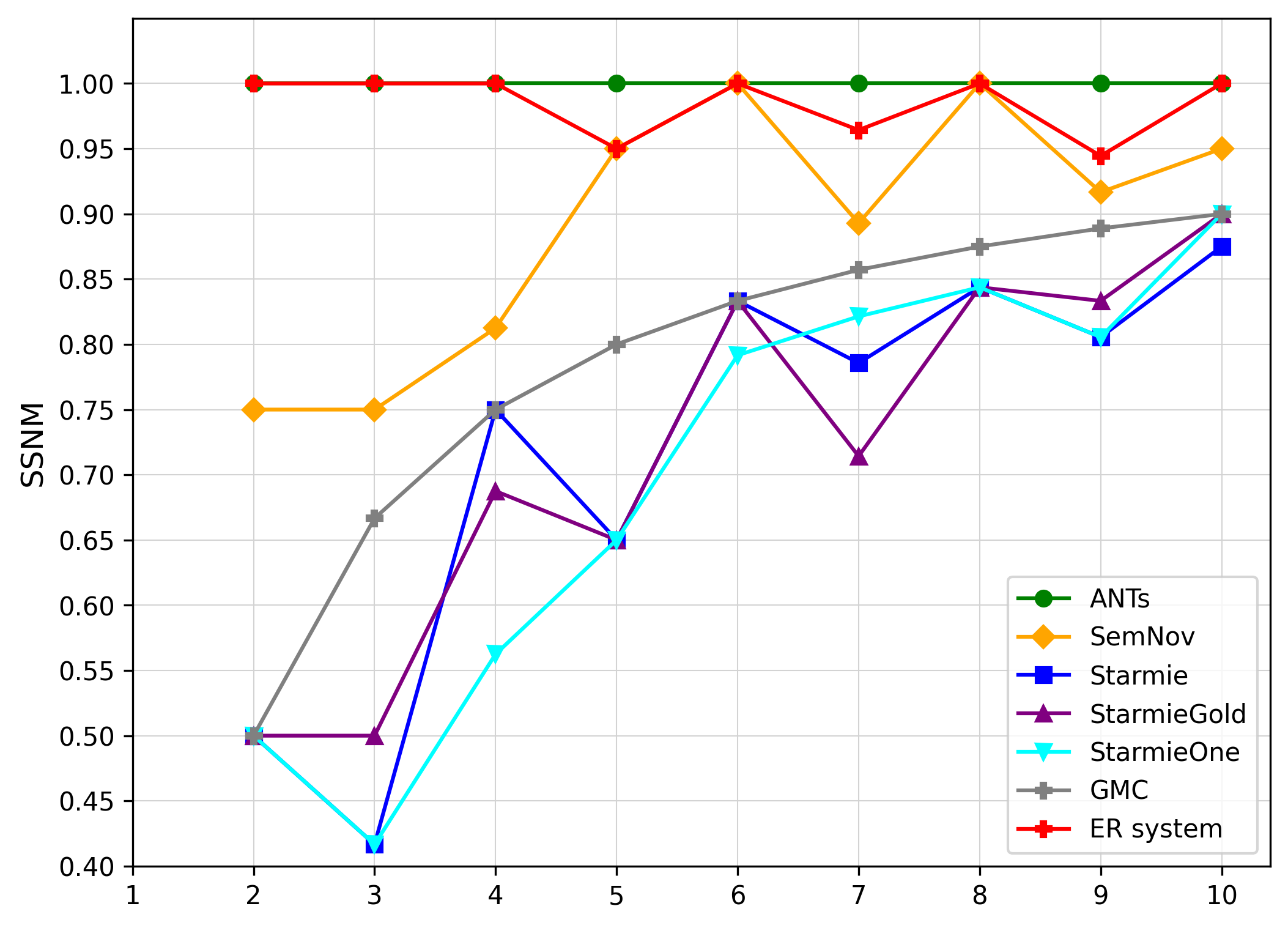}
    \caption{}
    \label{fig:ssnmfig3extendedand_filtered}
  \end{subfigure}
  \hfill
 \begin{subfigure}[t]{0.48\columnwidth}
    \centering
    \includegraphics[width=\linewidth]{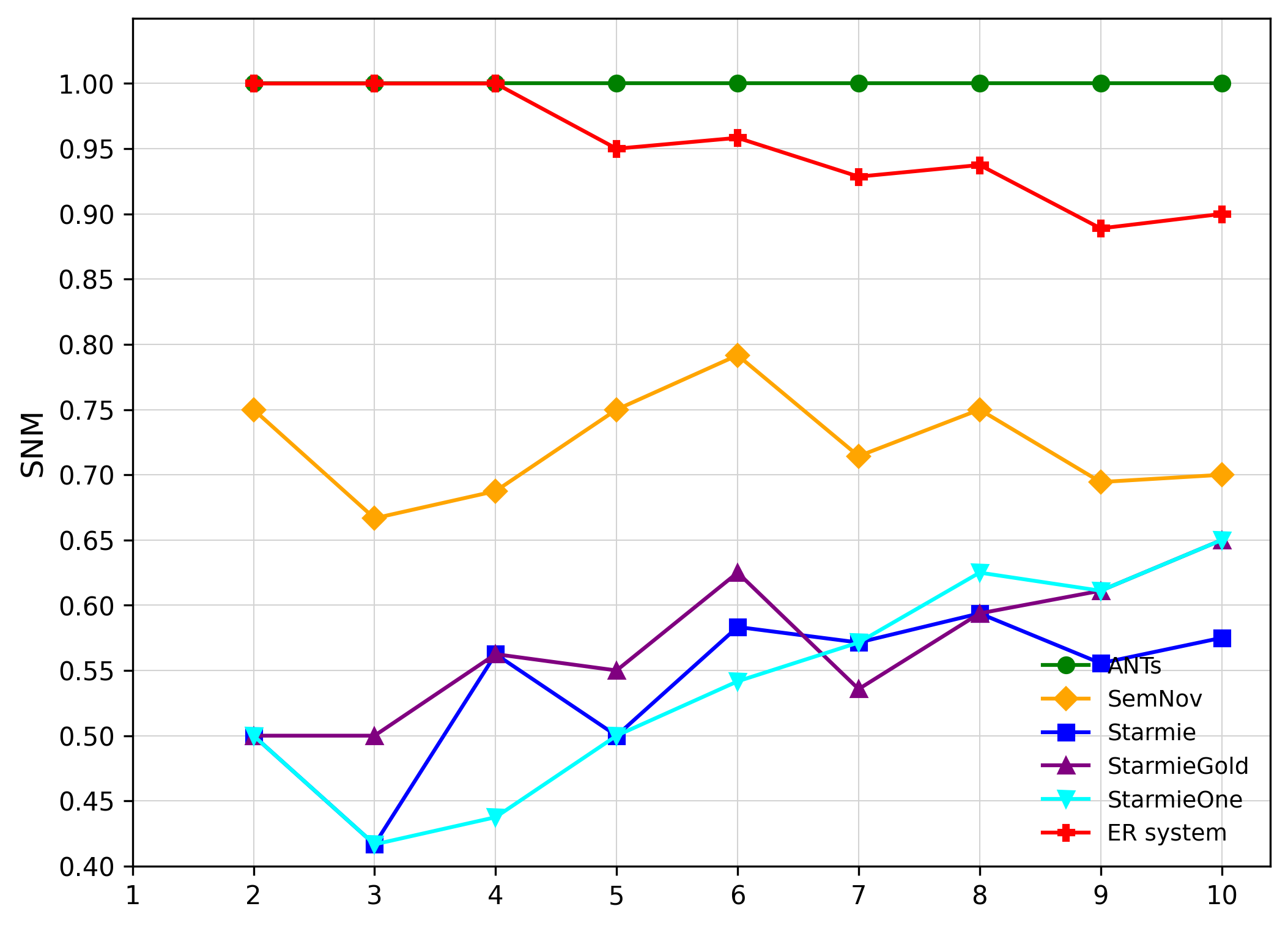}
    \caption{}
    \label{fig:snmfig3_filtered}
  \end{subfigure}

  \caption{\small Results on Ugen-v2 queries with more than 12 unionable tables (Ugen12).}
  \label{fig:snm_filtered_}
\end{figure}

\begin{figure}[!htbp]
\centering

\begin{subfigure}[t]{0.485\columnwidth}
  \includegraphics[width=\linewidth]{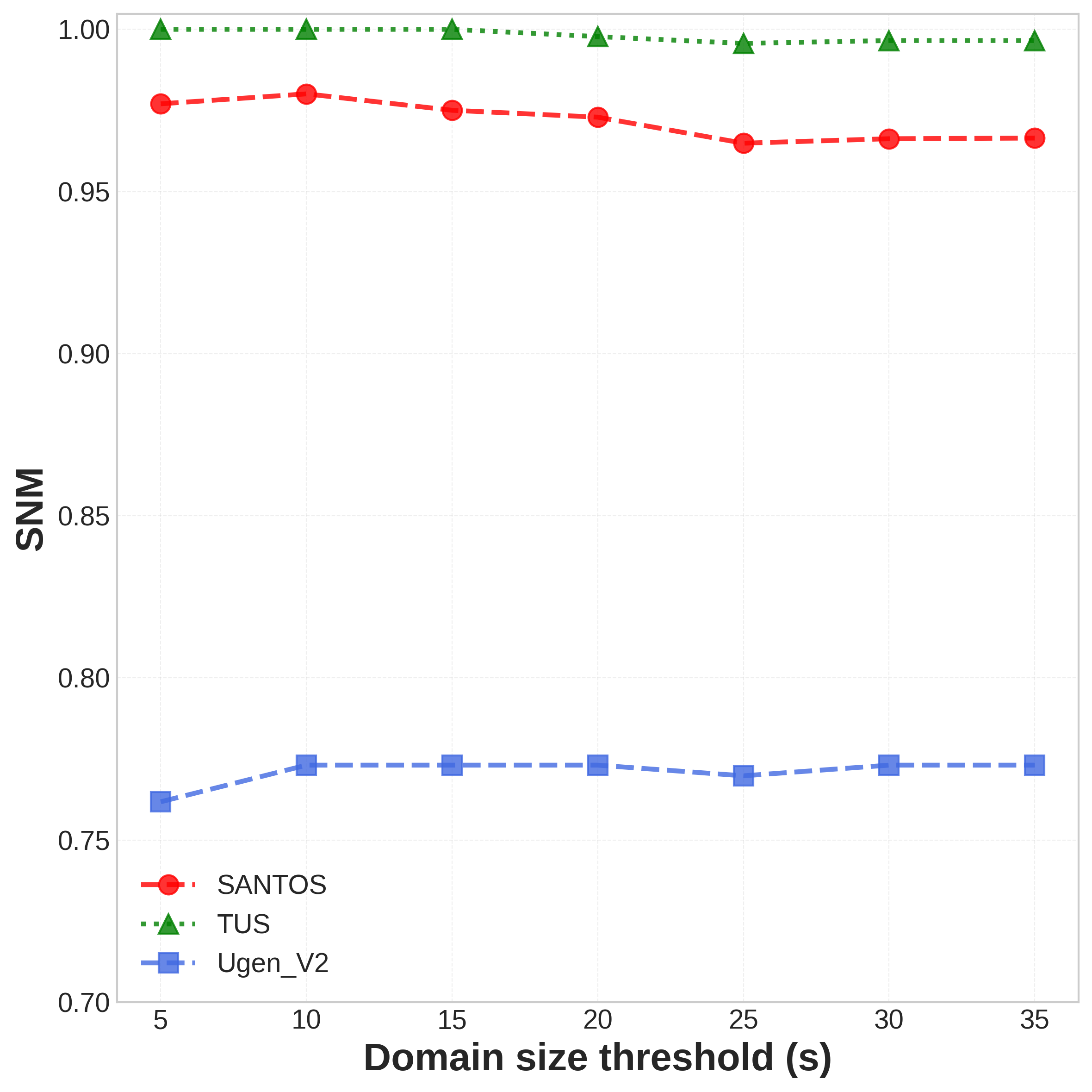}
  \caption{\small }
  \label{fig:snm_s}
\end{subfigure}
\hfill
\begin{subfigure}[t]{0.485\columnwidth}
  \includegraphics[width=\linewidth]{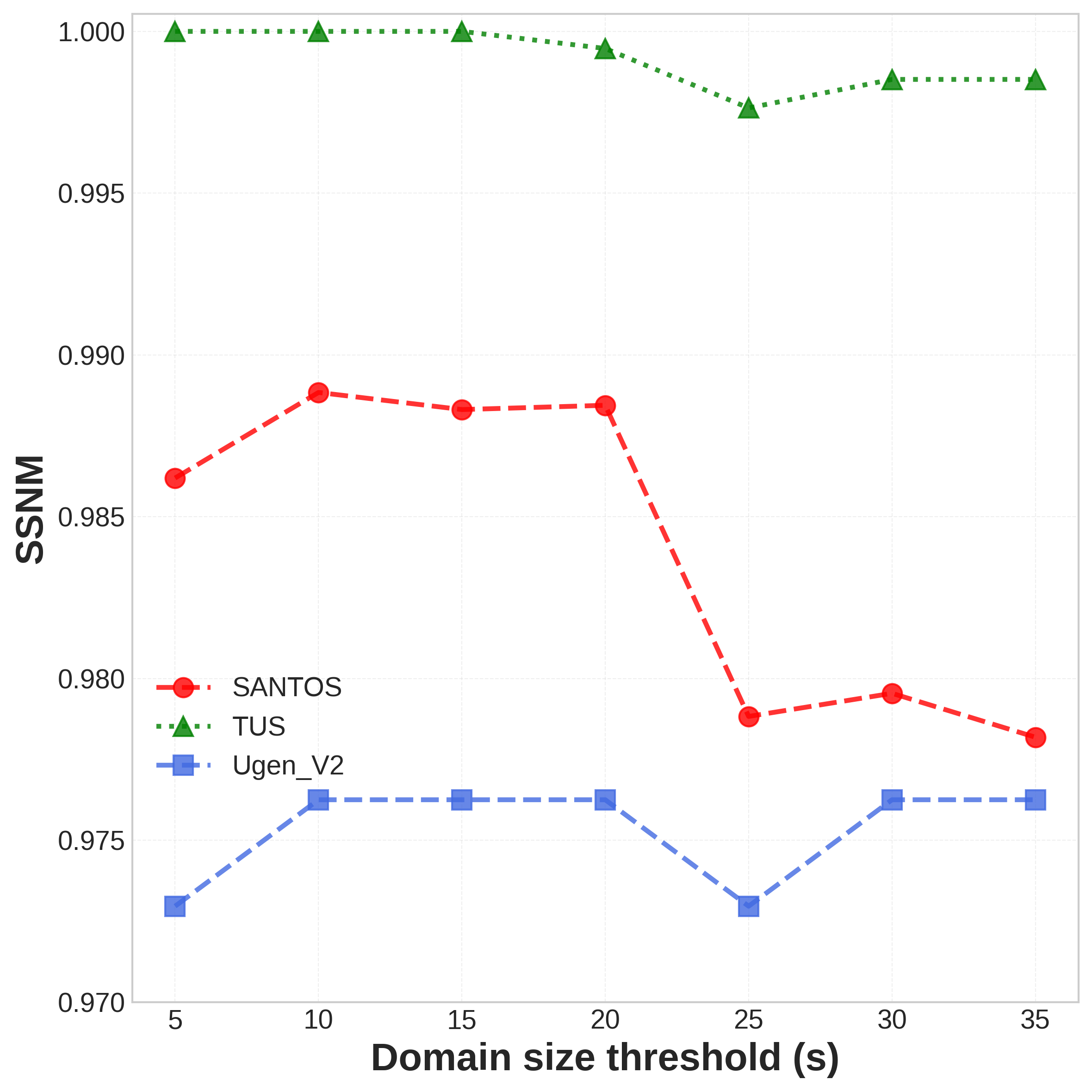}
  \caption{\small }
  \label{fig:ssnm_s}
\end{subfigure}

\begin{subfigure}[t]{0.485\columnwidth}
  \includegraphics[width=\linewidth]{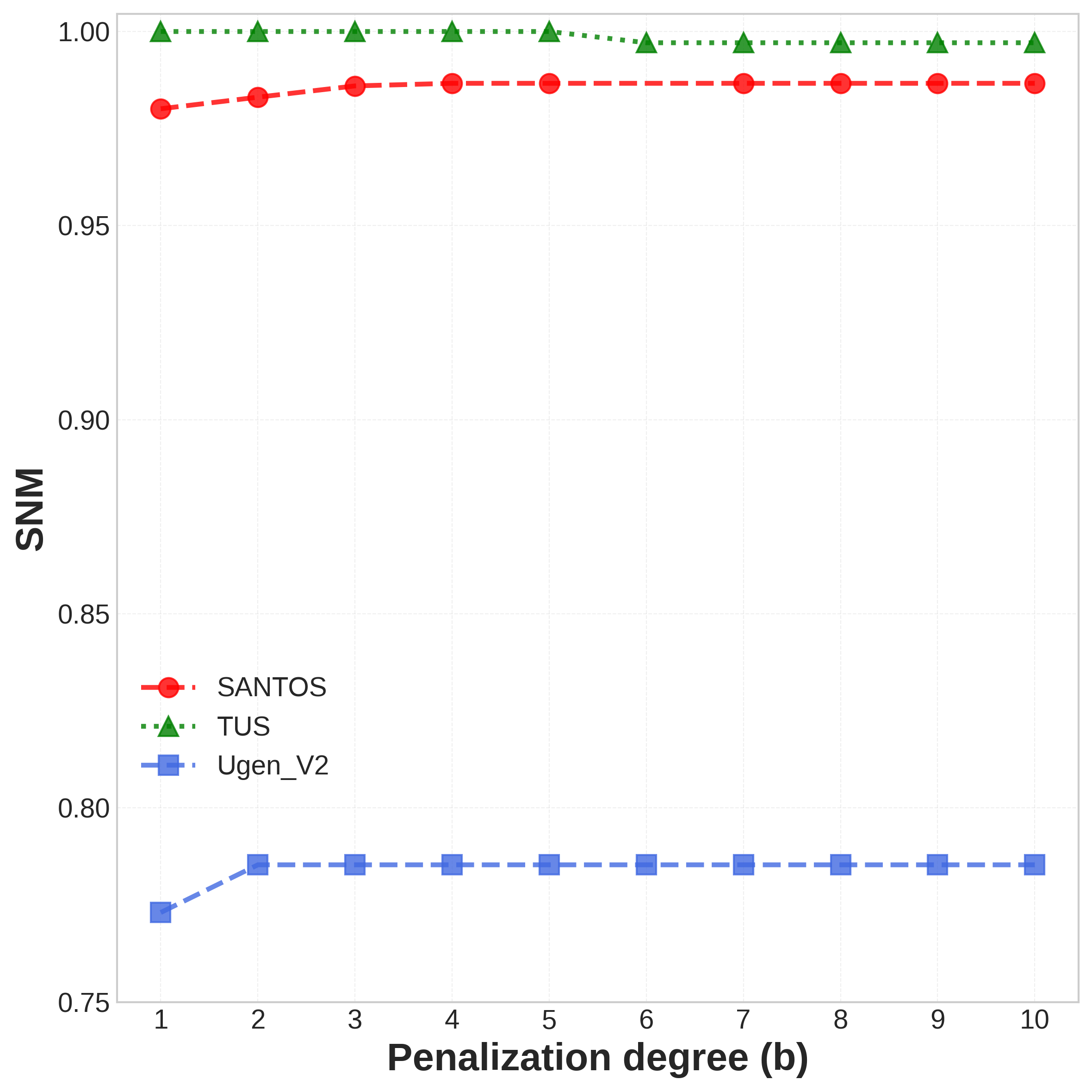}
  \caption{\small }
  \label{fig:snm_b}
\end{subfigure}
\hfill
\begin{subfigure}[t]{0.485\columnwidth}
  \includegraphics[width=\linewidth]{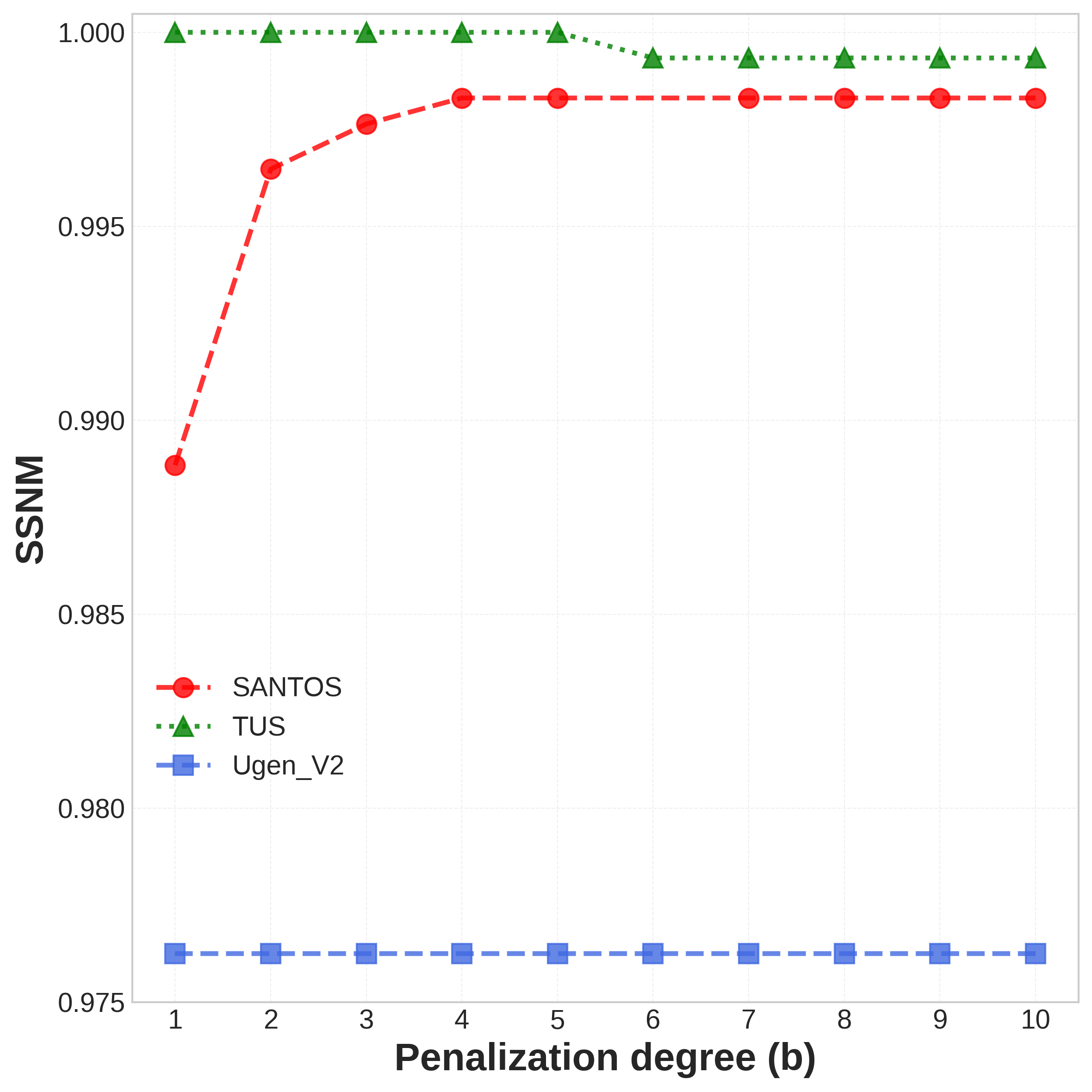}
  \caption{\small }
  \label{fig:ssnm_b}
\end{subfigure}

\caption{\small Effect of different values of $b$ and $s$ on \estimatename’s effectiveness, measured by SNM and SSNM. All values are averaged over all queries and all $l \in [2,10]$.}
\label{fig:ssnm_vs_snm_params}
\end{figure}

\subsection{Ablation Study: Hyperparameter Effects}\label{apnx:ablStudy_hyperparameter} 
There are two key parameters in our method: the domain-size threshold $s$ and the penalization degree $b$. In our main experiments, we set these values empirically based on a small-scale parameter exploration and we use $s = 20$ and $b=1$ for all datasets. In this ablation study, we investigate their impact on the effectiveness of \estimatename. We adopt a two-stage (sequential) hyperparameter-tuning strategy: we first tune $s$ and, conditional on its selected value, we then tune $b$. In the first stage, we fix the penalization degree at $b = 1$ and report SNM and SSNM for $s \in \{5, 10, 15, 20, 25, 30, 35\}$, averaged over all queries and all $l \in [2,10]$ across the three benchmarks. As shown in Figures~\ref{fig:snm_s} and \ref{fig:ssnm_s}, the value of $s$ that yields the best SNM and SSNM is $10$ for Santos, $5$, $10$, or $15$ for TUS, and $10$, $15$, or $20$ for Ugen-v2. This confirms that $s$ should be selected in a data-dependent manner. Next, we fix $s = 10$, its optimal value across all datasets, and report results for penalization degrees $b \in [1,10]$. As shown in Figures~\ref{fig:snm_b} and \ref{fig:ssnm_b}, for TUS both SNM and SSNM attain their maximum for $b \in [1,5]$. On Santos, SNM is relatively stable and less sensitive to $b$ than SSNM; both metrics reach their maximum for $b \in [4,10]$. For Ugen-v2, SNM is maximized for any $b \in [2,10]$, while SSNM is not noticeably affected by $b$.
 
\begin{figure}[!t]
  \centering

    \begin{subfigure}[t]{0.485\columnwidth}
    \centering
    \includegraphics[width=\linewidth]{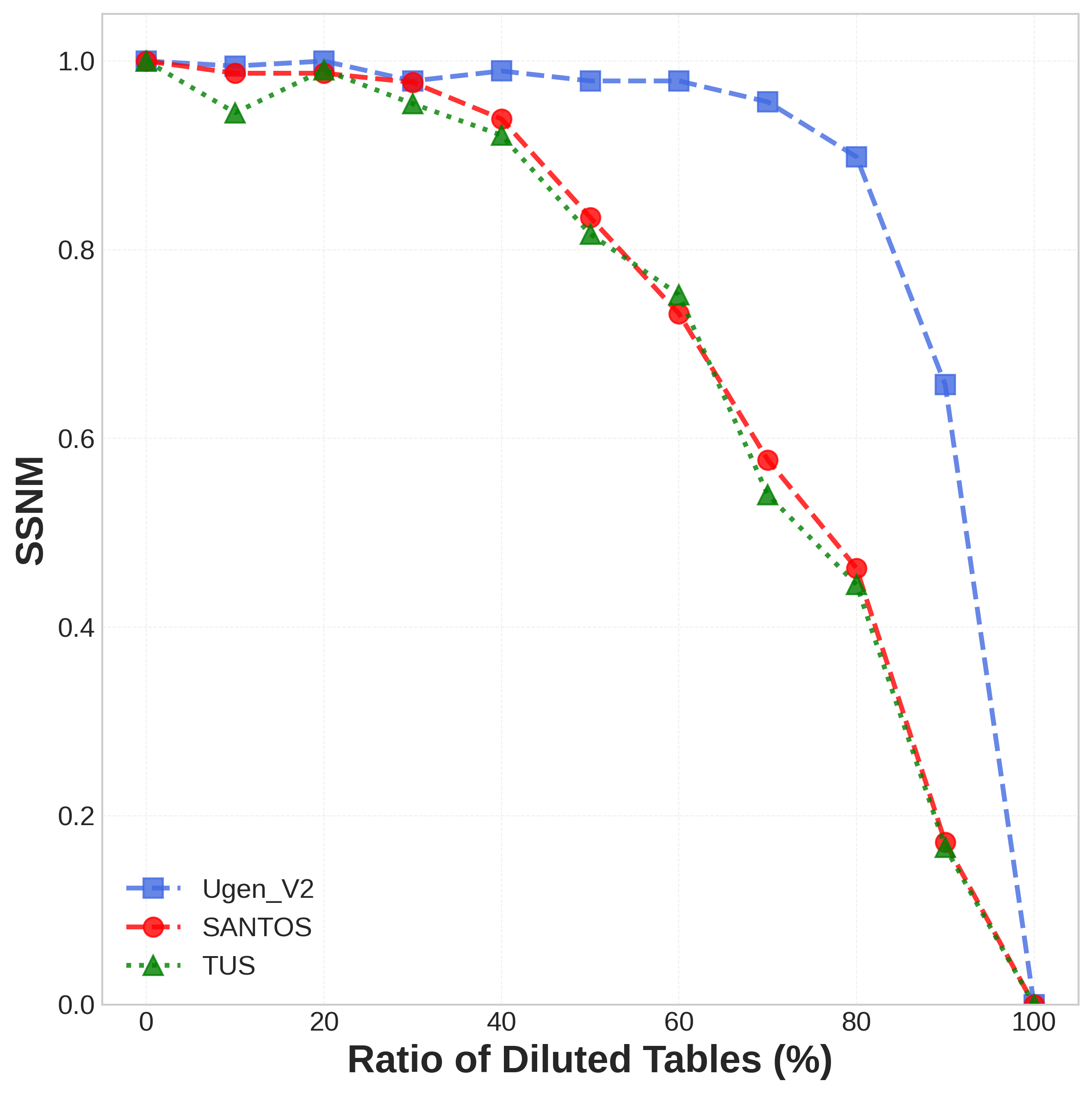}
    \caption{}
    \label{fig:ssnm_InitialSearch}
  \end{subfigure}
  \hfill
  \begin{subfigure}[t]{0.485\columnwidth}
    \centering
    \includegraphics[width=\linewidth]{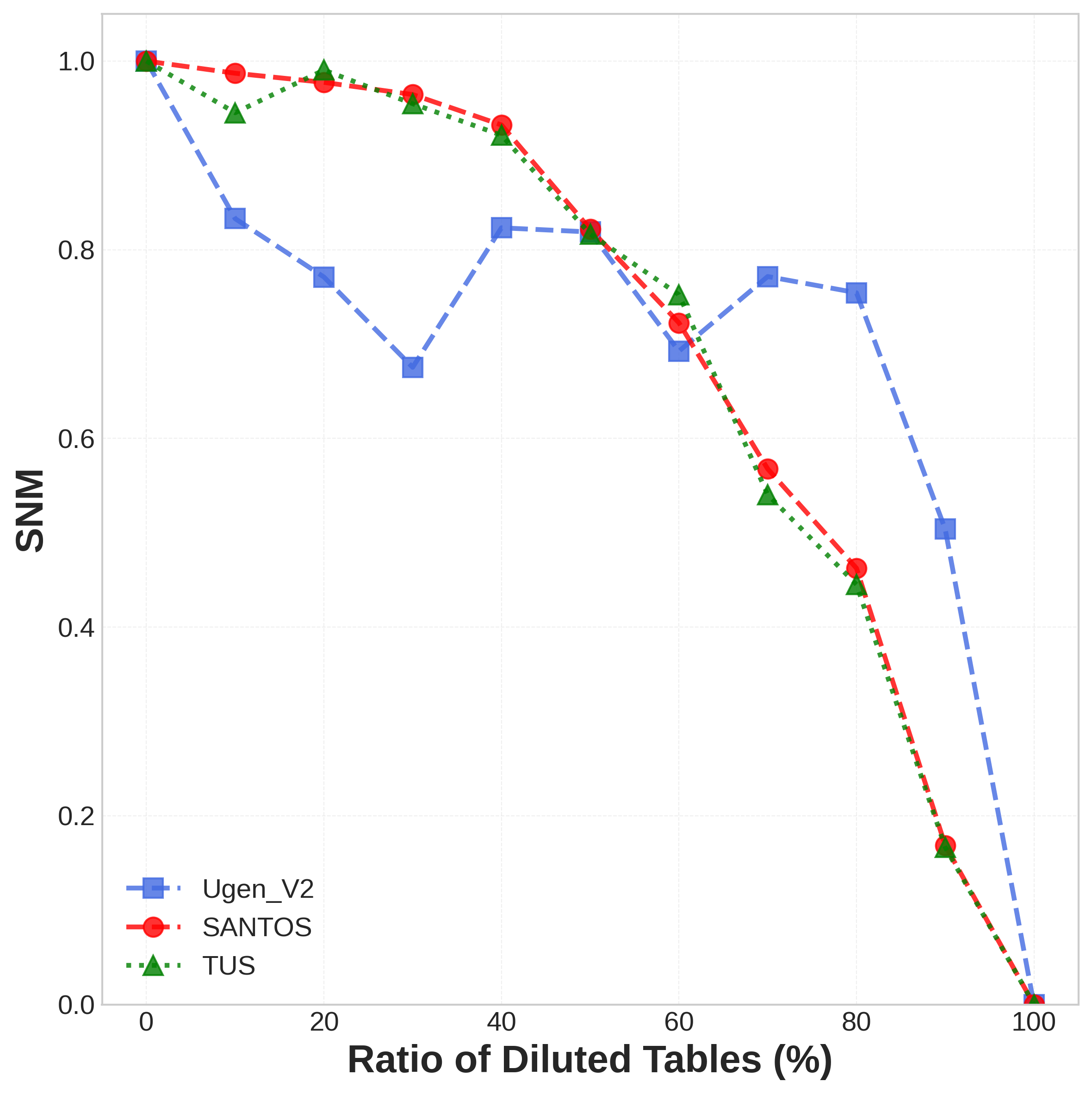}
    \caption{}
    \label{fig:snm_InitialSearch}
  \end{subfigure}

  \caption{\small Impact of initial search result's redundancy on \estimatename’ performance, averaged over all queries in each dataset.}
  \label{fig:InitialSearch}
\end{figure}
\subsection{Ablation Study: Effect of the Initial Search System}\label{subsection:intialresultimpact}
In this section, we examine how the novelty of the initial search results influences the effectiveness of our novelty-aware reranker. 
Specifically, we aim to answer the following question: \textit{How does the ratio of diluted to non-diluted tables, used as a proxy for redundancy in the initial results fed to \estimatename, affect the final SSNM and SNM scores?} To this end, we replace the initial search results with synthetic, randomly generated result sets, where we control the fraction of diluted tables returned by a hypothetical unionable table search system. We then feed these synthetic results to \estimatename and evaluate the novelty of the reranked lists. Ideally, we would report the novelty scores of the reranked result (Definition~\ref{df:nscore_resultset}), but computing them proved prohibitively time-consuming. Instead, we report SSNM and SNM for different values of the diluted tables ratio.  For each dataset we fix the size of the initial random result  to 15 tables and include all queries from that dataset. We then ask \estimatename to rerank these 15 tables and select 10 tables as the final output. As shown in Figure~\ref{fig:ssnm_InitialSearch}, increasing the ratio of diluted tables in the initial search result is clearly detrimental to \estimatename'performance on all datasets in terms of SSNM. As the redundancy of the initial result set increases, \estimatename' effectiveness decreases. This decrease is gradual at first, indicating that reranking with \estimatename is capable of handling moderate levels of redundancy. However, once the fraction of diluted tables exceeds 40\% on Santos and TUS, and approximately 70\% on Ugen-v2, the performance drop becomes much steeper. A similar  pattern holds for SNM on Santos and TUS (Figure~\ref{fig:snm_InitialSearch}). For Ugen-v2, however, the dataset is much sparser: most queries have only a small number of unionable tables and the target ratios cannot always be achieved. This explains the unusual increases observed around the 40\% and 70\% ratio.

}
 
\end{document}